\documentclass[letter,11pt]{article}
\usepackage[utf8]{inputenc}
\usepackage{amsmath}
\usepackage{amsthm}
\usepackage{thmtools, thm-restate} 

\usepackage[margin=1in]{geometry}

\newtheorem{atheorem}{Theorem}[section]

\usepackage[table]{xcolor}
\usepackage{cite}

\usepackage[ruled,vlined,boxed]{algorithm2e}

\usepackage{tcolorbox} 

\usepackage{booktabs}
\usepackage{array}
\definecolor{lightgreen}{RGB}{180,240,180}
\definecolor{lightyellow}{RGB}{250,245,205}
\definecolor{lightred}{RGB}{255,220,230}
 
\usepackage{adjustbox} 
 
\usepackage{url}  

\usepackage{tikz} 
\usetikzlibrary{positioning}
\usepackage{caption} 
\usepackage{subcaption}

\usepackage{graphicx} 
\graphicspath{{./Figs-relations/}}

\usepackage[font=small,labelfont=bf]{caption} 

\newcommand{\omv}{\textnormal{\textsf{OMv}}\xspace}
\newcommand{\oumv}{\textnormal{\textsf{OuMv}}\xspace}
\newcommand{\st}{$s$-$t$\xspace}
\newcommand{\paths}[3]{${#1}$-${#2}$ ${#3}$-paths\xspace}
\newcommand{\er}{Erd\H{o}s-R\'{e}nyi\xspace}
\newcommand{\sol}{\textsc{Sol}\xspace}
\newcommand{\ptg}{\textsc{P3toGeneral}\xspace}

\newcommand{\tp}{t_{\textnormal{pre}}}
\newcommand{\tu}{t_{\textnormal{u}}}
\newcommand{\tq}{t_{\textnormal{q}}}

\usepackage{amsmath,amsthm,amssymb,amsfonts,bbm}
\usepackage{url}
\usepackage{enumerate}
\usepackage{hyperref}
\usepackage[capitalize]{cleveref}

\let\originalleft\left
\let\originalright\right
\renewcommand{\left}{\mathopen{}\mathclose\bgroup\originalleft}
\renewcommand{\right}{\aftergroup\egroup\originalright}

\newcommand{\per}[1]{\left(#1\right)}

\newcommand{\set}[1]{\{#1\}}

\newcommand{\ceil}[1]{\lceil{#1}\rceil}

\newcommand{\R}{\mathbb{R}}
\newcommand{\N}{\mathbb{N}}

\newcommand{\Z}{\mathbb{Z}}
\newcommand{\E}{\mathbb{E}}
\newcommand{\F}{\mathbb{F}}

\newcommand{\calG}{\mathcal{G}}

\newcommand{\eps}{\epsilon}

\newcommand{\card}[1]{\left|{#1}\right|}

\newcommand{\symdif}{\triangle}

\newcommand\Prob[2]{{\Pr_{#1}\left[ {#2} \right]}}

\newcommand\Expc[2]{{\mathop{\E}_{#1}\left[ {#2} \right]}}

\newcommand{\ind}[1]{\mathbbm{1}_{\scriptscriptstyle{#1}}}
\newcommand{\sd}[2]{\mathrm{SD}\left({#1},{#2}\right)}

\newcommand{\poisson}[1]{\mathrm{Pois}\per{{#1}}}
\newcommand{\bernoulli}[1]{\mathrm{Ber}\per{{#1}}}

\newcommand{\equivmod}[1]{{\equiv_{_{#1}}}}
\newcommand{\sameparity}{\ \equivmod{2}\ }

\DeclareMathOperator{\poly}{poly}

\newcommand{\argmin}[1]{\underset{#1}{\mathrm{argmin}}}

\newtheorem{theorem}{Theorem}[section]
\newtheorem{lemma}[theorem]{Lemma}

\newtheorem{corollary}[theorem]{Corollary}
\newtheorem{claim}[theorem]{Claim}
\newtheorem{fact}{Fact}

\newtheorem*{theorem*}{Theorem}
\newtheorem*{lemma*}{Lemma}

\theoremstyle{remark}
\newtheorem*{remark}{Remark}
\newtheorem{observation}{Observation}

\theoremstyle{definition}
\newtheorem{definition}{Definition}
\newtheorem*{definition*}{Definition}

\theoremstyle{conjecture}

\newtheorem*{conjecture*}{Conjecture}

\crefname{atheorem}{Theorem}{Theorems}

\newcommand{\alg}[1]{{\mathcal{#1}}}
\newcommand{\advdist}{\mathcal{D}_\textnormal{adv}}
\newcommand{\advseq}{{\mathcal{S}_\textnormal{adv}}}
\newcommand{\advhist}{{\mathcal{H}_\textnormal{adv}}}
\newcommand{\uadv}{u_\textnormal{adv}}
\newcommand{\vadv}{v_\textnormal{adv}}
\newcommand{\muadv}{\mu_\textnormal{adv}}
\newcommand{\redseq}{\mathcal{S}_\textnormal{red}}
\newcommand{\redhist}{\mathcal{H}_\textnormal{red}}
\newcommand{\udif}{{u_\textnormal{dif}}}
\newcommand{\vdif}{v_\textnormal{dif}}
\newcommand{\roundnumber}{{t}}
\newcommand{\samedist}{\overset{\mathrm{d}}{=}}

\makeatletter 
\let\c@figure\c@table
\let\ftype@figure\ftype@table
\makeatother

\title{Smoothed Analysis of Dynamic Graph Algorithms}
\author{}
 \author{{Uri Meir \thanks{Tel-Aviv University, Tel-Aviv, Israel. Email: \texttt{urimeir.cs@gmail.com}}} \and {Ami Paz \thanks{LISN --- CNRS \& Universit\'e Paris-Saclay. Email: \texttt{ami.paz@lisn.fr}}}}
\date{}

\begin{document}
\maketitle
\thispagestyle{empty} \begin{abstract}
Recent years have seen significant progress in the study of dynamic graph algorithms, and most notably, the introduction of strong lower bound techniques for them (e.g., Henzinger, Krinninger,  Nanongkai and Saranurak, STOC 2015; Larsen and Yu, FOCS 2023). 
As worst-case analysis (adversarial inputs) may lead to the necessity of high running times, a natural question arises: in which cases are high running times really necessary, and in which cases these inputs merely manifest unique pathological cases?

Early attempts to tackle this question were made by
Nikoletseas, Reif, Spirakis and Yung (ICALP 1995) and by
Alberts and Henzinger (Algorithmica 1998), who considered models with very little adversarial control over the inputs, and showed fast algorithms exist for them.
The question was then overlooked for decades, until Henzinger, Lincoln and Saha (SODA 2022) recently addressed uniformly random inputs, and presented algorithms and impossibility results for several subgraph counting problems.

To tackle the above question more thoroughly, we employ \emph{smoothed analysis}, a celebrated framework introduced by Spielman and Teng (J. ACM, 2004). An input is proposed by an adversary but then a noisy version of it is processed by the algorithm instead.
This model of inputs is parameterized by the amount of adversarial control, and fully interpolates between 
worst-case inputs and a uniformly random input.
Doing so, we extend impossibility results for some problems to the smoothed model with only a minor quantitative loss. That is, we show that partially-adversarial inputs suffice to impose high running times for certain problems. 
In contrast, we show that other problems become easy even with the slightest amount of noise. 
In addition, we study the interplay between the adversary and the noise, leading to three natural models of smoothed inputs, for which we show a hierarchy of increasing difficulty stretching between the average-case and the worst-case complexities.
\end{abstract}

\newpage
\thispagestyle{empty} 
\tableofcontents

\thispagestyle{empty} 
\setcounter{page}{0}

\newpage

\section{Introduction}
We study \emph{dynamic graph algorithms}---data structures for handling graphs that undergo changes. 
The input is composed of an initial $n$-node graph, followed by a sequence for edge changes (addition or removal) and queries.
The goal is to maintain enough information on the graph in order to promptly answer a predetermined query (e.g., is the graph connected,
or how many paths of length~$3$ connect two predefined nodes).

Recent years have seen a large body of work regarding lower bounds for dynamic graph algorithms.
One line of work concerns unconditional lower bounds ~\cite{PatrascuT11a,CliffordGL15,LarsenY23,PatrascuD06};
for example, deciding the connectivity of a dynamic graph (i.e.\ whether the graph is connected) requires $\Omega(\log n)$ time per update provided the query time is constant~\cite{PatrascuD06}.
Another celebrated line of work proves  polynomial conditional lower bounds~\cite{HenzingerKNS15,AbboudW14,Patrascu10,DorHZ00,RodittyZ11,KopelowitzPP16};
for example, even the seemingly simple problem of counting paths of length 3 between two predefined nodes $s$ and $t$ (henceforth, \paths{s}{t}{3}) requires $\Omega(n^{1-\epsilon})$ time per update\footnote{%
That is, cannot be solved in $O(n^{1-\epsilon})$ time per update, for any constant $\epsilon>0$}
unless queries take roughly quadratic time~\cite{HenzingerKNS15},
assuming the \omv conjecture (see \cref{sec:intro-technical}).
Similarly to other algorithmic models, it is natural to ask whether known results in worst-case analysis\footnote{%
    In the literature of dynamic graph algorithms, 
    worst-case running time is sometimes contrasted with amortized running time. 
    Here and all throughout, worst-case simply refers to an adversarial input, as opposed to a uniformly random (average-case) input or a smoothed input.
    }
encapsulate inherent impossibilities, or just superficially allow for pathological
yet unrealistic
examples. 

Preceding the advances in lower bound techniques, a couple of early works proposed algorithms specialized for distributional inputs of restricted form, with the number of edges (the graph density) playing a crucial role.
In~\cite{NikoletseasRSY95}, the input sequence is assumed to be entirely stochastic (no adversarial choices) and follow simple rules to ensure 
a predefined density.
In~\cite{AlbertsH98}, an adversary determines ahead of time which steps will add an edge and which will remove one, but the choice of edges to add or remove is made uniformly at random.
Consequently, the graph in each round is distributed uniformly over all graphs of a given density,
a fact that is crucial in the analysis.

Following advances in lower bound techniques for dynamic graph algorithms (and for fine-grained complexity in general), a recent work~\cite{HLS22} revived the average-case analysis of dynamic graph algorithms, offering upper and lower bounds for subgraph counting problems under uniformly random inputs.
For example, it shows that \paths{s}{t}{3} can be counted
with $O(1)$ update and query times on such inputs,
while maintaining the number of such paths of length~$5$
requires either $\Omega\left( n^{1-\epsilon} \right)$ time per update or roughly quadratic time per query.
Note that some classical problems in the field, such as connectivity, are trivial for a uniformly random input: one can blindly answer each query positively, since a random graph is connected w.h.p.

These results point to the fragility of existing lower bounds, but they only apply with very little adversarial control over the input.
In our work, we strive to capture input models in which the adversary has more, but not perfect control.
In particular, we are led by the following question.

\begin{tcolorbox}[colback=blue!10, colframe=blue!30, halign=center]
	{
        How robust are lower bounds for dynamic graph algorithms?
        }%
\end{tcolorbox}

To address this question, we apply \emph{smoothed analysis}, a celebrated theoretical framework introduced by Spielman and Teng~\cite{SpielmanT04,SpielmanT09}.
In this approach, a smoothed input is created by considering an adversarial input
and then perturbing it by adding random noise;
the complexity is computed with respect to the smoothed input (e.g., on expectation).
Smoothed analysis was applied to different types of inputs;
when used on graphs, a random perturbation is applied to an adversarially-chosen graph 
by independently flipping the (in)existence of each edge with a certain probability~\cite{krivelevich2006smoothed,friedrich2011smoothed,krivelevich2015smoothed}.
A more recent line of work applied smoothed analysis to a graph-based computational model closer to ours called dynamic networks, where the computation also occurs in rounds (note that this is a model of distributed computation, while ours is centralized).
There, each round is executed with an entirely new communication graph, and a smoothed input is simply defined by perturbing each such graph independently as above~\cite{DFGN18,MPS20,DFGN22}.

We continue the conceptual line of applying random perturbations.
In dynamic graph algorithms each round consists of a \emph{single} edge change, which is too subtle to invoke noise through a global perturbation as in the aforementioned cases. 
Instead, we apply smoothing per round: with some probability and independently of other rounds, we replace the adversarial edge with a uniformly random edge.
Over many changes, this accumulates to a global perturbation, naturally extending previous graph smoothing models. 
This approach yields diverse input models and results. 

In addition to its theoretical merits, smoothed analysis is often argued to be a realistic model of inputs. This applies for dynamic graphs as well as to any computation model, since inputs are prone to minor perturbations before being processed (e.g., due to communication failures or human errors).
That is, even if (due to chance or malicious intentions) the intended input should impose high running times, it is well-incentivized to analyze the de facto performance of an algorithm assuming a perturbed version of the input is processed instead (i.e., a smoothed input).

\subsection{Smoothed dynamic graphs}

We consider fully dynamic graphs with a fixed set of nodes $V=[n]$ and changing edges.
The input consists of the edges of an initial graph, followed by a sequence of edges to be flipped.
The complexity of a dynamic graph algorithm is measured by the time $\tp$ it needs for processing the initial graph, the time $\tu$ for processing an edge change, and the time $\tq$ for answering a query.

We next define a $p$-smoothed input, for a parameter $0\leq p\leq 1$, which we name the \emph{oblivious flip adversary} model of smoothing.
First, the adversary fixes an initial graph and a sequence of edge flips.
The initial graph is smoothed as follows: each node pair remains consistent with the adversarial choice with probability~$p$, and otherwise re-sampled (uniformly from~$\set{0,1}$).
Then, each edge flip remains unchanged with probability~$p$, and otherwise re-sampled uniformly from~$\binom{[n]}{2}$.
Clearly, when 
$p=0$ the process coincides with an average-case (uniformly random) input, while for $p=1$ it results in the standard adversarial (worst case) input.

\paragraph{Study cases.}
Smoothed analysis can be relevant to any problem presenting a gap between the worst-case and average-case complexities.
We consider two classes of such problems: counting small subgraphs, and
deciding certain graph properties.
Interestingly, the different classes exhibit fundamentally different smoothed complexities as a function of $p$, the fraction of adversarial changes.

Two examples are depicted in \cref{table: preliminary list of results} for constant values of $p$.
The complexity of counting
\paths{s}{t}{3} 
increases quickly with $p$, e.g., for all values $p\in[0.01, 1]$ the complexity is asymptotically identical to the worst-case ($p=1$).
In contrast, the complexity of connectivity remains as in the average case ($p=0$) for any constant $p$;
e.g., it is constant for any $p\in[0,0.99]$.
This already 
exemplifies different behaviors in terms of noise robustness:
the hardness of counting subgraphs hinges on a small set of adversarial edges and therefore persists even with a lot of randomness, while the decision problems we consider concern with global properties, and even a small amount of noise guarantees that a blind positive answer to all queries is correct w.h.p.

\begin{table*}
	\centering
	\scriptsize
	\begin{tabular}{|l|lr|lr|lr|}
		\toprule
		Problem 
		& \multicolumn{2}{l|}{Average case ($p=0$)}
		& \multicolumn{2}{l|}{Smoothed case  ($0<p<1$ const.)} 
		& \multicolumn{2}{l|}{Worst case ($p=1$)}\\ 
		\toprule
		
		\st 3-paths
		&\cellcolor{lightgreen}$O(1)$
		&\cellcolor{lightgreen} \cite{HLS22}
		&\cellcolor{lightred}$\Omega(n^{1-\epsilon})$, $O(n)$
		&\cellcolor{lightred}T.\ref{thm:lb for counting st3 paths}, 
		L.\ref{lem:counting st 3paths}
		&\cellcolor{lightred} $\Omega(n^{1-\epsilon})$, $O(n)$		
		&\cellcolor{lightred}\cite{HenzingerKNS15}, \cite{HanauerHH22}
		\\
		\midrule
		Connectivity
		&\cellcolor{lightgreen}$O(1)$
		&\cellcolor{lightgreen}L.\ref{lem:ub_oblivious_flip}
		&\cellcolor{lightgreen}$O(1)$
		&\cellcolor{lightgreen}L.\ref{lem:ub_oblivious_flip}
		&\cellcolor{lightred}  $\Omega(\log n)$, $\tilde O(\log n)$
		&\cellcolor{lightred} \cite{PatrascuD06},\cite{HuangHKPT23} 
		\\
		\bottomrule
	\end{tabular}
	~
	\caption{
		The update time complexities of \paths{s}{t}{3} and connectivity, with a non-adaptive flip adversary and constant $0<p<1$.
		The update times are for $\tq=O(1)$ and $\tp=O(n^{3-\epsilon})$.
		The \paths{s}{t}{3} lower bounds are conditioned on the \omv conjecture.
		}
	\label{table: preliminary list of results}
\end{table*}

\vspace{15mm} 
\paragraph{A smooth transition.}
A recent work on average-case complexity~\cite{HLS22}
focuses on counting small subgraphs.
Perhaps the strongest message in it is that some problems, such as counting \paths st5, are hard not only in the worst-case but also on average. 
In contrast, counting \paths{s}{t}{2} is easy even in the worst-case (that is, solvable with constant update and query times).
Our focus is the third option presented: for a few problems, including counting \paths{s}{t}{3} and \paths{s}{t}{4}, there is a gap: these problems can be solved quickly in the average case, although they admit strong (conditional) lower bounds in the worst-case.

For constant $p$, we show that  counting \paths{s}{t}{3} and \paths{s}{t}{4} 
are as hard as the worst case and require linear update time. 
More interestingly, for sub-constant $p$, we show that both can be solved with $O(pn)$ update time and constant query time, 
while $\Omega\left(pn^{1-\eps}\right)$ time per update is necessary even if we allow almost quadratic query time.
That is, the complexity linearly depends on $p$, the 
fraction of adversarial changes, as stated in \cref{table: short paths counting}.
In particular, the problem is as easy as the average case for very small values of $p$ (e.g., $p=1/n$), but becomes asymptotically as hard as the worst case already with $p=0.01$, with a gradual transition between the two.

\subsection{Models of smoothing}
\label{sec:intro-models}
Smoothed analysis was originally introduced for studying the simplex algorithm, where the inputs are continuous and the added noise is Gaussian.
The field has greatly evolved since: first, it is now typical to analyze the complexity of a computational problem, rather than that of a specific solution (algorithm), leading to more robust results.
Second, the model of random noise was adjusted to discrete inputs, which are common throughout computer science.
Lastly, the framework has also been applied to several different computational models, each posing its own unique challenges.

In our case, the inputs are discrete and arrive in an online manner. 
This gives rise to several natural models of smoothing, differing in the adaptivity of the adversary to the noise and in the way changes are encoded.
In the following we present three variants of $p$-smoothed inputs:
oblivious flip adversary,
oblivious add/remove adversary, 
and an adaptive adversary.
Each model constitutes a different intermediate model between worst-case and average-case, and the relations between them is depicted in \cref{fig:relations between models}.
In \cref{fig:complexities separating models} 
we exemplify how the complexities of concrete problems increase differently as we transition from easier to harder models of input.

\subsubsection{Oblivious adversaries}
\paragraph{Operation types in worst and average cases.}
A change in a dynamic graph can be seen
as choosing a potential edge and flipping it, or as choosing a potential edge and whether to add or to remove it.
In the typical worst-case analysis, both views coincide:
an algorithm can always check the edge's status in $O(1)$ time and do nothing if it remained unchanged, so an adversary will never choose to leave an edge unchanged.
Hence, worst-case complexity does not depend on the operation type.
Prior work concerning average-case complexity does not address the relation between flip and add/remove operations.
In \cref{sec:adver} we show that up to constant factors, the average-case complexity of any problem is also independent of the operations' encoding.

\paragraph{Operation types in smoothed analysis.}
We have already described the oblivious flip adversary, which chooses a sequence of \emph{edges to flip}, and then some of these edges are re-sampled uniformly.
The \emph{oblivious add/remove adversary} is similarly defined, but instead of edges to flip, it chooses \emph{edges to add or remove} before some of its operations are replaced by random ones.

A priori, one might suspect that an oblivious add/remove adversary is stronger than an oblivious flip one, as the former can, e.g., fixate on certain edges and persistently remove them.
But one might also suspect the opposite, as an oblivious flip adversary is guaranteed to make a change at each round, 
while an add/remove adversary might ``waste'' rounds by trying to add an existing edge or remove a missing one.
As it turns out, the first intuition is true: if $p$ is high enough, an add/remove adversary can simulate a hard case instance on a small subgraph,
which we use in order to prove lower bounds for it. 
Furthermore, using a proxy model with lazy flips we show that the oblivious add/remove adversary is always at least as strong as the flip one: the $p$-smoothed complexity in the oblivious flip model is not higher than the $p'$-smoothed complexity in the oblivious add/remove model, for some~$p' \in [p/2,p]$.

\begin{table*}
	\centering
	\scriptsize
	\begin{tabular}{|l|lr|lr|lr|}
		\toprule
		Problem 
		& \multicolumn{2}{l|}{Average case} 
		& \multicolumn{2}{l|}{Smoothed case} 
		& \multicolumn{2}{l|}{Worst case}\\ 
		\toprule
		
		\st 2-paths
		& \multicolumn{5}{l}{\cellcolor{lightgreen}$O(1)$		}
		&\cellcolor{lightgreen}\cite{HLS22}
		\\
		\midrule
		\st 3-paths
		&\cellcolor{lightgreen}$O(1)$
		&\cellcolor{lightgreen} \cite{HLS22}
		&\cellcolor{lightyellow}$\Omega(pn^{1-\epsilon})$, $O(pn)$
		&\cellcolor{lightyellow}T.\ref{thm:lb for counting st3 paths}, 
		L.\ref{lem:counting st 3paths}
		&\cellcolor{lightred}$\Omega(n^{1-\epsilon})$, $O(n)$		&\cellcolor{lightred}\cite{HenzingerKNS15}, \cite{HanauerHH22}
		\\
		\midrule
		\st 4-paths
		&\cellcolor{lightgreen}$O(1)$
		&\cellcolor{lightgreen} \cite{HLS22}
		&\cellcolor{lightyellow}$\Omega(pn^{1-\epsilon})$, $O(pn)$
		&\cellcolor{lightyellow}T.\ref{thm:lb for counting st4 paths s3 and s4 cycles}, 
		L.\ref{lem:counting st 4paths}
		&\cellcolor{lightred}$\Omega(n^{1-\epsilon})$, $O(n)$
		&\cellcolor{lightred}\cite{HenzingerKNS15}, C.\ref{cor:counting st 4paths worst-case}
		\\
		\midrule
		\st 5-paths
		&\cellcolor{lightred}$\Omega(n^{1-\epsilon})$
		& 
		\multicolumn{5}{r|}{\cellcolor{lightred}\cite{HLS22}}
		\\
		\bottomrule
	\end{tabular}
	~
	\caption{%
		The update complexities $\tu$ of short paths counting,
		for $\tq=O(1)$ and 
		$\tp=O(n^{3-\epsilon})$, as a function of $p$.
		For these problems the smoothed complexities with all adversaries coincide.
		The lower bounds are conditioned on the \omv conjecture.}
	\label{table: short paths counting}
\end{table*}

\subsubsection{Adaptive adversary}
Another issue that emerges when defining smoothing models is \emph{adaptivity to the noise}.
Indeed, in a process that involves both adversarial and random choices, a present (``online'') adversary 
aware of the random noise applied 
might disrupt the algorithm more than an oblivious one.
Similar distinctions were recently made in smoothed analysis of dynamic networks~\cite{MPS20} and online learning~\cite{HRS21}. 

Unlike the oblivious adversaries, an \emph{adaptive adversary} can choose its edge changes in an online manner --- it chooses an edge to flip knowing all previous changes (adversarial or randomly-chosen).
We show that this is sometimes enough in order to cause much slower running times, separating the oblivious and the adaptive smoothed models.
As in the worst case, an adaptive adversary never asks to add an existing edge or to remove a non-existing one, thus defining it with flip operations or with add/remove ones makes no difference.

\paragraph{Queries.} Queries are part of the input, and are sometimes even crucial for devising worst-case lower bounds, especially for problems where queries are parametrized.
This is the case, e.g., in the lower bound
for $(u,v)$-connectivity~\cite{PatrascuD06},
where queries are made with parameters $u,v$ and ask whether the nodes $u$ and $v$ are currently in the same connected component.
Hence, parameterized queries could also be studied in a smoothed manner (and in the average-case), yet the problems studied here have no parameterized queries and the issue is left for future work.

\subsection{Main results}
\label{sec:intro-results}    
\cref{table: results comparison} compares our results to known results in average-case and worst-case models, focusing on  
update times of algorithms with constant query times.
We discuss these results below.

\begin{figure}
	\centering
	\includegraphics[scale=.8]{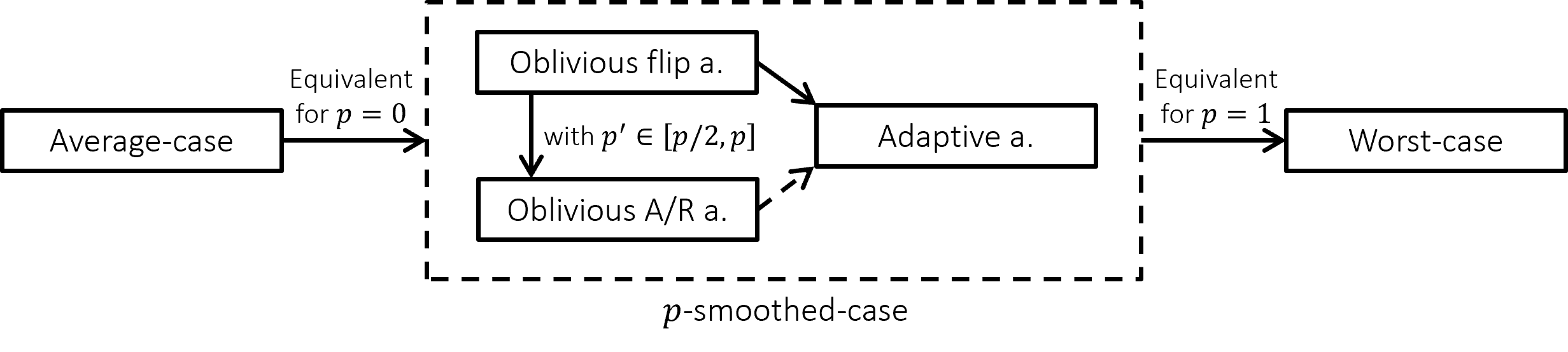}
	\caption{Relations between models (\cref{sec:intro-models} and \cref{sec:hierarchy-separations}).
		An arrow from model~A to model~B represents that model B is at least as hard as model A (that is, the complexity of a problem can only increase).
		A solid arrow represents that this relation is strict, i.e., for some problem and range of $p$, 
		the asymptotic complexity increases from A to B.} 
	\label{fig:relations between models}
\end{figure}

\subsubsection{Decision problems}
We consider several decision problems, such as
connectivity and bipartite perfect matching (deciding whether a bipartite graph contains a perfect matching).
These are trivial to solve on a random static graph
since they hold w.h.p.
In \cref{sec:upper-easy-with-flip} we point out that these hold w.h.p.\ also at each round of an average-case dynamic graph, and are thus also easy to decide in this setting.
We extend this argument to $p$-smoothed dynamic graphs with an oblivious flip adversary, 
show it holds even with very little randomness ($p$ being almost $1$),
and thus reach a trivial constant-time algorithm in this setting as well (\cref{lem:ub_oblivious_flip}).
Specifically, we show that the random perturbation of the initial graph  guarantees the first (polynomially many) graphs in the sequence hold such properties with a polynomially small error, allowing for a trivial ``yes'' answer.
To deal with longer sequences of changes, we argue even more randomness affects the later rounds, and each graph now holds the desired properties with only an \emph{exponentially small} error.
Despite the simple intuition, the formal argument here is not completely straightforward, as the graph at each round is in fact statistically far from a uniformly random graph (e.g.,  the parity of the number of edges distinguishes the two).

\begin{atheorem}[informal]
	\label{athm:decision problems-alg}
    For any $p\in[0,1-\frac{26\log n}{n})$, \textbf{connectivity} and \textbf{bipartite perfect matching} can be decided in constant update and query times under the oblivious flip adversary.
\end{atheorem}

This theorem might raise the question whether the oblivious flip adversary has any power at all, or does it coincide with the average-case complexity.
We later separate the models by showing they induce different complexities for subgraph counting problems.

While the trivial ``yes'' algorithm works well for  bipartite perfect matching with an oblivious flip adversary,
it fails for the stronger adversaries --- oblivious add/remove adversary, and adaptive adversary.
To see this, consider an adversary that repeatedly removes all the edges touching a fixed node;
if $p$ is relatively large, this adversary will manage to isolate the node at least in some of the rounds, rendering the trivial algorithm wrong.

We formalize this intuition by proving lower bounds using a new embedding technique, which might also be applicable in other settings of smoothing of graphs.
The adversary focuses on a small set of $\hat n = \hat n(n,p)$ nodes and controls the induced subgraph on these nodes, as well as its connections with the rest of the graph.
It embeds a worst-case lower bound graph on these $\hat n$ nodes,
such that the decision on the entire graph hinges on the $\hat n$-node subgraph.

An adaptive adversary can embed a worst-case construction for bipartite perfect matching on $\hat n =pn/20$ nodes: 
between every two designated queries, it changes the internal structure of the subgraph to be as in an $\hat n$-node worst case, and in parallel ``fixes'' any random edge change that touches these nodes.
This results in \cref{thm:lb for many problems-smoothed advers}, giving a polynomial lower bound on the update time of bipartite perfect matching.

Surprisingly, a similar strategy works for the oblivious add/remove adversary,
even though it is not aware of the random changes.
This requires larger values of $p$ and assures a control of a smaller subgraph, yet \cref{thm:lb for many problems-smoothed oblivious ar} provides a polynomial lower bound on the update time for this model.

\begin{table*}
\setlength{\arrayrulewidth}{1pt} 
  \hspace*{-0.03\linewidth}
\begin{adjustbox}{max width=1.05\linewidth}
	  \begin{tabular}{|l|lr|lr|lr|lr|lr|}
		\toprule
		Problem 
		& \multicolumn{2}{l|}{Average case} 
		& \multicolumn{2}{l|}{Oblivious flip a.} 
		& \multicolumn{2}{l|}{Oblivious AR a.} 
		& \multicolumn{2}{l|}{Adaptive adv.} 
		& \multicolumn{2}{l|}{Worst case model}\\ 
		\toprule
		
  \multicolumn{11}{|l|}%
  {\textbf{Small subgraph counting}}\\  
		\midrule
  \begin{tabular}{l@{}}
        \st 3-paths
        \end{tabular}
        &$O(1)$
        & \cite{HLS22}
		& 
  \multicolumn{4}{l}{
  $\Omega(pn^{1-\epsilon})$, $O(pn)$}
        &
  \multicolumn{2}{r|}{
		T.\ref{thm:lb for counting st3 paths}, 
		L.\ref{lem:counting st 3paths}}
		& 
		$\Omega(n^{1-\epsilon})$, $O(n)$		&\cite{HenzingerKNS15}, \cite{HanauerHH22}
		\\
		\midrule
        \begin{tabular}{l@{}}
        \st 4-paths
        \end{tabular}
        &$O(1)$
        & \cite{HLS22}
		& 
  \multicolumn{4}{l}{
  {$\Omega(pn^{1-\epsilon})$, $O(pn)$}}
        &
  \multicolumn{2}{r|}{
		T.\ref{thm:lb for counting st4 paths s3 and s4 cycles}, 
		L.\ref{lem:counting st 4paths}}
		& 
		$\Omega(n^{1-\epsilon})$, $O(n)$
        &\cite{HenzingerKNS15}, C.\ref{cor:counting st 4paths worst-case}
		\\
		\midrule
  \begin{tabular}{l@{}}
		$s$ triangles
  \end{tabular}
          &$O(1)$
        & \cite{HLS22}
        &\multicolumn{4}{l}{
  {$\Omega(pn^{1-\epsilon})$, $O(pn)$}}
		&  
  \multicolumn{2}{r|}%
  {T.\ref{thm:lb for counting st4 paths s3 and s4 cycles}, L.\ref{lem:counting s triangles}}
		& $\Omega(n^{1-\epsilon})$,
		$O(n)$
		& \cite{HenzingerKNS15}
		\\
		\midrule
  \begin{tabular}{l@{}}
		$s$ 4-cycles
  \end{tabular}
          &$O(1)$
        & \cite{HLS22}
  &\multicolumn{4}{l}
  {$\Omega(pn^{1-\epsilon})$, $O(pn)$}
		&  
  \multicolumn{2}{r|}%
  {T.\ref{thm:lb for counting st4 paths s3 and s4 cycles},  L.\ref{lem:counting s 4cycles}}
		& $\Omega(n^{1-\epsilon})$,
		$O(n)$
		& \cite{HLS22}, C.\ref{cor:counting s 4cycles worst-case}
		\\
		\midrule
		\midrule
  \multicolumn{11}{|l|}{\textbf{Other key problems}}\\  

 
		\midrule
      \begin{tabular}{l@{}}
		Connectivity
    \end{tabular}
        &\multicolumn{2}{l}{$O(1)$}
        &\multicolumn{2}{r|}{L.\ref{lem:ub_oblivious_flip}}
        &\multicolumn{2}{c|}{---}     
        &$\Omega(\log pn)$
		& T.\ref{thm:lb for connectivity-smoothed advers}
		& $\Omega(\log n)$, $\tilde O(\log n)$\hspace{-5ex}
    	& \cite{PatrascuD06},\cite{HuangHKPT23} 
		\\
  		\midrule  
		    \begin{tabular}{l@{}}
			Perfect matching\\
		\end{tabular}
		&\multicolumn{2}{l}{$O(1)$}
		&\multicolumn{2}{r|}{L.\ref{lem:ub_oblivious_flip}}
		&
		$\Omega(n^{\delta/3-\epsilon})$
		& T.\ref{thm:lb for many problems-smoothed oblivious ar}
		&
		$\Omega(p^{2-\epsilon}n^{1-\epsilon})$
		& T.\ref{thm:lb for many problems-smoothed advers}
		& $\Omega(n^{1-\epsilon})$,
		$O(n^{1.495})$\hspace{-5ex}
		& \cite{HenzingerKNS15}, \cite{Sankowski07}
		\\
		\bottomrule
	\end{tabular}
	\end{adjustbox}
	~
	\caption{%
    Update times $\tu$ of different problems assuming $\tq = O(1)$ and $\tp = o(n^{3-\epsilon})$.
	The complexity of deciding connectivity with an oblivious add/remove adversary (marked ---) remains open.
    The lower bound for perfect matching with an oblivious add/remove adversary is for a constant $0<\delta<1$ and $1-n^{-\delta} \leq p \leq 1$.
    All the polynomial lower bounds are conditioned on the \omv conjecture.
  }
	\label{table: results comparison}
\end{table*}

\begin{atheorem}[informal]
	\label{athm:decision problems-lb}
For any $p\in[0,1]$,
there is no algorithm for 
	\textbf{bipartite perfect matching}
 under an adaptive adversary 
with 
$\tu(n)=O(p^{2-\epsilon}n^{1-\epsilon})$
and
$\tq(n)
=O((pn)^{2-\epsilon})$
for any $\epsilon>0$,
unless the \omv conjecture fails.

If $p\in[1-n^{-\delta} ,1]$ for a constant $0<\delta<1$,
then there is no algorithm 
for \textbf{bipartite perfect matching} under an oblivious add/remove adversary
with 
$\tu(n)=O(n^{\frac{\delta}{3}-\epsilon})$
and
$\tq(n)
=O(n^{\frac{2\delta}{3}-\epsilon})$
for any $\epsilon>0$ 
unless the \omv conjecture fails.
\end{atheorem}

By choosing the best $\delta$ for a given value of $p\in[0,1]$, \cref{athm:decision problems-lb}
implies that there is no algorithm under the oblivious add/remove model with running times 
$\tu(n)=O(\min\big\{\frac{1}{1-p}, n\big\} ^{1/3-\eps})$
and
$\tq(n)
=O(\min\big\{\frac{1}{1-p}, n\big\} ^{2/3-\eps})$.

Finally, we remark that both theorems also apply for
\textbf{bipartite maximum matching} and \textbf{bipartite minimum vertex cover},
and \cref{athm:decision problems-lb} also applies for
	\textbf{counting $s$ $k$-cycles for $k\geq 3$ odd} (and hence $s$-triangle detection),
	\textbf{counting \paths{s}{t}{3}} and \textbf{counting \paths{s}{t}{4}};
 the latter two are subsumed by \cref{athm:small subgraphs-lb} below.

\begin{figure}
	\centering
	\begin{subfigure}[b]{0.41\textwidth}
		\centering
		\includegraphics[width=\textwidth
		]{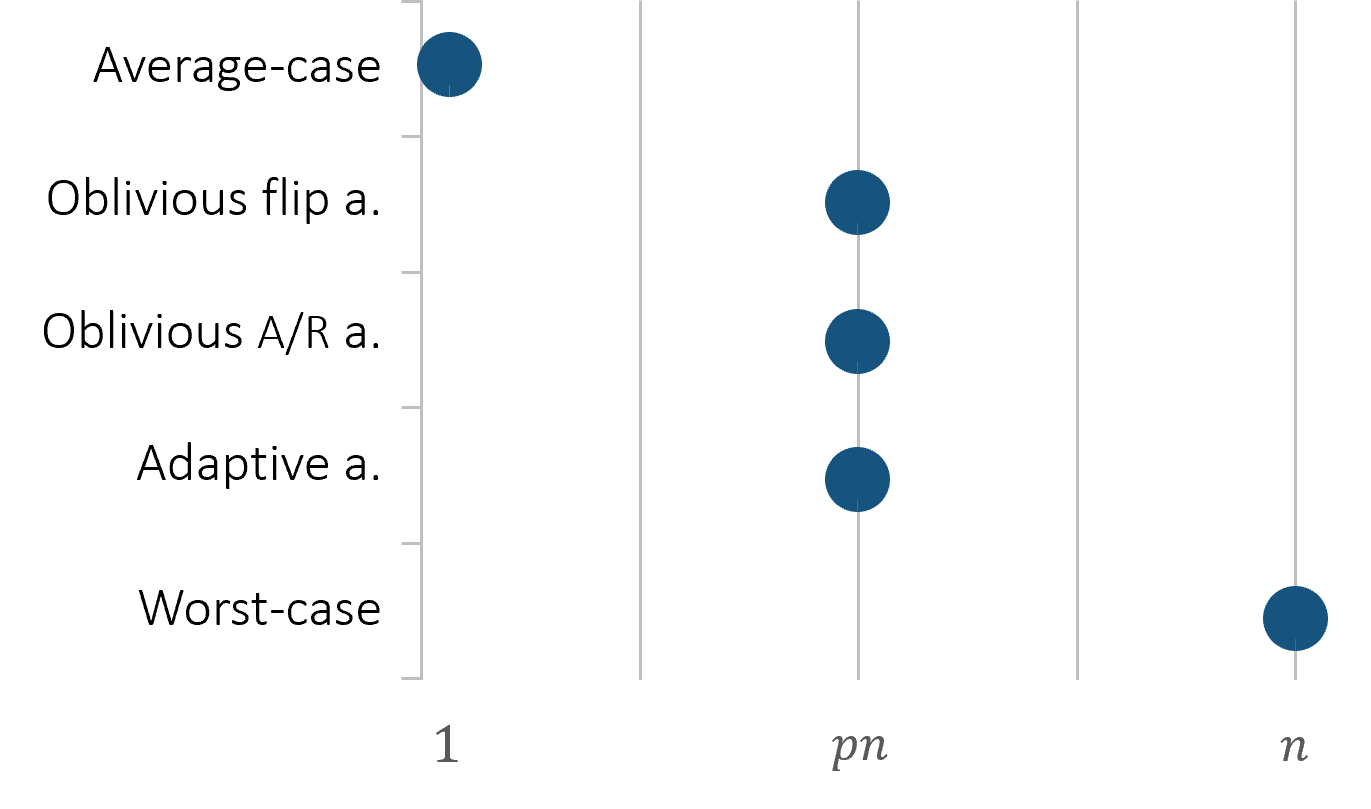}
		\caption{Small subgraph counting}
		\label{fig:subgraph counting}
	\end{subfigure}
	\hfill
	\begin{subfigure}[b]{0.54\textwidth}
		\centering
		\includegraphics[width=\textwidth
		]{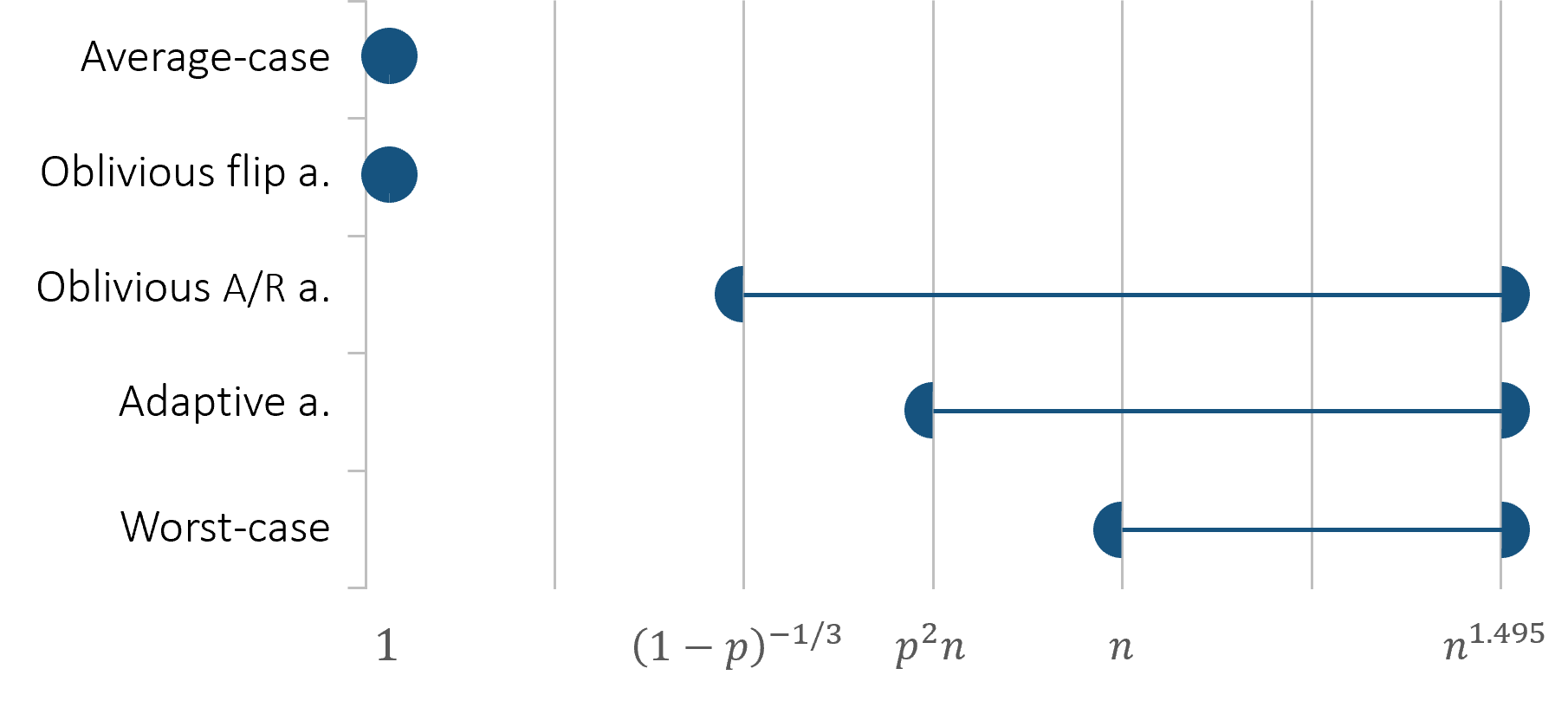}
		\caption{Bipartite perfect matching}
		\label{fig:bipartite perfect matching}
	\end{subfigure}
	\caption{
    The asymptotic update time of two problems under different input models, exemplifying the separations depicted in \cref{fig:relations between models}.
    Separations between smoothed and non-smoothed models arise from \cref{fig:subgraph counting} using, e.g., $p = 1/\sqrt{n}$. \cref{fig:bipartite perfect matching} separates the different smoothed adversaries: adaptive adversary can force higher update time for any $p = \omega(1/\sqrt{n})$, but the two oblivious adversaries are only separated for high smoothing parameter ($p = 1-1/\poly(n)$).}
	\label{fig:complexities separating models}
\end{figure}

\subsubsection{Counting small subgraphs}
Four subgraph counting problems were shown to have constant average-case update time~\cite{HLS22}, but $\Omega\left(n^{1-\eps}\right)$ worst-case update time.
These problems are counting \paths st3, \paths st4, cycles of length $3$ through a fixed node~$s$ ($s$~triangles), and cycles of length $4$ through a fixed node~$s$ ($s$~$4$-cycles).
We establish that their $p$-smoothed update time is essentially identical across all three smoothed models (\cref{fig:subgraph counting}). 
Specifically, in
lemmas~\ref{lem:counting st 3paths},
\ref{lem:counting st 4paths},
\ref{lem:counting s triangles} and 
\ref{lem:counting s 4cycles}
we present simple algorithms (a la~\cite{HLS22}) that achieve $O(pn)$ expected update time even with an adaptive adversary (the strongest adversary).
Roughly speaking, a change to an edge that touch $s$ (or $t$) is ``expensive'' --- it affect many short paths starting from $s$ (or $t$) and thus results in an $O(n)$ update time; other changes take only $O(1)$ time to process.
Expensive changes occur at random only $1/n$ of the time, but an adversary (of any type) can perform only them, resulting in probability $p + 1/n$ for an $O(n)$-update time change, and an $O(pn+1)$ expected update time overall
(which translates to the same amortized update time).

\begin{atheorem}[informal]
	\label{athm:small subgraphs-alg}
There is an algorithm for counting \textbf{\paths{s}{t}{3}, \paths{s}{t}{4}, $s$ triangles} and \textbf{$s$ $4$-cycles}, with 
	 $\tu=O(pn+1)$ (expected)
	and $\tq=O(1)$ 
	under any $p$-smoothing model.
\end{atheorem}

Substituting $p=1$ in the theorem gives a worst-case algorithm with $O(n)$ update time.
This improves upon the state of the art for \paths st4~\cite{DemetrescuI04} and $s$ 4-cycles~\cite{HanauerHH22}, (corollaries~\ref{cor:counting st 4paths worst-case} and~\ref{cor:counting s 4cycles worst-case}).

In theorems~\ref{thm:lb for counting st3 paths} and~\ref{thm:lb for counting st4 paths s3 and s4 cycles} we complement these upper bounds by almost-matching conditional lower bounds on the update time.
Specifically, we show an $\Omega\left(p \cdot  n^{1-\eps}\right)$ lower bound on the update time, for any $\eps>0$, that applies even to an oblivious flip adversary (the weakest one).
This is achieved by a series of reductions, starting from the \omv conjecture.
This part is the most technically involved in this work, and we discuss it further in \cref{sec:intro-technical}.

\begin{atheorem}[informal]
	\label{athm:small subgraphs-lb}
	Any algorithm for the problems from \cref{athm:small subgraphs-alg} must satisfy 
    \[
	\begin{aligned}
        \tp(n) + \frac{n}{p} \cdot \tu(n) + n \cdot \tq(n) + \frac{n^2}{p}= \tilde{\Omega} \left(n^{3-\eps}\right),
    \end{aligned}
	\]
    for any $\eps > 0$,
    under any $p$-smoothing model,
    unless the \omv conjecture fails.
    Specifically, no algorithm 
    can simultaneously have
    $\tp=O(n^{3-\epsilon})$,
    $\tu=O(pn^{1-\epsilon})$,
    and
    $\tq=O(n^{2-\epsilon})$.
\end{atheorem}

This completes the picture for \paths{s}{t}{k} and the transition between $k=2$ and $k=5$
(see \cref{table: short paths counting}), showing how the complexity escalates with both the path length and the degree of adversarial control (parameter~$p$).

\subsubsection{Hierarchy between smoothing models}
\label{sec:hierarchy-separations}

While all our models coincide with the worst-case model at $p = 1$ and an average-case input at $p = 0$, 
in \cref{sec:adver} we show 
a hierarchy among the smoothing models for different regimes of $0 < p < 1$.
The adaptive adversary
can be defined with either operation type, 
and can disregard the additional information it is privy to, making it at least as strong as both its oblivious counterparts.
The oblivious add/remove adversary can simulate the oblivious flip adversary (while using $p' = p/(2-p)\in [p/2, p]$, as stated in \cref{claim:oblivious ar vs flip}).
This induces a (non-strict) hierarchy between the models (\cref{fig:relations between models}). We next discuss separation results for these models.

For some problems, such as counting small subgraphs, the three smoothing models coincide with one another (\cref{fig:subgraph counting}): they all admit 
an $O(pn)$ algorithm (\cref{athm:small subgraphs-alg})
and a conditional $\Omega(pn^{1-\epsilon})$ lower bound (\cref{athm:small subgraphs-lb}).
These bounds separate the smoothing models from the worst case for $0\leq p<n^{-\delta}$ and from the average case 
for $n^{-\delta}<p\leq1$, for any constant $\delta > 0$.

For other problems the three smoothing models demonstrate separations, i.e., different complexities for the same problem.
This is the case with perfect bipartite matching (\cref{fig:bipartite perfect matching}):
under an oblivious flip adversary with $p\in[0,1-\frac{26\log n}{n})$, the problem admits a trivial ``yes'' algorithm since the graph contains enough random edges at all times (\cref{athm:decision problems-alg}).
In contrast, an adaptive adversary can embed a worst-case lower bound construction within a small set of nodes, leading to a polynomial lower bound of $\Omega(p^{2-\epsilon}n^{1-\epsilon})$ for the problem (\cref{athm:decision problems-lb}).
This separates the oblivious flip adversary from the adaptive one for almost any value of $p$, e.g., for $p\in \big[\frac{\log n}{\sqrt n},1-\frac{26\log n}{n}\big)$.

The more surprising separation is between the 
oblivious flip adversary and the oblivious add/remove one, 
as operation types have no effect on the complexity in all other regimes (worst-case, average-case, and adaptive smoothed model).
Using our embedding technique,  \cref{athm:decision problems-lb} gives a super-constant lower bound on the update time of bipartite perfect matching when $p \approx 1 - o(1)$. This separates the two oblivious models,
e.g.,
when $1 - p\in \big(\frac{26\log n}{n},\frac{1}{\log n}\big]$.

Finally, we conjecture that the oblivious add/remove adversary is strictly weaker than the adaptive adversary, but this is left unresolved for the moment.

\subsection{A technical overview of the subgraph counting lower bounds}
\label{sec:intro-technical}

The most technically involved proof is the lower bounds for counting small subgraphs, even for the oblivious flip adversary, which we overview in this section.
The general outline is based on that of lower bounds for counting larger subgraphs in the average case~\cite{HLS22}, consisting of a chain of fine-grained reductions.
Evidently, these tools alone do not suffice for the small subgraphs we consider here, as these are easy to count in the average case. 
Our hardness results utilize an \emph{adversarial strategy}, and in its core lies a reduction that generates a sequence of changes resembling certain $p$-smoothed sequences of changes, pinpointing the dependency on~$p$.
To quantify this resemblance we use the standard notion of statistical distance between distributions (or random variables). 
In our arguments we repeatedly use the \emph{Poissonization} technique, described below, to break dependencies between key random variables. 
This technique is well-suited for our smoothed models, and we believe it should prove useful for other smoothed models of computation as well.

\paragraph{Poisson distribution.}
We recall the \emph{Poisson distribution} $\poisson{\lambda}$ is supported on $\N_{\geq 0}$ and parameterized by its mean value $\lambda$ (see~\cite[Section~8.4]{books:probability}).
In a nutshell, 
$\poisson{\lambda}$ is the number of successful coin tosses 
when tossing infinitely many times a coin with infinitesimally small success probability, such that the expected number of successful attempts is~$\lambda$. 
Formally, this is the limit for $n \to \infty$
of the binomial distribution 
$B(n, \lambda/n)$.
Useful properties of this distribution are summarized in \cref{fact:poisson_properties} in \cref{sec:tools_from_probability}.

\paragraph{The Poissonization technique.}
A particularly useful property of the Poisson distribution is known as \emph{Poissonization} (\cref{fact:item:poissonization} of \cref{fact:poisson_properties}), and we next describe it in our setting.
Consider a distribution~$\mathcal{D}$ over a set $E$ of edges.
A multi-set of edges sampled from~$\mathcal{D}$ can be described by its histogram $\mathcal{H}:E\to \N_{\geq 0}$, where $\mathcal{H}(e_i)$ represents the number of times the edge~$e_i\in E$ is sampled.
When drawing $r$ such samples, the histogram values $\mathcal{H}(e_1), \mathcal{H}(e_2), \dots$ exhibit a slight dependency among them, merely due to the fact that their sum is fixed to $r$.
This dependency poses complications in probabilistic and information-theoretic analysis, especially
if additional constraints on the histogram values exist.
To address this, we independently sample the value $\mathcal{H}(e_i)$ for each edge from $\poisson{r\cdot \mathcal{D}(e_i)}$, 
where $\mathcal{D}(e_i)$ represents the probability of $e_i$ in $\mathcal{D}$.
This yields a total of $t = \poisson{r}$ samples, which might not be precisely $r$ samples,
but remains well-concentrated around~$r$
and even allows sampling conditioned on individual histogram values having e.g., a desired parity.

We apply Poissonization in several ways in our proofs, most prominently in Step~$2$ below.
This helps to simplify and clear up some hardness results of~\cite{HLS22} and correct inaccuracies.
More crucially, Poissonization helps us overcome the additional challenges that arise due to the adversarial presence.

\medskip
Sections~\ref{sec:lb:ac_oumv_to_restricted_graphs},~\ref
{sec:lb:restricted_graphs_to_general_graphs}
and~\ref{sec:lb:wrapping_things_up} 
together prove a lower bound for counting \paths st3, which is extended to other problems in \cref{sec:lb_other_small_graphs}.
The bounds are conditioned on the \omv conjecture~\cite{HenzingerKNS15}, which formalizes the common belief that a truly sub-cubic combinatorial algorithm for matrix multiplication does not exist.

\paragraph{Step 1: massaging the \omv conjecture}~\\
The first step reduces between online algebraic matrix-vector multiplication problems. 
It consists of an algebraic reduction that was presented in~\cite{HenzingerKNS15} and is standard by now, 
and two reductions that were presented in~\cite{HLS22}.
The latter are reviewed in \cref{app: from omv to parity oumv}.

\begin{definition}[\omv problem]
	The \emph{\omv} problem is the following online problem. Fix an integer~$n$. 
	The initial inputs are a Boolean $n$-vector $v_0$ and a Boolean $n\times n$ matrix $M$.
	Then, Boolean vectors~$v_i$ of dimension $n$ arrive online, for $i=1,\dots, n$.
	An algorithm solves the \omv problem correctly if it outputs the Boolean vector $Mv_i$ before the arrival of $v_{i+1}$, for $i = 0,1,\dots, n-1$, and outputs~$M v_n$.
\end{definition}

\begin{restatable}%
	[\omv conjecture~{\cite[Conjecture 1.1]{HenzingerKNS15}}]%
	{conjecture}{omvconj}%
	\label{conj:omv}
	For any constant $\eps > 0$, there is no $O\left(n
	^{3-\eps}\right)$-time algorithm that solves
	\omv with error probability at most $1/3$ in the word-RAM model with $O\left(\log n\right)$-bit words.
\end{restatable}

After a standard reduction from the \omv problem to the \oumv problem, two other reductions are made.
The first reduces to \oumv with operations modulo $2$ (that is, over $\F_2$ instead of $\R$), as defined below.
The second reduction is a typical worst-case to average-case reduction over $\F_2$, showing the problem is as hard on random inputs. 
We define the \oumv problem and state the reduction to conclude this step.

\begin{restatable}[parity \oumv problem]{definition}{DefParityOUMV}
	The \emph{parity \oumv} problem is the following online problem. Fix an integer $n$. The initial inputs are two Boolean $n$-vectors $u_0, v_0$, and a Boolean $n\times n$ matrix $M$.
	Then, pairs of Boolean vectors $(u_i, v_i)$ of dimension $n$ arrive online, for $i=1,\dots, n$.
	An algorithm solves the \emph{parity \oumv} problem correctly if it outputs the Boolean value $u_i^T M v_i$ (over $\F_2$) before the arrival of $(u_{i+1},v_{i+1})$, for $i = 0,1,\dots, n-1$, and outputs $u_n^T M v_n$.
	
	An algorithm solves the \emph{average-case parity \oumv} problem if it answers correctly, with error probability at most $1/20$, over inputs that are chosen independent and uniformly at random (that is, all entries of the matrix $M$ and the vectors $u_i, v_i$ are i.i.d.\ random bits).
\end{restatable}

The sequence of reductions concludes in the following lemma.

\begin{restatable}[\hspace{1sp}{\cite[Section 2]{HLS22}}]{lemma}{lemavgparoumv}
\label{lem:conjecture_to_ac_oumv}
    For any constant $\eps > 0$, there is no $O\left(n
	^{3-\eps}\right)$-time algorithm that solves
	the average-case parity \oumv problem, unless the \omv conjecture fails.
\end{restatable}

\begin{remark}
    At this point,~\cite[Section~2.1]{HLS22} reduces the former problem to a newly defined algebraic online problem, the \emph{biased average-case \oumv problem}.
    Unfortunately, this leads to some inconsistencies in their proof (e.g., the statements regarding this online problem use the notions of update and query times, which are irrelevant to online problems).
    To circumvent this,
    our argument diverges from theirs, directly reducing the parity \oumv problem to a dynamic graph problem.
\end{remark}

\paragraph{Step 2: from \oumv to counting \paths{s}{t}{3} in restricted graphs}~\\
In \cref{sec:lb:ac_oumv_to_restricted_graphs},
we show that average-case parity \oumv can be solved using an algorithm that counts \paths{s}{t}{3} in $P_3$-partite $p$-smoothed dynamic graphs
with an oblivious flip adversary.

A $P_3$-partite graph is composed of $n'=2n+2$ nodes partitioned as 
$V = \set{s} \sqcup A \sqcup B \sqcup \set{t}$, where $\card{A} = \card{B} = n$,
and can only have edges of types 
$sA, AB, Bt$
(that is, $E\subseteq (\{s\}\times A)\cup(A\times B)\cup(B\times\{t\})$).
To represent the \oumv problem using this graph,
consider $u,v$ and $M$ as the indicators of the $sA, Bt$ and $AB$ edges, respectively (e.g., $(s,A_j)\in E\iff u(j)=1$).
It is not hard to see that the product $u^TMv$ is exactly the number of \paths st3.

Assume we have a $P_3$-partite graph that represents $M$ and $(u_{i-1},v_{i-1})$ as above,
and we get the next pair $(u_i,v_i)$
and want to update the graph accordingly.
In worst-case analysis, such reduction can use adversarial choices: flip the $sA$ and $Bt$ edges to represent $(u_i,v_i)$, leave the $AB$ edges intact,
and then query the number of \paths st3 to retrieve $u_i^TMv_i$.
When random changes are involved in the process, however, undesired changes to $AB$ edges are bound to occur. 
Moreover, such changes are \emph{more likely} than others, as there are $n^2$ potential edges of types $AB$ and only $2n$ potential edges of the other types ($sA$ and $Bt$).
Next, we describe how to overcome this challenge.

For the mere problem of having \emph{any} changes in $AB$ edges, a standard solution used in~\cite{HLS22} serves here as well: the reduction creates $3$ copies of the graph, 
such that whenever a query is made, each has different $AB$ edges but they have exactly the same $sA$ and $Bt$ edges.
To be precise, the $AB$ edges in the 3 graphs will represent 3 matrices $M_1,M_2,M_3$ such that $M \sameparity M_1+M_2+M_3$, 
while the $sA$ and $Bt$ edges in all the graphs represent $u_i$ and $v_i$.
This assures that the sum modulo~2 of the number of \paths st3 in the three graphs gives $u_i^TMv_i$ with the original matrix $M$:
\begin{equation}
    \label{eq:intro:3_matrices}
    u_i^T M v_i \sameparity u_i^T M_1 v_i + u_i^T M_2 v_i + u_i^T M_3 v_i .    
\end{equation}

The harder challenge is to synthesize $3$ such correlated change sequences that resemble authentic $p$-smoothed change sequences. 
To this end, consider an \emph{oblivious} adversary that chooses an $sA$ or $Bt$ edge u.a.r. 
After $p$-smoothing, each change has the following distribution over the graph edges
\[
\begin{aligned}
	\advdist^p (e):=
	\begin{cases}
		\frac{p}{2n} + \frac{1-p}{n(n+2)}, & \text{if $e$ is of type $sA$ or $Bt$}\\
		\frac{1-p}{n(n+2)}, & \text{if $e$ is of type $AB$}
	\end{cases}\;.
\end{aligned}
\]

This guarantees the highest chances for a change of type $sA$ and $Bt$ after smoothing, 
which is crucial for avoiding redundancies in the reduction and obtaining a tight lower bound.
Moreover, since all edge changes are drawn from $\advdist^p$, we can simulate a sequence by generating a multi-set of edges and later randomize the order in which they are changed.
We create a sequence of roughly $\roundnumber = 5 n \log(n) / p$ changes using Poissonization.
This breaks dependencies and allows us to construct a multi-set of samples by drawing the number of appearances of each edge independently.

\subparagraph{Edges of types $sA$ and $Bt$.}
Recall that changes to $sA$ and $Bt$ are shared by all $3$ copies; hence, we create one multi-set of such changes, $S_0$. 
Let~$\udif=u_i-u_{i-1}$ and $\vdif=v_i-v_{i-1}$ be the difference vectors, and denote by $z_e$ the number of times the $sA$ edge $e=(s,a_j)$ is added to~$S_0$. 
Our goal is to ensure that for each such $sA$ edge, $z_e \sameparity \udif(j)$ (and similarly for $Bt$ edges and $\vdif$).
Naively using
rejection sampling (resampling the multi-set $S_0$ until all conditions are met) would take $2^{\Omega(n)}$ attempts.
Instead, we use Poissonization:
for each $sA$ edge $e$, draw a Poisson number of samples $z_e$ independently,
conditioned on the desired parity, i.e.,  $z_e \sim \poisson{\advdist^p (e)\cdot \roundnumber} \vert_{z_e \sameparity \udif(e)}$.
Rejection sampling of each $z_e$ value separately only requires $\Theta(n)$ attempts overall.

\subparagraph{Edges of type $AB$.}
Drawing $z_e$ for each $AB$ edge $e$ in a similar way would require $\Omega\left(n^2\right)$~time, which is too costly for this fine-grained reduction. 
However, only $O(t) = o(n^2)$ changes of type $AB$ are expected, so we instead draw the total  number of $AB$ changes from a Poisson distribution with parameter $t_{AB} = (1-p)nt /(n+2)$, and then assign an $AB$ edge u.a.r.\ for each of these changes.
In actuality, we wish to create $3$ correlated copies, where $AB$ edges cancel out. 
To this end, we draw three multi-sets $S_1, S_2, S_3$, each composed of $\poisson{t_{AB}/2}$ edges of type $AB$, where each edge is chosen u.a.r.
Each copy uses exactly two of the sets $S_1, S_2, S_3$, such that each set is used twice. This ensures that~\eqref{eq:intro:3_matrices} holds.

Adding $S_0$ to the multi-set of each copy and shuffling each of them separately produces
$3$ sequences of edge changes
that are statistically close to a sequence of $\poisson{t}$ edge changes drawn from $\advdist^p$.
Bounding the statistical distance limits the additional error from the imperfect simulation, guaranteeing the error probability on the output $u_i^T M v_i$ is only slightly higher than the error probability of $3$ executions of the \paths st3 counting algorithm on genuine $p$-smoothed inputs.

\begin{remark}
The simulation above is perhaps the crux of the reduction, generating inputs for the dynamic graph algorithm from inputs to the \oumv problem.
We note that the simulation used in~\cite{HLS22} was much simpler and easily shown to be efficient, but its proof of correctness was incomplete. 
One place that is lacking is a false proof for Claim C.2. While this is fixable for their specific case, it stems in an inescapable subtlety of handling dependencies. 
Our use of Poissonization in both the simulation and the analysis tackles these subtleties and formally prove our claims, as well as theirs.
\end{remark}

\paragraph{Step 3: from restricted graphs to general graphs}~\\
We continue in \cref{sec:lb:restricted_graphs_to_general_graphs}, by showing how to count \paths st3 in the restricted class of $P_3$-partite graphs using an algorithm that counts such paths in general graphs.
This reduction uses \emph{inclusion-edgeclusion}: it takes a $P_3$-partite graph $\calG$, and constructs $16$ (general) graphs by adding edges that were previously not allowed (\emph{exterior} edges). The same approach was taken in the average-case analysis, but smoothed inputs require more attention due to the adversarial presence.

To account for the adversarial presence, the transition from a $p$-smoothed $P_3$-partite graph to $p'$-smoothed graphs must use $p' \leq p$ that sufficiently weakens the adversary.
Still, this can be done using $p' \in [p/2, p]$, 
which is crucial for a lossless reduction leading to (nearly) tight lower bounds.

The $16$ new dynamic graphs all use the original node set
$V = \set{s} \sqcup A \sqcup B \sqcup \set{t}$, and initially all have the same interior edges as the $P_3$-partite graph.
Their initial exterior edges are correlated in the following sense: they \emph{properly partition} the four exterior edge sets $E_{At}, E_{AA}, E_{BB}, E_{sB}$.
That is, there exists four partitions, one for each edge set, such that the exterior edges of each of the~$16$ initial graphs 
constitute a unique combination of the parts, one from each edge set.

We next describe how to construct a single change, using a parameter $\alpha(n,p) \approx \frac{1}{2-p}\in[1/2,1]$.
We flip a coin $C\sim Bin(\alpha)$; 
if $C = 1$ we make an interior change, the next change in the $p$-smoothed sequence of $\calG$, 
while if $C = 0$ we change a random exterior edge.
Choosing $\alpha$ carefully and $p' = \alpha \cdot p$ guarantees that a sequence of such changes is $p'$-smoothed on a general graph, for an appropriate adversarial strategy. We apply this sequence of changes to each of the $16$ graphs.

Since the graphs undergo the exact same changes, they properly partition the four exterior edge sets at all times.
This invariant ensures that the total number of exterior \paths st3 (i.e., not of type $sABt$) in all $16$ graphs is always easy to compute.
After querying all graphs for their \paths st3 count, we sum the answers, subtract the exterior paths, and divide by $16$.
This gives the number of $sABt$ paths in $\calG$, since each interior path appears in all graph and was counted $16$ times.

\medskip
To summarize, steps 1-3 dictate a chain of fine-grained reductions as follows.
An algorithm for counting \paths st3 in general graphs can be used to count them in a $P_3$-partite graph, which in turn can be used to solve the average-case parity \oumv problem.
A solution for the average-case parity \oumv problem implies a solution to the \omv problem, contradicting the \omv conjecture.

\subsection{Related work}

\subsubsection{Dynamic graph algorithms}

Dynamic graph algorithms constitute a deep and well-studied research topic, 
as was recently presented in the thorough survey of Hanauer, Henzinger and Schulz~\cite{HanauerHS22}.
The most recent dynamic algorithm for connectivity is by
Huang, Huang, Kopelowitz, Pettie and Thorup~\cite{HuangHKPT23}, and has $\tu=O(\log n \log^2\log n)$ amortized update time and $\tq=O(\log n/\log\log\log n)$ time per query.
Small subgraph counting was recently studied by Hanauer, Henzinger and Hua~\cite{HanauerHH22}, 
who showed, e.g., that counting \paths st3 and $s$-triangles (triangles through a given node $s$) can be done with $\tu=O(n)$ and constant query time 
(see also \cref{sec:poly lbs} below).
Kara, Nikolic, Ngo, Olteanu and Zhang~\cite{KaraNNOZ20} studied databases that support enumeration of triangles in trinary relations,
and their work implies triangle counting in dynamic graphs with $\tp=O(m_0^{3/2}), \tu=O(m^{1/2})$, and~$\tq=O(1)$.

P\u{a}tra\c{s}cu and Demaine~\cite{PatrascuD06} 
proved the first unconditional lower bound for dynamic graph algorithms, focusing on connectivity problems in undirected graphs.
One of their results asserts that there is no dynamic algorithm for deciding whether the graph is connected with a constant query time
and $\tu=o(\log n)$	amortized time
(see~\cref{lem:lb for connectivity-original}),
a results that remains the best known to date.
Follow-up works on their work~\cite{PatrascuT11a,CliffordGL15,LarsenY23},
considered, e.g., the case of parameterized queries (querying whether two given nodes are in the same connected component) and directed graphs.

Following a line of work on conditional lower bounds in dynamic graphs~\cite{AbboudW14,Patrascu10,DorHZ00,RodittyZ11,KopelowitzPP16}, which used several different conjectures,
Henzinger, Krinninger, Nanongkai and Saranurak~\cite{HenzingerKNS15} introduced the \omv conjecture as a unified conjecture.
They showed, e.g., that 
$s$-triangle detection,
bipartite matchings
and
bipartite vertex cover cannot be done with
polynomial preprocessing time, 
	$\tu=O(n^{1-\epsilon})$
and
$\tq=O(n^{2-\epsilon})$,
for any $\epsilon>0$ and even amortized and randomized,
unless the \omv conjecture is false.	
The main focus of the works until this point was on worst-case analysis,
and average-case analysis was rarely applied, as described next.

\subsubsection{Beyond worst-case analysis of dynamic graph algorithms}

The first to study average-case complexity of dynamic graph algorithms were 
Nikoletseas, Reif, Spirakis and Yung~\cite{NikoletseasRSY95}.
They consider an initial empty graph, a phase of addition of random edges, and then a phase where random edges are added or remove with equal probabilities. 
In this setting, they present an algorithm for $(u,v)$-connectivity, i.e., where a query includes a pair $(u,v)$ of nodes and has to determine if they are in the same connected component.

Following this, Alberts and Henzinger~\cite{AlbertsH98} 
defined the ``restricted randomness'' model.
To this end, they note that each change consists of a type (add or remove) and a parameter (which edge to add or to remove).
In their model, an adversary chooses a sequence of types (add or remove) for the changes,
and then the parameter (edge) of each change is chosen uniformly at random from the relevant set (existing edges for a remove operation, node-pairs without edges for an add operation).
They show that several problems have better amortized time complexities in the restricted randomness model compared to the best worst-case bounds known at the time (almost no lower bounds where known at the time).
The problems considered include connectivity and its variants (edge and node connectivity), 
maximum cardinality matching, 
minimum spanning forest, and bipartiteness. 

Indeed, one could see this model as an average-case model, but we prefer to use this term only for uniformly random inputs (as in~\cite{NikoletseasRSY95,HLS22}). 
It can also be seen as a smoothed analysis model, though without the ability to parameterize the amount of randomness.

As this model consists in an adversary and a random process, it should not come as a surprise that both an oblivious adversary and an adaptive one can be considered in this setting.
Here, an adversary that chooses the sequence of types in advance is called an oblivious one, and an adversary that can choose each type according to the outcome of the random process so far is an adaptive one.
These variants are not equivalent:
for example, an adaptive adversary can make sure that all the $n$-edge graphs in the sequence are connected (by adding an $n$-th edge to an $(n-1)$-edge graph only if it is connected), while with an oblivious adversary every $n$-edge graph has equal probability, and most of them are not connected.
The distinction between adversaries is not made in~\cite{AlbertsH98}, but the proofs use the fact that each graph has equal probability (conditioned on the number of edges), thus the model there is oblivious.

The average-case complexity of dynamic graph algorithms was disregarded for decades, until recently Henzinger, Lincoln and Saha~\cite{HLS22}
approached it again with new tools.
They studied subgraph counting problems, 
presented some simple algorithms inspired by~\cite{AlbertsH98},
and focused mainly on proving conditional lower bounds.
Their proofs use ideas from fine-grained complexity~\cite{HenzingerKNS15}, and specifically lower bounds for average-case complexity in this setting~\cite{DalirrooyfardLW20,Boix-AdseraBB19}.
For some problems, such as counting \paths st5, they show that the average-case complexity cannot be better than the worst-case one (up to an $n^\epsilon$ factor), while for others, such as 
\paths st3, they prove a big gap exists.
As smoothed analysis lie between the worst case and the average one, we focus the current work only on problems of the last type, where a gap exists.
Some of our lower bound proofs use techniques from~\cite{HLS22};
for example, while they do not give a lower bound for \paths st3 (since it is easy in the average case) we give a lower bound for this problem in the smoothed case
by extending the technique they used for proving a lower bound for \paths st5.

\subsubsection{Smoothed analysis in different models}
Smoothed analysis was introduced by Spielman and Teng~\cite{SpielmanT04} in relation to using pivoting rules in the simplex algorithm. 
Since then, it have received much attention in sequential algorithm design.
In particular interest to us are smoothed analysis results for static graphs~\cite{krivelevich2006smoothed, friedrich2011smoothed,krivelevich2015smoothed} 
and dynamic networks~\cite{DFGN18,MPS20,GMPS21,DFGN22},
which we use as inspiration for our model of noise.

In recent years, smoothed analysis was also applied in settings of a repeated game, where inputs are chosen successively.
These include dynamic networks, population protocols~\cite{SS21} and online learning~\cite{HRS20,HRS21}.
The effect of the adversary's adaptivity was first studied with regard to information spreading in dynamic networks~\cite{MPS20},
which is shown to be polynomially slower when the adversary is adaptive compared to an oblivious adversary.
The adaptivity question was also studied in the context of online learning~\cite{HRS21}, and a handful of problems were shown to have the same smoothed complexity whether the adversary is adaptive or oblivious.
From these, the dynamic networks works are the most relevant to ours, as they defined similar smoothing models. 
Their results, however, focus on round complexity and not computation time, and thus are incomparable with ours.

\section{Preliminaries}
\label{sec:prelim}
We use $[n]$ for the set $\set{1,2,\dots,n}$, and $\symdif$ to denote the symmetric difference between two sets. $T$ denotes the number of edge changes in the dynamic graph.

We extend the definition of a dynamic graph to the case where some updates are chosen at random while others are adversarial. While we explore variants on this model, the basic definition is for the \emph{adaptive flip adversary}:

\begin{definition}[$p$-smoothed dynamic graph, adaptive]
    A $p$-smoothed dynamic graph with an adaptive adversary is a dynamic graph $\calG = (G_0,\dots,G_T)$ created by the following random process.
    \begin{itemize}
        \item Initially, $H_0$ is proposed by an adversary. 
        For each node-pair $e\in \binom{V}{2}$, 
        with probability $p$ it is included in $G_0$ according to $H_0$, and
        with probability $1-p$ it is re-sampled uniformly, i.e., included in $G_0$ as an edge with probability $1/2$.
        
        \item At step $i\in [T]$,
        a node-pair $e\in \binom{V}{2}$
        is proposed by the adversary.
        With probability~$p$ this $e$ remains unchanged, and with probability~$1-p$ it is discarded and a new  $e$ is chosen uniformly at random from $\binom{V}{2}$.
        The new graph $G_i$ has edge set $E_i = E_{i-1} \symdif \set{e}$.
    \end{itemize}
\end{definition}

We also consider \emph{oblivious} adversaries, that fix all their choices ahead of time, but then in execution each choice is replaced by a random choice with probability $1-p$. In this model we allow for either a flip adversary or an add/remove adversary:

\begin{definition}[$p$-smoothed dynamic graph, oblivious]
    A $p$-smoothed dynamic graph with an oblivious adversary
    is a dynamic graph $\calG = (G_0,\dots,G_T)$ created by the following random process.
    \begin{itemize}
        \item An adversary proposes 
        an initial graph $H_0$ and a sequence of $T$ changes, where each change consists of a node-pair $e_i\in \binom{V}{2}$.     
        An add/remove adversary also chooses an action $a_i$: add or remove.

        \item 
        For each node-pair $e\in \binom{V}{2}$, 
        with probability $p$ it is included in $G_0$ according to $H_0$, and 
         with probability $1-p$ it is re-sampled uniformly, i.e., included in $G_0$ as an edge with probability~$1/2$.

        \item At step $i\in [T]$, randomly choose $c_i\sim \bernoulli{p}$ and act accordingly:
        \begin{itemize}
            \item if $c_i=1$: use the edge $e_i$.
            For a flip adversary, flip it, and the new edge set is $E_i = E_{i-1} \symdif \set{e_i}$.
            For an add/remove adversary use action $a_i$ and accordingly the resulting edge set is either $E_i = E_{i-1} \setminus \set{e_i}$ or $E_i = E_{i-1} \cup \set{e_i}$.

            \item if $c_i=0$: an edge $e$ is chosen uniformly at random from $\binom{V}{2}$, and the new graph $G_i$ has the edge set $E_i = E_{i-1} \symdif \set{e}$.
        \end{itemize}
    \end{itemize}
\end{definition}

Note that a graph created by this process can also be thought of as a graph created step by step as in the adaptive adversary case, but by an adversary (flip or add/remove) that does not see any of the random choices in the initial graph and previous steps.

\paragraph{The initial graph.}
Smoothing the initial graph is crucial for the model to extrapolate between worst-case and average-case inputs, but it seem to make little difference in all our results.
Our algorithms sometime assume an adversarial initial graph (that is, their analysis does not use the smoothing of the initial graph).
Our lower bounds rely on a uniformly random initial graph $G_0$, using the following simple observation:
\begin{observation}
    \label{observation:initial graph}
    The adversarial strategy that proposes an initial graph $H_0$ uniformly at random results in a smoothed initial graph $G_0$ which is uniformly random as well. 
    This strategy can be employed by any type of adversary and any parameter $0\leq p\leq 1$.
\end{observation}

\paragraph{Restricted dynamic graphs.}
In ~\cref{sec:lb small subgraphs}, we consider a variant where the graph edges are restricted to a set $R_H \subseteq \binom{V}{2}$ of potential edges.
In this model \emph{all changes}, both random and adversarial, as well as the initial graph $G_0$, are on $R_H$ instead of the entire set $\binom{V}{2}$.

We particularly focus on $(H,\Pi)$-partite graphs, defined as follows.
\begin{definition}[$(H,\Pi)$-partite graph]
    Fix a simple graph $H$ on node set $[k]$, and consider a graph $G$ on node set $[n]$.
    We say that $G$ is an $(H,\Pi)$-partite graph with a mapping $\Pi:[n]\to[k]$ if any edge in $G$ is mapped to an edge in $H$, i.e. if $(u,v)\in E_G$ then $(\Pi(u),\Pi(v)) \in E_H$.

    If all the graphs $G$ of a dynamic graph $\calG$ on $[n]$ are $(H,\Pi)$-partite, we say $\calG$ is $(H,\Pi)$-partite.
    In this case we use $R_H$ to denote the set of edges that are allowed in $\calG$, i.e. 
    \[
        R_H := \set{(u,v) \mid (\Pi(u),\Pi(v)) \in E_H} .
    \]
\end{definition}
For example, if $H$ is composed of two nodes and a single edge and $G$ is a $(H,\Pi)$-partite graph for some $\Pi$, 
then $G$ is bipartite.
In some cases we describe $\calG$ on a node set with $k$ parts, calling it $H$-partite (where $H$ has $k$ nodes). 
In these cases $\Pi$ and $H$ are implicit for ease of presentation.
For example, in~\cref{sec:lb small subgraphs} we discuss $P_3$-partite graphs where $\calG$ uses node set $V = \set{s} \bigsqcup A \bigsqcup B \bigsqcup \set{t}$. It is then clear from context that $\Pi(A)$ is fixed ($\Pi(B)$ as well), and $H$ is a $3$-path ($P_3$) graph: $\Pi(s) - \Pi(A) - \Pi(B) - \Pi(t)$.

For a partition of the node set of $\calG$, we naturally categorize edges and longer paths using the different parts.
For example, in a general graph on $4$ sets of nodes there could be many edge types, but if the graph is $P_3$-partite (with the mapping above) then only edges of types $sA, AB, Bt$ exist and $3$-paths from $s$ to $t$ are all of type $sABt$.

Our proofs use some probabilistic tools, detailed in \cref{sec:tools_from_probability} and shortly described next.
We make an extensive use of the \emph{Poisson distribution} and more importantly the \emph{Poissonization} technique (see, e.g., \cite[Section~8.4]{books:probability}).
We use $\poisson{t}$ for the Poisson distribution with parameter $t$ and sometimes for a random variable drawn from this distribution.
Some useful properties of the Poisson distribution are laid out in \cref{fact:poisson_properties}.
For the distance between two distribution $P, Q$ over the same domain we use the well-known \emph{statistical distance} (sometimes called total-variation distance), denoted by $\sd PQ$.
Some useful properties of the statistical distance can be found in \cref{fact:SD_properties}.

\section{The relative power of the adversaries}
\label{sec:adver}

\paragraph{Adaptive vs.\ oblivious adversaries.}
While we assume an adversary is aware of the problem in hand, the smoothing parameter $p$ and other parameters of the problem, 
awareness to the actual noisy changes occurred is a completely different matter
(as is the case in other multi-round smoothing models). 
We show that for some problems, an adaptive adversary can cause higher running times compared to the oblivious adversaries, separating the models. 
Perhaps surprisingly, we show that for counting small subgraphs this is not the case, and adaptivity has little to no affect on the complexity of the problem.

\subsubsection*{Flip vs.~add/remove operations} A priori, it is unclear whether an algorithm that handles edge-flip operations can deal with add/remove operations, or vice versa.
This is typically a non-issue, as the vast majority of research in dynamic graph algorithm sticks to worst-case analysis, where the adversary has full knowledge of the (in)existence of edges and picks additions or removals accordingly.
Slightly more subtle is the case of average-case analysis, where the location of each change is chosen at random. 
The recent work defining average-case analysis of dynamic graphs~\cite{HLS22} assumed an edge chosen at random is always flipped, i.e., an \emph{average-case flip model}.
Let us define the \emph{average-case add/remove model} as one where in each round, both an edge and an operation type are chosen at random 
(this may sound similar to~\cite{AlbertsH98} where the type was chosen by an adversary and the edge at random,
but the models are actually very different).
Unsurprisingly, we show below that the average-case model with flip and with add/remove operations are equivalent, up to a constant factor.

\begin{claim}[Equivalence between flips and add/remove models in average-case analysis]
	\label{claim:random flip vs ar}
	An average-case adversary with flip operations has the same power as an average-case adversary with add/remove operations.
	That is, 
	they have the same success probability up to $1/n^c$ additive error, and the same amortized times up to constant factors, over any sequence of $\Omega(\log n)$ changes.
\end{claim}

\begin{proof}
	The key for the proof is translating add/remove changes to flip changes with laziness.
	
	For the sake of the proof, consider a third average-case model, \emph{lazy flip}.
	In each step of this model no change occurs with probability $1/2$, call this `null', 
    and otherwise a uniformly random edge is flipped.

    This lazy flip model is entirely equivalent to the average-case add/remove model\footnote{Note the equivalence can only hold for an add/remove model with equal probabilities. 
    Otherwise, the graph tends to a different density, while a lazy flip tends to $1/2$ with any laziness factor.}.
    Indeed, first choose the edge $e$ at random. 
    In the former model we choose whether to flip it or not, while in the latter we choose whether to try adding or try removing it. 
    In both models the edge is  consequently flipped with probability $1/2$.
	
	We next show equivalence between the lazy and non-lazy flip models.
	Any algorithm for the non-lazy flip model with amortized time $\tu$ can be run on the lazy model, doing nothing when the step is `null',
	and the amortized update time can only decrease (and does, roughly by a factor of $2$). 
	Formally, any sequence of $s$ steps consists of $s' \leq s$ flips. 
	The total computation is at most $\tu \cdot s' \leq \tu \cdot s$, or at most $\tu$ amortized.
	For correctness, note that the distribution of each of the $s'$ changes is uniform over the edges, just as in the non-lazy model, leading to the exact same error probability. 
	
	For the other side, consider an algorithm $\alg{A}$ with update procedure $u(\cdot)$ for the lazy model.
	This algorithm might use the `null' steps to perform computation.
	We can instead use a new algorithm $\alg{A'}$ with the following update procedure $u'(\cdot)$, where a variable $counter$ was initialized to $0$ in pre-processing:
	\begin{itemize}
		\item If the step is `null', increase $counter$ by $1$.
		
		\item If the step is an edge $e$, run $u(\text{null})$ for $counter$ times, then $u(e)$, and set $counter = 0$.
	\end{itemize}
	Evidently, $\alg{A}$ and $\alg{A'}$ do the exact same computations over any lazy input sequence.
	
	Next, consider an algorithm $\alg{B}$ for the non-lazy flip model.
	The algorithm simulates $\alg{A'}$, with update operation $u_{\alg{B}}(e)$:
	\begin{itemize}
		\item Toss random fair coins until the first heads appears, use $t$ for the number of tosses. 
		
		\item Run $u'(e)$ with $couner = t-1$.
	\end{itemize}
	
	\paragraph{Amortized update time.} Consider an execution of $\alg{B}$ on a sequence of $s = \Omega(\log n)$ edge updates, and denote by $T$ the total amount of coin tosses made by $\alg{B}$ in
	$u_{\alg{B}}(e)$ operations.
	The random variable $T$ is simply the amount of fair coin tosses required for $s$ heads.
	By a standard Chernoff bound, since $s = \Omega(\log n)$, then $4s$ fair tosses have at least $s$ heads with probability at least $1 - 1/n$, which gives a lower bound on the probability that $T \leq 4s$.
	Note that each execution of $u'(\cdot)$ with $counter = t-1$ executes $u(\cdot)$ for $t$ times. Therefore, $\alg{B}$ calls $u(\cdot)$ exactly $T$ times over $s$ updates, and its amortized time is of at most $(\tu \cdot T)/s$, and w.h.p. $4 \tu$.
	
	\paragraph{Correctness.} Note that each of the $T$ executions of $u(\cdot)$ gets as input `null' with probability~$1/2$ and a change with probability~$1/2$, so the guaranteed error is kept, but for the added $1/n$ factor for the event $\set{T > 4s}$.
\end{proof}

\subsubsection*{Operation type in smoothed models.}
In both worst-case and average-case analysis the model can be equivalently defined with flip operations or add/remove operations. 
Perhaps surprisingly, it turns out this is not quite the case for the smoothed model.
First, we note that an adaptive adversary, which is fully aware of the current graph, can use either operations type (same as the worst-case adversary). 
This leaves us with a single adaptive smoothed model, where the adversary is at least as strong as the oblivious one (no matter the operation type used).

For oblivious adversaries, which choose their entire input ahead of the execution, we show below that the add/remove adversary is not weaker than the flip adversary 
(i.e., the complexity of any problem in the presence of an add/remove adversary is at least as high as with a flip adversary).
We later show the other direction is false. 
The intuition here
is that an oblivious flip adversary 
quickly loses track of the graph, whereas an oblivious add/remove adversary can control a small part of it by ``insisting'' on some edges begin added or removed.
Note that, by definition, whenever randomness controls a change, it does not draw an action with equal probability but instead always flips this edge. 
This is justified by \cref{claim:random flip vs ar} above: 
smoothing the operation type 
as well as the choice of the edge would have simply resulted in a lazy model that chooses a null operation with probability $p/2 \leq 1/2$, affecting the amortized times by a factor of at most $2$.

We now focus on $p$-smoothed models with oblivious adversaries. 
A strategy similar to the proof of \cref{claim:random flip vs ar} is applied here: 
we define a proxy lazy model, with laziness only over adversarial flip changes, 
and show it is equivalent (up to constants) to the non-lazy flip model.
However, in this setting an add/remove adversary is not forced to choose the operation at random (although it can), making this adversarial model potentially stronger. 
Our results from \cref{sec:upper-easy-with-flip} and \cref{sec:lb with embedding} show that this is indeed the case, separating the models.

Our first result shows that any lower bound shown using a flip oblivious adversary can be translated to a similar lower bound (up to constant factors) using an add/remove oblivious adversary.

\begin{claim}[Add/remove is stronger than flip for oblivious adversary in a $p$-smoothed model]
	\label{claim:oblivious ar vs flip}
    If there exists an algorithm $\alg{A}$ that solves a problem $\mathcal{P}$ with error $\delta$ over a $p$-smoothed sequence of $s$ steps with an oblivious add/remove adversary,
    then there exists an algorithm $\alg{B}$ that solves $\mathcal{P}$ with error probability at most $\delta + 1/n$ over a $p'$-smoothed sequence of $\Theta(s)$ steps with an oblivious flip adversary, for $p'= p/(2-p) \in [p/2,p]$, and provided that $s = \Omega(\log n)$.
\end{claim}

\begin{proof}
    As before, we use a proxy model we call the oblivious \emph{lazy}-flip $p$-smoothed model: every step in the sequence is `null' with probability $p/2$, flips an adversarially chosen edge with probability $p/2$, and flips a uniformly random edge with probability $1-p$.
    This is equivalent to giving the adversary probability $p$, but then with probability $1/2$ nullifying the step it chose.
    
    On the one hand, an oblivious add/remove adversary can simulate \emph{any} oblivious lazy-flip adversary. This is done by copying the strategy to choose $e$, but then choosing to add or remove it with probability $1/2$.
    Doing so, each step consists of a random flip with probability $1-p$, or the adversarial edge $e$ with probability $p$. If an adversarial edge is chosen, it is only flipped with probability $1/2$. This is the exact same distribution as the proxy lazy model above.  

    On the other hand, as before, the laziness can be dealt with, as formalized next.
    Assume an algorithm $\alg{A}$ with update procedure $u(\cdot)$ for the $p$-smoothed model with an oblivious AR adversary.
    In particular the same algorithms works in the lazy $p$-smoothed model with an oblivious flip adversary.
    As before, we define $\alg{A'}$ with update procedure $u'(\cdot)$ that uses a counter of null steps, and execute $u(\cdot)$ multiple time upon the next actual change.
    
    Consider an algorithm $\alg{B}$ for the non-lazy $p'$-smoothed model with flip adversary for $p'$ that we set later.
    The update procedure of $\alg{B}$ does the following on edge $e$:
    \begin{itemize}
        \item Toss random coins with bias $1- p/2$ towards heads, until the first heads appears. Call the number of coins used $t$.

        \item Run $u'(e)$ with plugged in value $counter = t-1$.
    \end{itemize}

    \paragraph{Amortized update time.} Note that the probability for head is $1-p/2 \geq 1/2$, and so for a sequence of length $s = \Omega(\log n)$ changes, the total amount of coin tosses $T$ is at most $4s$ with probability at least $1 - 1/n$.
    Hence, the total time for updates is at most $\tu \cdot T$, and the amortized update time is bounded by $\tu \cdot (T/s)$, and w.h.p. by $4\tu$.
    
    \paragraph{Correctness.} In the sequence of $T$ calls to $u(\cdot)$, each change is null with probability $p/2$, an adversarial flip with probability $p' (1-p/2)$, and a random flip otherwise.
    Choosing $p' = p/(2-p)$ (note that $p' \in [p/2,p]$), the probability for the adversarial flip is $p/2$, meaning $\alg{A}$ only errs with probability $\delta$. Adding the error for the event $\set{T > 4s}$, we get overall error $\delta + 1/n$.
\end{proof}

\section{Upper bounds}
\label{sec:upper}
In this section we present dynamic algorithms for the problems studied in this work. 
We start with proving that some problems become trivial in the presence of input randomness, 
then present algorithms for counting short paths, 
and finally use the latter algorithms to also count cycles.

We use $s$ and $t$ to denote two different fixed nodes (i.e., that are hard-coded in the algorithms), and $u$ and $v$ for arbitrary nodes, or nodes that can come as inputs.
All the paths and cycles we consider are simple. We use $m_0=|E(G_0)|$.

\subsection{Easy problems with an oblivious flip adversary}
\label{sec:upper-easy-with-flip}
This section deals with  properties that hold w.h.p.\ for random graphs, e.g., connectivity.
In essence, we show the randomness in the model guarantees that such properties hold w.h.p.\ at all times, which allows for a trivial algorithm that outputs ``yes''. 
While this is intuitive for $p=0$ (i.e., the average-case model of~\cite{HLS22}), we extend this intuition to almost any value of $p\in[0,1)$.  

Our first claim is applicable to monotone graph properties, and relies solely on the random changes occurred in the initial graph.

\begin{claim}
    \label{claim:oblivious_flip_contains_G_np}
    Fix a monotonely increasing graph property $\mathcal{P}$ and real numbers $\alpha \in(0,1]$ and $q \in [0,1/2]$, such that on an \er graph ~$G \sim G_{n,q}$,
    \[
        \Prob{G \sim G_{n,q}}{G \textnormal{ does not satisfy } \mathcal{P}} \leq \alpha .
    \]
    Then for every smoothing parameter $p \leq 1-2q$ and every round $r \geq 0$ in a $p$-smoothed dynamic graph in the oblivious flip model, it holds that
    \[
        \Prob{}{G_r \textnormal{ does not satisfy }\mathcal{P}} \leq \alpha .
    \]
    
\end{claim}

\begin{proof}
    We denote by $C_i$ the random string used for the smoothing process at round $i$, and by $C_0$ for the random string used to smooth the initial graph.
    Initially, the adversary proposes $H_0$, and the smoothing process chooses each node-pair w.p. $1-p$ and then flips it w.p. $1/2$ (equivalent to re-sampling w.p. $1/2$), independently of other edges.
    That is, the randomness $C_0$ samples a graph $F_0 \sim G_{n,(1-p)/2}$ of flips, and changes the initial graph to $G_0 = H_0 \oplus F_0$.

    Next, we fix the (oblivious flip) adversarial strategy. The smoothed sequence of changes $e_1,\dots e_r$ applied on an empty graph results a (random) graph that only depends on $C_1, \dots, C_r$. Denoting this graph by $G'$, we have that 
    \[
        G_r = G_0 \oplus G' = F_0 \oplus \left(H_0 \oplus G'\right) .
    \]
    The crucial point is that the adversarial strategy is chosen in advance, and therefore it does not depend on $F_0$. Furthermore, the random strings $C_1,\dots, C_r$ are independent of $C_0$ and therefore of $F_0$.
    This means the distribution of $G_r$ can be described as follows: first draw $C_1,\dots, C_r$ which combined with the adversarial strategy gives $H_0 \oplus G'$, and then XOR this graph with $F_0 \sim G_{n,(1-p)/2}$.

    For any outcome of $C_1, \dots, C_r$, the graph $H_0 \oplus G'$ is fixed and $G_r$ then depends only on $F_0$: every node-pair $e$ is an edge in $G_r$ independently of the others, w.p. $(1-p)/2$ if $e \notin H_0 \oplus G'$ and w.p. $(1+p)/2$ if $e \in H_0 \oplus G'$.
    Using the inequalities $(1+p)/2 \geq (1-p)/2 \geq q$ and the monotonicity of $\mathcal{P}$, we have the following: 
    \[
    \begin{aligned}
        \Prob{C_0,\dots,C_r}{G_r \textnormal{ does not satisfy }\mathcal{P}}
        &= \Expc{C_1,\dots,C_r}{\Prob{C_0}{G_r \textnormal{ does not satisfy }\mathcal{P}}}\\
        &\leq \Expc{C_1,\dots,C_r}{\Prob{C_0}{F_0 \textnormal{ does not satisfy }\mathcal{P}}}\\
       &= \Prob{C_0}{F_0 \in \mathcal{P}}\\
        &\leq \Prob{G\sim G_{n,q}}{G \textnormal{ does not satisfy }\mathcal{P}} \\
        &\leq \alpha. &&\qedhere
    \end{aligned}
    \]
\end{proof}

We can leverage the claim to get the following result.
\begin{claim}
    \label{claim:ub_oblivious_flip}
    Fix $C > 2$. For any parameter $p\in[0,1-\frac{2C\log n}{n})$, the problems of 
     \textbf{connectivity, bipartite perfect matching, bipartite maximum matching} and \textbf{bipartite minimum vertex cover}
     admit constant update and query time algorithms with an oblivious flip adversary, with probability of at least $1-1/n$ to succeed through all $T$ rounds, provided that $T \leq n^{^{\frac{C-3}{2}}}$ rounds.
\end{claim}

\begin{proof}
    The premise of \Cref{claim:oblivious_flip_contains_G_np} holds for connectivity with $q = \frac{C\log n}{n}$ and $\alpha = 2n^{-(C-1)/2}$, for any $C > 2$.
    A similar result holds for bipartite perfect matching with $q = \frac{C\log n}{n}$ and $\alpha = n^{-(C-1)}$, for any $C > 2$ (see~\cite[Remark 4.3]{books:random_graphs_janson2011}).

    Thus, for any $p\in[0,1-\frac{2C\log n}{n})$, a $p$-smoothed dynamic graph is disconnected at round $r$ for any $r\in[n^{(C-3)/2}]$ with probability at most $2n^{-(C-1)/2}$ (and similarly for a bipartite graph having a perfect matching at all times). This allows for a trivial constant-time algorithm that does nothing during updates and outputs ``yes'' upon any query.    
    
    Using a union bound over $T \leq n^{^{\frac{C-3}{2}}}$ rounds, the algorithm succeeds in \emph{all} rounds with probability at least $1 - 1/(2n)$.
    
    For maximum matching, the existence of a perfect matching immediately implies a matching on all nodes, i.e., of size $n$ (for a bipartite graph with $n$ nodes in each side).
    Hence, an algorithm for maximum matching can answer $n$ to all queries.
    
    Similarly, picking all the nodes of one side gives a vertex cover of size $n$, while any cover must touch all the edges of the perfect matching described above and thus must have size at least $n$.
    An algorithm for minimum vertex cover can thus answer $n$ to all queries as well.
\end{proof}

Before we turn to another argument, we digest the meaning of the claim above in the context of comparison between models.
\begin{description}
    \item[Oblivious add/remove adversary.] 
    The proof of \Cref{claim:oblivious_flip_contains_G_np} would indeed fail if the adversary uses add/remove operations. This is for a good reason: in \Cref{thm:lb for many problems-smoothed oblivious ar} we show that such an adversary can disrupt any fast algorithm with. e.g., $p = 1 - 1/\sqrt{n}$. 
    This in fact separates the two models.

    \item[Average-case model.] 
    Looking at \Cref{claim:oblivious_flip_contains_G_np}, one might wonder if the $p$-smoothed oblivious flip adversary can pose any difficulty or is it simply equivalent to not having no adversary at all (the average-case, $p=0$).
    In~\cref{sec:lb small subgraphs}, we show the two models are separated for several subgraph counting problems. We prove a polynomial lower bound for the update time in the $p$-smoothed oblivious flip model whenever $p = \Omega(n^{c})$ (for fixed $c> 0$), which stands in contrast to the upper bound of $O(1)$ for the average case shown by~\cite{HLS22}.
\end{description}

The argument in \Cref{claim:oblivious_flip_contains_G_np} relies only on the randomness injected to the initial graph, and for~$p$ close to $1$ it can only be useful for a polynomial number of rounds (hence the restriction on $T$).
This restriction, however, can be lifted. Indeed, more randomness is being injected to the model with time, and we next show that once $\poly(n)/(1-p)$ rounds have passed, each graph exhibits a ``typical'' behavior of a \emph{uniformly random graph}.
This is true even though each such graph might be statistically far from an actual uniformly random graph (to see this, consider the number of edges mod $2$ as the distinguisher). 

\begin{claim}
    \label{claim:oblivious_flip_looks_random}
    Fix a graph property $\mathcal{P}$ and a constant $\alpha>0$ such that on a uniformly random graph~$G \sim G_{n,\frac{1}{2}}$,
    \[
        \Prob{}{G \textnormal{ does not satisfy } \mathcal{P}} \leq \alpha .
    \]
    Then for every smoothing parameter $p < 1$ and $r_p := n\cdot \binom{n}{2} / (1-p)$, every round $r \geq r_p$ in a $p$-smoothed dynamic graph in the oblivious flip model satisfies
    \[
        \Prob{}{G_r \textnormal{ does not satisfy }\mathcal{P}} \leq 4r\cdot \left(\alpha + \frac{n^2}{e^{2n}}\right) .
    \]
\end{claim}

A key element in the proof is a stochastic graph $G_{\ell}$.
First, we pick $\ell = \poisson{r}$, and consider the graph after $\ell$ rounds, $G_{\ell}$
(which is a random variable depending both on the smoothing process and the choice of~$\ell$).
On the one hand, if $r \geq r_p$ then $G_{\ell}$ has negligible statistical distance from a uniformly random graph ($G_{n,\frac{1}{2}}$).
On the other hand, $\ell = r$ with probability at least $1/\poly(r)$, so events with an exponentially negligible probability in $G_{\ell}$ must have an exponentially negligible probability in $G_r$ as well.
Combining both parts shows that any event with a negligible probability in a uniformly random graph (e.g., being disconnected),
also have a negligible probability in $G_r$.

\begin{proof}
    Fix any sequence of changes (adversarial strategy) $e_1,\dots, e_t$ and a round~$r$, and focus on the first $\ell = \poisson{r}$ steps.
    We shall analyze the distribution of the stochastic graph $G_{\ell}$, where randomness comes from the choice of $\ell$ and the smoothing process.

    Our first goal is to show 
    that $G_\ell$ is distributed almost uniformly, that is, 
    \begin{equation}
        \label{eq:poison_round_is_almost_uniform}
        \sd{G_{\ell}}{G_{n,\frac{1}{2}}} \leq \frac{n^2}{e^{2n}} .
    \end{equation}
    Using the Poissonization technique, we construct $G_{\ell}$ by independently drawing the number of random changes chosen by nature, $\poisson{(1-p)\cdot r}$, and the number of adversarial changes, $\poisson{p\cdot r}$.
    Next, we draw a random permutation $\sigma$ on their order, choosing \emph{which} rounds use random changes and which follow the adversarial choice.
    Finally, each random change draws a uniformly random edge to flip.
    
    As the adversarial strategy is fixed, any permutation $\sigma$ entirely fixes the set of adversarial changes made throughout the entire process.
    We first focus on the initial graph $G_0$ and apply only these adversarial flips, 
    obtaining a stochastic graph $G'_{\ell} = (V,E'_{\ell})$ whose randomness only comes from $G_0$ and $\sigma$. 
    
    Next we focus on the edges chosen by the randomness. Each edge is drawn with probability $1/\binom{n}{2}$, and so we can use Poissonization for the set of $\poisson{(1-p)\cdot r}$ random changes. Indeed, this draw is equivalent to drawing independently the number of randomness-chosen flips for each edge $z_e \sim \poisson{(1-p)\cdot r / \binom{n}{2}}$.
    If we start from an empty graph and make these changes, we end up with a stochastic graph $G''_{\ell} = (V,E''_{\ell})$ that represents all random flips made throughout the process. 

    The order in which edges are flipped does not affect the resulting graph $G_{\ell} = (v,E_{\ell})$. Thus, starting from the initial graph $G_0$ and applying both adversarial and random flips, we can write
    \[
        E_{\ell} = E'_{\ell} \symdif E''_{\ell} ,
    \]
    where both sides are random variables according to the smoothing process and the choice of $\ell$.

    We are now ready to analyze $G_{\ell}$ via $G''_{\ell}$. Choose $r_p = n \binom{n}{2} / (1-p)$, and consider any round $r \geq r_p$.
    The process that produces $G''_{\ell}$ flips each potential edge in the graph $z_e \sim \poisson{\lambda_r}$ times, where $\lambda_r = (1-p)\cdot r / \binom{n}{2}$. For any round $r \geq r_p$ we have $\lambda_r \geq \lambda_{r_p} = n$. By \cref{fact:item:poisson_parity}~in~\cref{fact:poisson_properties}, we know the parity of $z_e$ (and the existence or inexistence of the edge $e$) is almost uniform:
    \[
        \Prob{}{z_e \sameparity 0}
        = \frac{1 + e^{-2\lambda_r}}{2} 
        \in \left[ \frac{1}{2}\ ,\  \frac{1 + e^{-2n}}{2}\right]
    \]
    Denoting the parity of $z_e$ by $\ind{e}$, we get an indicator for the existence edge $e$ in $G''_{\ell}$. 
    We also use $U_1$ for a uniformly random bit, thus the previous inequality can be written in terms of statistical distance:
    \[
        \sd{\ind{e}}{U_1} \leq e^{-2n} .
    \]
    Due to Poissonization, all $z_e$ are independent from one another, and so are the indicators $\ind{e}$, implying the distribution of $G''_{\ell}$ is a product of all $\ind{e}$. Similarly, a uniformly random graph $G_{n,\frac{1}{2}}$ is a product of $U_1$ for each potential edge.
    By sub-additivity of statistical distance for product measures (\cref{fact:item:SD_subadditivity}~in~\cref{fact:SD_properties}), we get:
    \[
        \sd{G''_{\ell}}{G_{n,\frac{1}{2}}}
        \leq \sum_{e}\sd{\ind{e}}{U_1}
        \leq \frac{n^2}{e^{2n}} .
    \]
    Lastly, recall that $G_{\ell}$ is made by flipping the edges of $G'_{\ell}$ in $G''_{\ell}$. That is, any fixed permutation $\sigma$ sets the adversarial flips and in turn sets a bijection between $G_{\ell}$ and $G''_{\ell}$. This bijection only only re-orders the probability vector, attaching each probability to a new graph instead of the old one (in a bijective manner). For this reason, it does not affect the statistical distance from a uniformly random graph (where all probabilities are the same). Finally, averaging over all permutations $\sigma$ we obtain  
    \[
        \sd{G_{\ell}}{G_{n,\frac{1}{2}}} \leq \frac{n^2}{e^{2n}} ,
    \]
    which proves \eqref{eq:poison_round_is_almost_uniform}.
    
    Since the two graphs are statistically close,  \cref{fact:item:SD_error_difference}~in~\cref{fact:SD_properties} implies
    \begin{equation}
        \label{eq:poisson_round_probability}
        \Prob{}{G_{\ell} \textnormal{ does not satisfy } \mathcal{P}} 
        \leq \Prob{}{G_{n,\frac{1}{2}} \textnormal{ does not satisfy } \mathcal{P}} + \sd{G_{\ell}}{G_{n,\frac{1}{2}}}
        \leq \alpha + \frac{n^2}{e^{2n}} .
    \end{equation}

    Finally, to connect $G_{\ell}$ and $G_r$, recall that $\ell$ is distributed according to $\poisson{r}$, and so it has a rather high probability of choosing exactly round $r$:
    \[
        \Prob{}{\poisson{r} = r} = \frac{r^re^{-r}}{r!} 
        \geq \frac{1}{2\sqrt{\pi r}}
        \geq \frac{1}{4r} ,
    \]
    with the first inequality deriving from Stirling's approximation.

    This means that $G_{\ell}$ is somewhat likely to choose the graph $G_r$. By the total law of probability:
    \[
    \begin{aligned}
        \Prob{}{G_{\ell} \textnormal{ does not satisfy } \mathcal{P}} 
        &= \sum_{i\in\N} \Prob{}{\ell = i} \cdot \Prob{}{G_{i} \textnormal{ does not satisfy } \mathcal{P}} \\ 
        &\geq \Prob{}{\ell = r} \cdot \Prob{}{G_{r} \textnormal{ does not satisfy } \mathcal{P}}\\
        &\geq \frac{1}{4r} \cdot \Prob{}{G_{r} \textnormal{ does not satisfy } \mathcal{P}} .
    \end{aligned}
    \]
        
    Combining this inequality with \cref{eq:poisson_round_probability}, the proof is concluded.    
\end{proof}

Using the above claim we can guarantee a $p$-smoothed dynamic graph stays connected w.h.p.\ even over a near-exponential number of rounds.
As mentioned above, $G_{n,q}$ satisfies some nice properties w.h.p.\ even for $q = o(1)$ (see~\cite{books:random_graphs_janson2011}).
For a uniformly random graph ($q=1/2$) we quantify the negligible probability that these properties don't hold, following similar arguments as in~\cite{ER59,ER64}.

\begin{claim}
    \label{claim:negligible_probabilities}
    There exists $N_0$ such that the following hold for any $n \geq N_0$: 
    \begin{enumerate}
        \item The probability of an \er graph over $n$ nodes w.p.~$1/2$ for each edge ($G_{n,\frac 1 2}$) to be disconnected is at most~$\frac{2n}{2^{n/2}}$.

        \item The probability of a random bipartite graph with $n$ nodes on each side and each edge w.p.~$1/2$ to have no perfect matching is at most $\frac{4(n+1)^2}{2^{(n+1)/2}}$.
    \end{enumerate}
\end{claim}

\begin{proof}
    Any disconnected graph must have a connected component of size $k \leq n/2$ nodes. 
    The probability of a set of $k$ nodes to indeed be disconnected from the rest of the graph $G_{n,\frac{1}{2}}$ is exactly $2^{-k(n-k)} \leq 2^{-\frac{nk}{2}}$.

    We can union bound over all sets of size $k$ and over all sizes $k$. The probability of a graph being disconnected is thus at most
    \[
        \sum_{k=1}^{\ceil{n/2}} \binom{n}{k} \cdot \left(\frac{1}{2}\right)^{nk/2}
        \leq \sum_{k=1}^{\ceil{n/2}} n^k \cdot \left(\frac{1}{2^{n/2}}\right)^{k}
        = \sum_{k=1}^{\ceil{n/2}} \left(\frac{n}{2^{n/2}}\right)^k 
        \leq \frac{2n}{2^{n/2}} ,
    \]
    where the last inequality is since the sum of an infinite geometric series is at most twice its first element (assuming $n/2^{n/2} \leq 1/2$ which holds for $n\geq 8$).

    For the second bullet, consider a graph with $n$ nodes on each side, call the sides $A, B$. 
    Any graph with no perfect matching must violate Hall's condition, meaning there exists $S \subseteq A$ of size $\card{S} = k$ with neighbor set $\card{\Gamma(S)} < k$.
    Differently put, there exists $S\subseteq A$, $\card{S} = k$, and a set $T\subseteq B$, $\card{T} = n-k+1$ such that all edges between $S$ and $T$ are missing.
    We union bound again over all possible pairs $S, T$ with sizes $k, n-k+1$, and over all values of $k$ to find an upper bound for the probability that Hall's condition is violated:
\[    \begin{aligned}
        \sum_{k=1}^{n}\binom{n}{k} \binom{n}{n-k+1} \left(\frac{1}{2}\right)^{k(n-k+1)}
        &\leq \sum_{k=1}^{n+1}\binom{n+1}{k} \binom{n+1}{n-k+1} \left(\frac{1}{2}\right)^{k(n-k+1)}\\
        &\leq 2\sum_{k=1}^{\ceil{(n+1)/2}}\binom{n+1}{k} \binom{n+1}{n-k+1} \left(\frac{1}{2}\right)^{k(n-k+1)}\\
        &\leq 2\sum_{k=1}^{\ceil{(n+1)/2}}\binom{n+1}{k}^2 \left(\frac{1}{2}\right)^{k(n+1)/2}\\
        &\leq 2\sum_{k=1}^{\ceil{(n+1)/2}} \left(\frac{(n+1)^2}{2^{(n+1)/2}}\right)^k\\
        &\leq \frac{4 (n+1)^2}{2^{(n+1)/2}} .
    \end{aligned}
    \]
    In the second inequality we used symmetry of the addends (potentially adding twice the addend for $k=(n+1)/2$ when $n$ is even). The two last inequalities are similar to those in the previous bullet, using $(n+1)^2/2^{(n+1)/2} \leq 1/2$ for $n \geq 20$. This proves the claim with $N_0 = 20$.
\end{proof}

While \cref{claim:ub_oblivious_flip}
is limited to $p\in[0,1-\frac{2C\log n}{n})$,
\cref{claim:oblivious_flip_looks_random} allows us to handle the regime $1-p = o(\log n / n)$, where very few perturbations are added to the initial graph, leaving it close to the worst-case.
The trick is to execute a known algorithm for the worst-case during the first $\approx n^3/(1-p)$ rounds, and only then rely on the randomness to kick in. Over a long enough sequence, the amortized running times will become similar to those of the average-case.

\begin{lemma}
    \label{lem:ub_oblivious_flip2}
    Fix $p \in [0,1-n^{-c}]$ for some $c > 0$. For the following problems, there exists an algorithm with all amortized running times~$O(1)$ on any sequence of polynomial length $\Omega\left(n^6/(1-p)\right)$ which succeeds w.h.p.\ on $p$-smoothed dynamic graphs with an oblivious flip adversary: \textbf{bipartite perfect matching, bipartite maximum matching} and \textbf{bipartite minimum vertex cover}.
    The same holds for \textbf{connectivity} with a sequence of polynomial length at least $\tilde{\Omega}\left(n^3/(1-p)\right)$.
\end{lemma}

\begin{proof}
    For all problems but connectivity, until round $r_p = n\binom{n}{2}/(1-p)$, we run a known algorithm that works in the worst-case. All problems admit such algorithms with $O(n^3)$ amortized running time (by solving each round separately). 
    For connectivity, 
    we use the algorithm of~\cite{HuangHKPT23} that has $\tilde O(\log n)$ update time and $\tilde O(\log n)$ query time.

    For connectivity and bipartite perfect matching, starting round $r_p$ the algorithm simply answers ``yes'' on all queries.
    
    Combining \cref{claim:oblivious_flip_looks_random} with \cref{claim:negligible_probabilities}, we know that for all steps $r_p \leq r \leq T$ we have for connectivity
    \[
        \Prob{}{G_r \textnormal{ is not connected}}
        \leq 4r\left( \frac{2n}{2^{n/2}} + \frac{n^2}{e^{2n}}\right)
        \leq \frac{2T \cdot n^2}{2^{n/2}} .
    \]
    Similarly, for bipartite perfect matching:
    \[
        \Prob{}{G_r \textnormal{ does not have perfect a matching}}
        \leq 4r\left( \frac{4(n+1)^2}{2^{(n+1)/2}} + \frac{n^2}{e^{2n}}\right)
        \leq \frac{8T \cdot (n+1)^2}{{2^{(n+1)/2}}} .
    \]
	
	That is, for both problems, the failure probability of the ``yes'' algorithm on each round is exponentially small. Taking a union bound over at most $T$ rounds, and as since $T = \poly(n)$, 
	we get that the ``yes'' algorithm succeeds on all rounds between $r_p$ and $T$ with probability $1 - o(1)$.
    
    For the complexity, recall that $T \geq n^6/(1-p)$ which implies $T \geq n^3 \cdot r_p$.
    The total running time of the first $r_p$ rounds is $n^3 \cdot r_p$, and the rest trivially answers ``yes'', taking $O(T - r_p)$ total computation. Over all $T$ rounds, this amortizes to $O(1)$, meaning constant amortized update and query times.
    
    As in the proof of \cref{claim:ub_oblivious_flip}, an algorithm for maximum matching can answer $n$ on all queries after round $r_p$,
    and algorithm for minimum vertex cover can answer $n$ on all queries.
\end{proof}

Finally, for a slightly more modest range for $p$ we improve upon \cref{claim:ub_oblivious_flip}:
\begin{lemma}
\label{lem:ub_oblivious_flip}
    For any $p\in[0,1-\frac{26\log n}{n})$, 
    there are algorithms for 
    \textbf{connectivity, bipartite perfect matching, bipartite maximum matching} and \textbf{bipartite minimum vertex cover}
    with constant update and query times with an oblivious adversary, with success probability at least $1-1/n$ through all $T$ rounds, provided that $T \leq 2^{n/3}$ rounds.
\end{lemma}

\begin{proof}
    For connectivity we prove that the ``yes'' algorithm is correct on all $T$ rounds with probability at least $1-1/n$. 
    The argument for bipartite perfect matching is almost identical,
    and the extension to maximum matching 
    and minimum vertex cover is as in the proof of \cref{claim:ub_oblivious_flip}.
    
    Using \cref{claim:oblivious_flip_contains_G_np} with $q = \frac{13\log n}{n}$ and $\alpha = 2n^{-6}$
    for all rounds $r = 0 ,\dots, n^4$ the graph $G_r$ is disconnected with probability at most $2n^{-6}$, and by a union bound over all rounds the algorithm accumulate an error probability of at most $1/(2n)$.
    We next apply \cref{claim:oblivious_flip_looks_random}, noting that $1-p \geq 1/n$ and therefore $r_p \leq n^4$. We thus apply the claim for all rounds $r=n^4+1, \dots, T$, combined with \cref{claim:negligible_probabilities}. This shows that each graph $G_r$ in these rounds is disconnected with probability at most $\poly(n) / 2^{n/2} \leq 1/(2n 2^{n/3})$, and by a union bound over all rounds, the algorithm accumulates an error probability of at most $1/(2n)$.
    Thus, the algorithm errs on any round between $1$ and $T$ with probability at most $1/n$.
\end{proof}

\subsection{Counting small subgraphs with any adversary}
\subsubsection{Counting short \st paths}
In this section we show that it is easy to count \st paths of length $1,2,3$ or $4$ in smoothed dynamic graphs,
with any adversary type.
Our algorithms in this section and the following one are similar to the average-case algorithms of~\cite{HLS22}, 
with a faster preprocessing procedure, as well as careful treatment of non-simple paths (such as \paths st4 of the form $(s,u,v,s,t)$) that
seem to have been previously overlooked.
Roughly speaking,
most algorithms have ``cheap'' and ``expensive'' changes, 
where the latter rarely happen at random but might be consistently chosen by an adversary.
We use an amortization argument
to bound the running times in the presence of an adversary, pinpointing the dependency on $p$.

The update time of our algorithms is constant for any $p = O(1/n)$, and then gradually increases with $p$, reaching $O(n)$ when $p$ is an arbitrarily small constant.
In particular, this coincides with previous bounds for the worst case ($p=1$) and the average case ($p=0$).
For the average case, the initialization time is reduced to $O(m_0)$, where $m_0$ is the number of edges in the initial graph,
compared to previous algorithms~\cite{HLS22} that have initialization times of $O(n^2)$ for \paths st3 and $s$~-triangles, and $O(n^\omega)$ for \paths st4,  and $s$~$4$-cycles.

Counting \st paths of length $1$ or $2$ is trivial even in the worst case~\cite{HLS22}, and hence also in smoothed and random dynamic graphs.
In order to count \st $3$-paths, we use an auxiliary dynamic algorithm that counts \paths{s}{u}{2}, for a fixed node $s$ and a node $u$ given as an input in each query.

\begin{lemma}
	\label{lem: s2u alg}
	There is an algorithm for $p$-smoothed dynamic graphs with any adversary type
	that has fixed nodes $s$ and $t$, receives a node $u$ in each query, 
	and has
	$O(m_0)$ preprocessing time, 
	$O(pn+1)$ expected update time, 
	and can answer a query with the number of $s$-$u$ $2$-paths excluding the edge~$(s,t)$
	in $O(1)$ time.
	In addition, the algorithm maintains a counter $c_{s2u}$ of such $s$-$u$ $2$-paths for each node $u$ in the graph.
	The algorithm can also be used without specifying a node $t$, in which case no edge is excluded.
\end{lemma}

The expected running can be thought of as the amortized running time, as for a reasonable amount of changes the total update time would concentrate around its expectation.

\begin{proof}
	The preprocessing is a simple BFS from $s$ up to depth $2$, ignoring the edge $(s,t)$,
	in $O(m_0)$ time.
	
	Upon an insertion or a removal of an edge $(v,u)$ (where $u,v\neq s$), 
	if $v\neq t$ then the algorithm checks the existence of the edge $(s,v)$ and update $c_{s2u}$ accordingly (increments $c_{s2u}$ if $(v,u)$ was added and $(s,v)$ exists, 
	decrements $c_{s2u}$ if $(v,u)$ was removed and $(s,v)$ exists, and ignores the change otherwise).
	If $u\neq t$, the algorithm similarly checks the existence of $(s,u)$ and updates $c_{s2v}$.
	
	Upon an update of an $(s,v)$ edge, $v\neq t$, the algorithm checks for all nodes $u\neq s$ if the edge $(v,u)$ exists and update $c_{s2u}$ accordingly, which takes $O(n)$ time.
	A random update will update an $(s,v)$ edge with probability $1/n$; an adversarial change, on the other hand, might perform such an update at each time, 
	and an adversarial change happens w.p.~$p$.
	
	Upon a query with a node $u$, the algorithm naturally returns $c_{s2u}$.
	Note that $c_{s2s}=0$ at all times.
	
	To conclude, for each update w.p.~$1-p$ it takes $O(1)$ time, and w.p.~$p$ it takes $O(n)$, for a total of $O(pn+1)$ expected time.
\end{proof}

With this algorithm in hand, counting \st $3$-paths is an easy task.
\begin{lemma}
\label{lem:counting st 3paths}
	There is an algorithm for $p$-smoothed dynamic graphs with any adversary type
	that counts the number of \st $3$-paths with preprocessing time $O(m_0)$, 
	update time $O(pn+1)$ in expectation, 
	and query time $O(1)$.
\end{lemma}
Note that we only consider simple paths, so the edge $(s,t)$ can never participate in an \st $3$-path.

\begin{proof}
	The algorithm uses the algorithm for maintaining the counter $c_{s2u}$ of $s$-$u$ $2$-paths from Lemma~\ref{lem: s2u alg}.
	In addition, it maintains a counter $c$ of \st $3$-paths, initialized to $0$.
	
	Upon initialization, the algorithm initializes the algorithm from Lemma~\ref{lem: s2u alg}, and while doing so it lists all the nodes $u$ with $c_{s2u}\neq 0$.
	Then, for each node $u$ in the list, if the edge $(u,t)$ exists, it increases $c$ by $c_{s2u}$.
	Since each \st $3$-path must be of the form $(s,v,u,t)$ with $u,v\neq s,t$, 
	the counter $c$ now contains exactly the number of \st $3$-paths.
	
	Upon an update of an edge $e$ not touching $t$, the algorithm updates the $c_{s2u}$ counters, and every time it changes a  counter $c_{s2u}$, it checks if the edge $(u,t)$ exists.
	If so, it updates $c$ by the change of $c_{s2u}$ (which can be positive or negative).
	Upon an update of an edge $(u,t)$, the algorithm updates~$c$ by $c_{s2u}$.
	
	For a query, the algorithm naturally returns $c$.
	
	Regarding the complexity, 
	note that the preprocessing phase only adds $O(1)$ operations for each node $u$ that was updated by the preprocessing phase of the algorithm from \cref{lem: s2u alg}, and thus they have the same asymptotic preprocessing time.
	Similarly, upon an update not involving~$t$, 
	the algorithm only performs $O(1)$ additional operations for each operation of the algorithm from \cref{lem: s2u alg},
	and does only $O(1)$ operations upon an update involving $t$, thus again the algorithms have the same asymptotic update time.
\end{proof}

Finally, we present an algorithm that maintains the number of \st $4$-paths.

\begin{lemma}
\label{lem:counting st 4paths}
	There is an algorithm for $p$-smoothed dynamic graphs 
	with any adversary type
	that counts the number of \st $4$-paths with preprocessing time $O(m_0)$, 
	update time $O(pn+1)$ in expectation, 
	and query time $O(1)$.
\end{lemma}

\begin{proof}
	The algorithm uses the algorithm from \cref{lem: s2u alg} twice, in order to keep track of the number $c_{s2u}$ of $s$-$u$ $2$-paths and the number $c_{t2u}$ of $t$-$u$ $2$-paths (in both cases, excluding $(s,t)$).
	In addition, it maintains a counter $c$ of \st $4$-paths, initialized to $0$.
	
	In the preprocessing, the algorithm first performs preprocessing of the two $2$-path counters, and while doing so keeps a list of all nodes $u$ with $c_{s2u}>0$.
	It then goes over this list, and for each $u$ with $c_{s2u}>0$ it adds $c_{s2u}c_{t2u}$ to $c$, thus counting all \st $4$-paths but including degenerate paths of the form $(s,v,u,v,t)$.
	Finally, it goes over all neighbors $v$ of $s$, and for each such $v$ that is also a neighbor of $t$, it subtracts $\deg(v)-2$ from $c$.
	This is done in order to account for all $(s,v,u,v,t)$ paths $v$ is a part of.
	In total, the preprocessing takes no more than the $O(m_0)$ preprocessing time of the $2$-path databases, and the last part also takes $O(\deg(s))=O(m_0)$ time.
	
	Upon an update of an edge $(u,v)$ (with $\{u,v\}\cap\{s,t\}=\emptyset$)
	the algorithm preforms 4 types of updates, corresponding to 4 types of \st $4$-paths through $(u,v)$, as follows.
	All the updates are done before updating the $2$-path counters.
	
	To account for paths of type $(s,u,v,x,t)$, $x\notin\{s,t,u,v\}$, the algorithm first checks that the edge $(s,u)$ exists. If so,
	if $(u,v)$ is \emph{added}, the algorithm adds $c_{t2v}$ to $c$;
	if $(u,v)$ is \emph{removed} and the edge $(v,t)$ does not exist, the algorithm subtracts $c_{t2v}$ from $c$;
	finally, if $(u,v)$ is \emph{removed} and the edge $(v,t)$ exists, the algorithm subtracts $c_{t2v}-1$ from $c$.
	The reason for this subtle difference is that in the first two cases we do not have to account for a degenerate path of the form $(s,u,v,u,t)$ since the 2-path $(v,u,t)$ is not counted in $c_{t2v}$ in these cases, 
	while in the last case the degenerate $4$-path $(s,u,v,u,t)$ is not counted in $c$ while the $2$-path $(v,u,t)$ is counted in $c_{t2v}$ so we should not subtract $1$ from $c$ for this $2$-path.
	
	The other 3 types of updates account for paths of types $(s,v,u,x,t)$, $(s,x,u,v,t)$, and $(s,x,v,u,t)$, and are completely analogous. All these updates take constant time.
	
	Upon an update of an edge $(s,u)$, the algorithm must go through all neighbors $v\neq s$ of $u$, and for each such neighbor update $c$ by $c_{t2v}$ if the edge $(u,t)$ does not exist in the graph, and update it by $c_{t2v}-1$ otherwise.
	This counts all paths of the form $(s,u,v,x,t)$ changed by the update of~$(s,u)$, where the last case makes sure degenerate paths are not counted.
	An update of an edge of the form $(t,v)$ is analogous.
	In both cases, the $2$-path counters are updated after $c$ is update, and the time complexity of these updates is~$O(n)$.
\end{proof}

Observe that both in the above proof and in the proof of \cref{lem: s2u alg}, any update takes $O(n)$ time even in the worst case;
in fact, this can already be observed by considering the proof of~\cite{HLS22} in detail.
This results in a deterministic worst-case algorithm with $O(n)$ (non-amortized) update time and constant query time.
The state of the art for the problem is $O(n^2\log^3n)$ update time and $O(\#_4P+\log n)$ query time, where $\#_4P$ is the output value.
This is achieved using a technique of~\cite{DemetrescuI04} for maintaining a set of shortest paths between any pair of nodes.

\begin{corollary}
\label{cor:counting st 4paths worst-case}
	There is a deterministic worst-case algorithm
	that counts the number of \st $4$-paths with preprocessing time $O(m_0)$, 
	update time $O(n)$, 
	and query time $O(1)$.  
\end{corollary}

\subsubsection{Counting short cycles through a node}

We now move from counting paths to counting cycles, and start by counting triangles ($3$-cycles) through a fixed node $s$.

\begin{lemma}
\label{lem:counting s triangles}
	There is an algorithm for $p$-smoothed dynamic graphs 
	with any adversary type
	that counts the number of triangles through $s$ with preprocessing time $O(m_0)$, 
	update time $O(pn+1)$ in expectation, 
	and query time $O(1)$.
\end{lemma}

\begin{proof}
	We use the algorithm from \cref{lem: s2u alg} in order to keep track of the number $c_{s2u}$ of $s$-$u$ $2$-paths.
	In addition, it maintains a counter $c$ of \emph{twice} the number of triangles through $s$, initialized to $0$.

	The algorithm first initializes all the $c_{s2u}$ counters and keeps a list of the nodes $u$ with $c_{s2u}>0$.
	It then goes over this list and for each such $u$, if $u$ is a neighbor of $s$, it adds $c_{s2u}$ to $c$.
	
	Upon an update, the algorithm keeps track of all updates of $c_{s2u}$ variables, and for each such update, if $u$ is a neighbor of $s$, it updates $c$ by the same amount.
	
	Note that each triangle $(s,u,v,s)$ through $s$ is counted exactly twice, both in the initialization and during the algorithm: such a triangle is counted once as $(s,u,v,s)$  when updating $c_{s2v}$, and once $(s,v,u,s)$ when updating $c_{s2u}$.
    Accordingly, the algorithm answers each query with $c/2$.
	The overhead over each operation of the algorithm from \cref{lem: s2u alg}
	is constant, 
    giving the desired preprocessing, update and query times.
\end{proof}

Next, we give an algorithm for counting $4$-cycles through a fixed node $s$.
While it is tempting to use the \st $4$-path counting algorithm with $s=t$ for this goal, as suggested in~\cite{HLS22},
the situation is a bit more subtle:
the $4$-path algorithm first counts also degenerate paths (e.g.\ of the form $(s,u,v,u,t)$) 
and then deduce them from the total count, and this deduction step assumes $s\neq t$.
In addition, when $s=t$ it will count each cycle twice, once at each direction, but this is easily fixable by dividing the output by $2$.
A simpler approach is maintaining two copies of $s$ in the graph, denoted $s$ and $s'$, and counting $s$-$s'$ $4$-paths.
We suggest yet another approach, that does not use the $4$-path algorithm at all.

\begin{lemma}
\label{lem:counting s 4cycles}
	There is an algorithm for $p$-smoothed dynamic graphs 
	with any adversary type
	that counts the number of $4$-cycles through a given node $s$ with preprocessing time $O(m_0)$, 
	update time $O(pn+1)$ in expectation, 
	and query time $O(1)$.
\end{lemma}

\begin{proof}
	The algorithm uses the algorithm from \cref{lem: s2u alg} in order to keep track of the number $c_{s2u}$ of $s$-$u$ $2$-paths (without a node $t$ and an excluded edge $(s,t)$).
	In addition, it maintains a counter $c$ of $4$-cycles through $s$, initialized to $0$.
	
	The algorithm first initializes all the $c_{s2u}$ counters and keeps a list of the nodes $u$ with $c_{s2u}>0$.
	It then goes over this list and for each such $u$ it adds $\binom{c_{s2u}}{2}$ to $c$.
	
	Upon an update, the algorithm keeps track of all updates of $c_{s2u}$ variables, and for each such update, is subtracts $\binom{c_{s2u}}{2}$ from $c$ before the update, and adds $\binom{c_{s2u}}{2}$ to it after.
	
	Since each $4$-cycles through~$s$ has exactly one mid-point $u$, and the number of $4$-cycles through~$s$ with a midpoint $u$ is exactly $\binom{c_{s2u}}{2}$, the algorithm is correct.
	The overhead over each operation of the  algorithm from \cref{lem: s2u alg}
	is constant.
\end{proof}

Finally, we observe that the above proof also implies a better worst-case algorithm than known,
as in the case of \paths st4.
We get a deterministic worst-case algorithm with $O(n)$ update time and constant query time,
improving upon the state of the art $O(m^{2/3})$ update time~\cite{HanauerHH22}
for any $m=\omega (n^{3/2})$.

\begin{corollary}
\label{cor:counting s 4cycles worst-case}
	There is a deterministic worst-case algorithm
	that counts the number of $4$-cycles through a given node $s$ with preprocessing time $O(m_0)$, 
	update time $O(n)$, 
	and query time $O(1)$.  
\end{corollary}

\section{Lower bounds via embedding}
\label{sec:lb with embedding}

In the case of an adaptive or oblivious add/remove adversary, 
we can achieve lower bounds for many problems by embedding a known worst-case lower bound graph as a subgraph of the entire smoothed graph.
To this end, an adaptive adversary can choose $\hat n=\Theta(pn)$ nodes, 
keep them disconnected from the rest of the graph, and fully control the structure of their induced subgraph. 
An oblivious add/remove adversary can do the same, but on $\hat n=n^{\delta/3}$ nodes for a small constant $\delta$.

\subsection{Polynomial lower bounds}
\label{sec:poly lbs}

We start by noting that an adaptive adversary can ``take control'' of parts of the graph,
i.e., choose a set $\hat V$ of nodes,
make sure it is disconnected from the rest of the graph,
and set the edges inside $\hat V$ as it desires.
We prove a more general lemma:
given a graph $G$, a set $R$ of potential edges and a set $R'\subseteq R$,
the adversary can reach a new graph $G'$ satisfying $(G\symdif G')\cap R=R'$, i.e. a graph 
where the edges of $R'$ are flipped, 
the edges of $R\setminus R'$ are as in the original graph,
and the rest of the edges may change arbitrarily.

\begin{lemma}[Embedding a graph]
	\label{lem:edge control adaptive}
	Consider an $n$-node graph $G$, a set $R\subseteq \binom{V}{2}$ of $r=|R|$ potential edges,
	and a subset $R'\subseteq R$ of size $|R'|= r'$.
	An adaptive adversary can reach a new graph $G'$ satisfying 
	$(G\symdif G')\cap R=R'$
	within at most 
    $\ell$ steps with probability at least
    $1 - 2\exp\left( -p \ell /40\right)$, provided that
    $\ell \geq \frac{10r'}{p}$  
    and
    $r\leq\frac{pn^2}{18}$.
\end{lemma}

We later use this lemma both to control a subset of nodes (their internal edges and the edges connecting them to the rest of the graph), and to simulate a sequence of worst-case changes on these nodes.

The adversary puts all its efforts on changing a subset of the edges of $R'$ and maintaining the edges of $R$ by reverting any random changes to them.
We analyze this as a 1-dimensional random walk process: 
the nature starts with a budget of at most $r'$ (edges that the adversary wants to change), and then a random change might add 1 to the sum by flipping an edge of $R\setminus R'$ or an edge of $R'$ that was already changed, 
while an adversarial change reduces it by 1 by reverting a change or fixing an edge of $R'$.
If the probability of the adversary to take a step is larger the probability of a random change to hit $R$, then by gambler’s ruin (see\cite[Section~7.2.1]{books:probability}) the process will eventually converge to $0$ (edges that the adversary wants to change), as desired.

\begin{proof}
	We say that a change is \emph{effective} if it changes an edge of $R$.
	The adversary will change an edge of $R$ as long as there is a change to make,
    which happens w.p.~$p$.
    The probability that a random change happens \emph{and} hits $R$ is $(1-p)\cdot r / \binom{n}{2}$. By our assumption, this is at most $p/9$.
    A sequence of $\ell$ steps is expected to have at least $p\ell$ effective steps (just by the adversarial turns), and by Chernoff, the probability of having less than $p\ell /2$ effective steps is at most $\exp\left(-p\ell/8\right)$.

    Consider the first $\ell' = p\ell/2$ effective steps, and let $f(i)$ be the number of edges that need to be flipped after $i$ effective steps.
    Conditioned on such a step, the probability of an adversarial change is at least $9$ times that of a random one, and therefore at least $9/10$.
    Thus, if $f(i-1) > 0$ then with probability\footnote{We need $f(i-1) > 0$ to allow the adversary to make a required change. For simplicity, the calculations are as if a random change is never in our favor.} at least $9/10$ we have $f(i) = f(i-1) - 1$, and otherwise $f(i) = f(i-1) +1$. In particular $f(i-1) - f(i) \in \set{-1,1}$.
    
    Denote by $E_t$ the event that $f(i) \neq 0$ for any $i\leq t$ (otherwise we are done within $t$ effective steps).
    Define for each step $y_i := (f(i-1) - f(i) + 1)/2$, then $y_1,\dots t_{\ell'}$ are Bernoulli (assuming $E_{\ell' - 1}$), each with expectation at least $9/10$ and overall $\Expc{}{\sum y_i} \geq 9\ell'/10$. By telescopic sum we can write:
    \[
    \begin{aligned}
        \Prob{}{f(\ell') > 0}
        &= \Prob{}{f(0) - f(\ell') < r'}\\
        &= \Prob{}{\sum_{i\in[\ell']} y_i
        < \frac{r' + \ell'}{2}}\\
        &\leq \Prob{}{\sum_{i\in[\ell']} y_i
        < (6/10)\ell'} \\
        &\leq \exp\left(-\ell'/20\right) ,
    \end{aligned}
    \]
    where the next-to-last inequality is due to the premise $r' \leq p\ell/10 = \ell'/5$, and the last by a Chernoff bound.
    Thus, if we weren't done within $\ell'-1$ effective steps, we must be done within $\ell'$ with only this error.
    
    Using union bound and plugging in $\ell' = p\ell/2$, the total error probability is at most $2\exp\left(-p\ell /40\right)$.
\end{proof}

We state known lower bounds for the worst-case (non-smoothed) setting,
with $\hat \tu$ and $\hat \tq$ update and query times.
\begin{lemma}
	\label{lem:lb for many problems-original}
    There is no dynamic algorithm for 
    \textbf{%
    	counting $s$ $k$-cycles for $k\geq 3$ odd} (and hence $s$-triangle detection),
    \textbf{%
    	counting $s$ $4$-cycles,
    	counting \paths{s}{t}{3},
    	counting \paths{s}{t}{4},
    	bipartite perfect matching,
    	bipartite maximum matching} 
    	and 
    	\textbf{bipartite minimum vertex cover} 
	with error probability at most $1/3$,
	polynomial preprocessing time,
	$\hat\tu(\hat n)=O(\hat{n}^{1-\epsilon})$
    and
    $\hat\tq(\hat n)
    =O(\hat{n}^{2-\epsilon})$
    (both amortized)
	for any $\epsilon>0$ 
	unless the \omv conjecture is false.	

    The proofs of all these claims use worst-case dynamic  graphs of a similar structure:
    either one or $\hat n$ phases, where each phase consists of $\hat n$ updates to different edges, followed by a single query.
\end{lemma}

These claims and their proofs appear in: 
\cite[Theorem 26]{HanauerHH22} for
$s$ $k$-cycles;
\cite[Lemma 7.6]{HLS22} for $s$ 4-cycles\footnote{The statement of \cite[Lemma 7.6]{HLS22} regards $\tq=O(1)$ only, but the proof clearly extends to a general $\tq$.};
\cite[Lemma 7.5]{HLS22} for \paths{s}{t}{4};
and
\cite[Corollary~3.4 and footnote~9]{HenzingerKNS15} for
\paths{s}{t}{3},
$s$-triangle detection (which is also implied by $s$ $k$-cycles counting),
bipartite matchings
and
bipartite vertex cover.

We next augment this lemma to get lower bounds for the $p$-smoothed setting with an adaptive adversary.

\begin{theorem}
	\label{thm:lb for many problems-smoothed advers}
	Fix $0 \leq p \leq 1$. 
	There is no dynamic algorithm for
	 $p$-smoothed dynamic graphs with an adaptive adversary
	 for 
	\textbf{%
		counting $s$ $k$-cycles for $k\geq 3$ odd} (and hence $s$-triangle detection),
	\textbf{
		counting \paths{s}{t}{3},
		counting \paths{s}{t}{4},
		bipartite perfect matching,
		bipartite maximum matching} 
	and 
	\textbf{bipartite minimum vertex cover} 
	with error probability at most $1/3$,
	polynomial preprocessing time,
	$\tu(n)=O(p^{2-\epsilon}n^{1-\epsilon})$
	and
	$\tq(n)
	=O((pn)^{2-\epsilon})$
	(both amortized)
	for any $\epsilon>0$,
	unless the \omv conjecture is false.	
\end{theorem}

In the next section, we show that 
counting $s$-triangles, $s$ $4$-cycles, \st 3-paths and \st 4-paths need even higher running times (in terms of $p$), i.e.~one cannot have 
$\tu(n)=O(pn^{1-\epsilon})$
and 
$\tq(n)=O(pn^{2-\epsilon})$.
Moreover, these lower bounds hold for the weaker, oblivious adversaries.

\begin{proof}
	When $p = o(1/n)$, the assertion is trivial. 
    Assume for contradiction that in the $p$-smoothed setting with adaptive adversary, one of the above problems has an algorithm with better running times than claimed on an $n$-node graph.
	For the bipartite problems, we will show the lower bound on a bipartite balanced graph.
	We devise a fast dynamic algorithm for the same problem in a worst-case dynamic graph $\cal H$ on $\hat n=pn/20$ nodes 
	by creating a $p$-smoothed $n$-node dynamic graph $\cal G$ and running the claimed fast algorithm on it.
	The new graph $\cal G$ is created by simulating a carefully chosen adversary that makes sure $\cal G$ has $\cal H$ as a disconnected subgraph;
	hence, the answers to the queries on $\cal G$ and $\cal H$ are correlated.
	This allows us to solve the problem on $\cal H$ with running times better than stated in \cref{lem:lb for many problems-original},
	a contradiction.
	
	Consider an instance of the problem on a fully-dynamic non-smoothed graph $\mathcal{H}$ on $\hat n =pn/20$ nodes, with an initial graph $H_0$.
    Recall that the worst-case dynamic graph is composed of phases, 
    where each phase consists of $\hat n$ updates followed by a query.
    We simulate the first phase, where $H_1$ is the final graph on which 
    the query is made;
    the rest pf the phases, if exist, are similar.
	
    Choose a random $n$-node graph $G_0$ where each edge exists w.p.~$1/2$, 
    which is a $p$-smoothed initial graph by \cref{observation:initial graph}. 
   	Note that there is no need to initialize the graph, and instead, we decide on the existence of each edge at the first time it is read or changed.
	Chooses an arbitrary set $\hat V$ of $pn/20$ nodes, 
	which includes the designated nodes ($s$ and sometimes $t$) if those exist in the problem,
	and is a balanced bipartite graph if the problem is on bipartite graphs.
    Let $R=\hat V\times V$ be the set of all potential edges inside $\hat V$ and connecting $\hat V$ to the rest of the graph. 
    Set $r=|R|$, so $\frac{p}{21}n^2\leq r\leq\frac{p}{20}n^2$ for large enough $n$.
    The set $R$ is fixed for the rest of this proof.
    
    By \cref{lem:edge control adaptive} with $R'\subseteq R$ and $r'\leq r$,
    the adversary can make sure the subgraph on $\hat V$ is isomorphic to the initial lower bound graph $H_0$,
    and is disconnected from the rest of the graph,
    in 
     $\ell=40 \frac{r}{p}\log n$
     steps w.h.p. 
	Start the simulation by these $\ell$ steps:
	chooses each step to be adversarial w.p.~$p$ and otherwise random,
	choose random changes uniformly at random, and the adversarial changes as in the lemma.
	
	The $n$-node graph (of $\cal G$) now has a subgraph on $\hat V$  isomorphic to $H_0$ and disconnected from the rest of the graph.
	We now simulate the first phase.
	Let $R'$ be the set of edges in $\hat V\times \hat V$ that are isomorphic to $H_0\symdif H_1$, and set $r'=|R'|$;
	the structure of the lower bound graph $\cal H$ guarantees 
	$r'=\hat n=pn/20$.
	The adversary uses \cref{lem:edge control adaptive} 
    to perform 
    $\ell=40 \frac{r'}{p}\log n$
    updates
    such that after them the subgraph on $\hat V$ is isomorphic to $H_1$ and disconnected from the rest of the graph w.h.p. 
    It then perform the first query, and returns the same value or a function of it, as follows.
	
	For all the subgraph counting problems, there is the same number of subgraphs through $s$ (or $s$ and $t$) in $H_1$ and in the disconnected subgraph on $\hat V$ isomorphic to it,
	so the algorithm returns the same.
	For maximum bipartite matching,
	the graph on $V\setminus\hat V$ is a random balanced bipartite graph on $n-\hat n$ nodes with each edge w.p.~$1/2$.
	Hence, it has a matching of $\frac{n-\hat n}{2}$ edges w.h.p.~\cite{ER64},
	and this quantity should be added to the output on $H_1$.
	The same output describes the minimum vertex cover size, by K\H{o}nig's Theorem~\cite{Konig31}.
	Finally, the graph has a perfect matching iff $h_1$ has a perfect matching, so the outputs are the same.

    For complexity, 
    note that the preprocessing takes $\tilde O(n^2)$ edge updates (in the graph $\cal G$),
    i.e. $\tilde O(n^2\tu(n))$ operations
    (and the other operations are negligible).
    By assumption, this is 
    $O(p^{2-\epsilon}n^{3-\epsilon})$
    which is polynomial in $\hat n$.
    Each phase of $\hat n$ updates of the original graph involves $O(n\log n)$ updates and a single query.
    Hence, each update of the original graph takes amortized 
    \begin{equation*}
    \begin{split}
    \hat\tu(\hat n)
    	&=\frac{1}{\hat n}n\log n\cdot \tu(n)\\
    	&=\frac{1}{\hat n}	O((pn)^{2-\epsilon})\\
    	&=O(\hat n^{1-\epsilon})
    \end{split}
    \end{equation*}
    and each query takes 
	\begin{equation*}
	\begin{split}
			\hat\tq(\hat n)
			&=\tq(n)
			=O((pn)^{2-\epsilon})
			=O(\hat n^{2-\epsilon})
		\end{split}
	\end{equation*}    
    contradicting \cref{lem:lb for many problems-original}.
\end{proof}

\subsection{Connectivity lower bound}

In \cref{lem:edge control adaptive}, if $r'$ is small, and we take the minimal number of steps $\ell = \Theta(r'/p)$, the error term might grow too large to apply this simulation for, say, $\poly(n)$ times.
One way to circumvent that is by making each simulation longer by a factor of $\log n$, reducing its error to $1/\poly(n)$, which allows us to union bound over all simulations.
However, in some cases (specifically, connectivity) this is too costly. 
We show a better result utilizing the amortization over all simulations.

\begin{lemma}
	\label{lem:edge control adaptive with few changes}
    Consider an $n$-node graph $G_0$, a set $R\subseteq \binom{V}{2}$ of $r=|R|$ potential edges,
    and a sequence $R'_1,\ldots R'_k$ of $k$ subsets $R'_i\subseteq R$ of size $r'_i=|R'_i|$ each, and $1\leq r_i\leq \hat r$.
    An adaptive adversary can guarantee $k$ phases, each starting with $G_{i-1}$ and ends with $G_i$ such that $\left(G_i \symdif G_{i-1}\right) \cap R = R_i$, using at most $\ell=\frac{12k \hat{r}}{p}$ steps,
    with error probability $1/k^c$ for any constant $c$ and sufficiently large $k$, provided $r\leq\frac{pn^2}{18}$.
\end{lemma}

\begin{proof}
    Using the guarantee on $r$, we can apply~\cref{lem:edge control adaptive}, using $\ell_{\max} = \frac{40(c+2) \hat{r} \log k}{p}$, with some constant $c > 0$. The error is upper bounded by $k^{-c-2}$, and $k^{-c-1}$ for all iterations using union bound.
    We therefore assume the adversary ``gives up'' on a simulation after $\ell_{\max}$ steps, which happens with a small probability. This allows us to show concentration over enough phases. 
	
    Consider a single phase $i$. As per the analysis from the proof of~\cref{lem:edge control adaptive}, denote $\ell_i := \argmin{j} \left(f(j) = G_i\right)$, the first time that $G_i$ is reached.

    Applying~\cref{lem:edge control adaptive} with different values gives $\Prob{}{\ell_i > t} \leq \exp\left(-pt/40\right)$. By Fubini's theorem,
    \[
        \Expc{}{\ell_i} 
        = \sum_{t=0}^{\infty} \Prob{}{\ell_i \geq t} 
        \leq \sum_{t=0}^{\infty} \exp\left(-pt/40\right) .
    \]
    But this last sequences converges (having a fixed fractional quotient, $e^{-p/40} < 1$), proving the expectation is finite $\Expc{}{\ell_i} < \infty$, which is crucial to invoke Wald's equation for stopping times later on. 
    
    Similarly to the proof of~\cref{lem:edge control adaptive}, we start at $f(0) = r_i$ but this time we count all steps, not only effective ones.
    We continue until $f(\ell_i) = 0$, the first time we reach $0$. In between $f(i) > 0$, and so each difference $X_j = f(j)-f(j-1)$ has the same distribution, $X$. With probability at least $p$, this step hit an edge of $R$, and conditioned on that, we have expected decrease of at least $4/5$.
    Overall, we write $\Expc{}{X} \leq -4p/5$. The sum of all differences is exactly $r_i$, and so by Wald's equation:
    \[
        r_i = \Expc{}{\sum_{i\in[\ell']} X_i} = \Expc{}{\ell_i} \cdot \Expc{}{X} .
    \]
    Plugging in the values and bounds, we finally get $\Expc{}{\ell_i} \leq 5r_i/(4p)$. Recall the fact that $\ell_i \in[0,\ell_{\max}]$, and note that each $\ell_i$ is independent of the others.
    
    Use $\ell := \sum_{i\in[k]} \ell_i$ with $\Expc{}{\ell} \leq 5\hat{r}k/(4p)$, and apply Hoeffding's inequality:
    \[
        \Prob{}{\ell \geq \frac{12\hat{r}k}{p}}
        \leq \Prob{}{\ell \geq \Expc{}{\ell} + \frac{10\hat{r}k}{p}}
        \leq \exp\left(- \frac{200\hat{r}^2k^2/p^2}{k \ell_{\max}^2}\right) 
        = \exp\left(-\frac{k}{1600c^2 \log ^2 (k)}\right)
        \leq \frac{1}{k^{c+1}} .
    \]
    Where the last inequality holds for large enough $k$.
    Applying union bound over the two errors, the adversary succeeds within $12k\hat{r}/p$ steps with error probability at most $1/k^c$. 
\end{proof}

We apply~\cref{lem:edge control adaptive with few changes} with $\hat{r} = 4$ and $k = \sqrt{n}$.
This gives error that is $o_n(1)$, and overall $\Theta\left(\sqrt{n}/p\right)$ steps (with no extra logarithmic factors, as desired).

We state a known lower bound for connectivity in the worst-case (non-smoothed) setting. 
\begin{lemma}[\hspace{1sp}{\cite[Section~9.1]{PatrascuD06}}]
	\label{lem:lb for connectivity-original}
	There is no dynamic algorithm for 
	\textbf{connectivity}
	with error probability $\hat n^{-\Omega(1)}$,
	constant query time,
	and 
	$\hat\tu(\hat n)=o(\log \hat n)$
	amortized time.
	
	The proof uses a sequence of phases,
	which must take $\Omega(\sqrt{\hat n}\log\hat n)$ steps each on average.
	Each phase is composed of 
	$2\sqrt{\hat n}$ edge flips
	and then 
	$\sqrt{\hat n}$
	subsequences, each composed of $O(1)$ edge flips and a query.
\end{lemma}

The connectivity lower bound is from \cite{PatrascuD06}, 
where the result appears in Section~9.1 and the details of the proof are in Theorem~2.4 and Section~6.1.

The proof uses a sequence of $\sqrt[3]{\hat n}$ phases,
each composed of two macro-operations:
a macro-updates, followed by verify-sum.
A macro-update can be implemented by $2\sqrt{\hat n}$ edge flips, 
and a verify-sum can be implemented by 
$\sqrt{\hat n}$ phases, each consists of $O(1)$ edge flips and a single query.
The paper shows that each phase must take $\Omega(\sqrt{\hat n}\log {\hat n})$ cell-prob steps in average, 
and hence the same number of operations,
which yields the lower bound.

Combining this result with \cref{lem:edge control adaptive} and \cref{lem:edge control adaptive with few changes}, we can give a lower bound for connectivity in the $p$-smoothed setting with an adaptive adversary.

\begin{theorem}
	\label{thm:lb for connectivity-smoothed advers}
	Fix $\log^2n/n < p \leq 1$. 
	There is no dynamic algorithm for
	$p$-smoothed dynamic graphs with an adaptive adversary
	for 
	\textbf{connectivity}
	with error probability $(pn)^{-\Omega(1)}$,
	constant query time,
	and 
	$\tu(n)=o(\log(p n))$
	amortized time.
\end{theorem}

\begin{proof}
	Assume for contradiction that in the $p$-smoothed setting with adaptive adversary, there is a connectivity algorithm as above with $\tu=o(\log(pn))$ on an $n$-node graph.
	We follow a simulation argument as in the proof of \cref{thm:lb for many problems-smoothed advers}.
	
	Consider an instance of the problem on a fully-dynamic non-smoothed graph $\mathcal{H}$ on $\hat n =pn/20$ nodes, with an initial graph $H_0$.
	Recall that the worst-case dynamic graph is composed of phases, 
	where each phase consists of $2\hat n$ updates,
	followed by $\hat n$
	sub-phases of $O(1)$ updates and a query.
	We simulate the first phase,
	where the first part of $2\sqrt{\hat n}$ edge flips ends with a graph $H'$, and 
	the queries in the next $k=\sqrt{\hat n}$ sub-phases are on the graphs $H_1,\ldots,H_k$.
	The rest of the phases are similarly simulated.
	
	Choose a random $n$-node graph $G_0$ where each edge exists w.p.~$1/2$ (a $p$-smoothed input by \cref{observation:initial graph}),
	and an arbitrary set $\hat V$ of $pn/20$ nodes.
	Let $R=\hat V\times V$ be the set of all potential edges inside $\hat V$ and connecting $\hat V$ to the rest of the graph.
	Set $r=|R|$, so $\frac{p}{21}n^2\leq r\leq\frac{p}{20}n^2$ for large enough $n$.
	Fix an edge $\hat e$ connecting $\hat V$ and the rest of the graph.
	
	By \cref{lem:edge control adaptive} with $R'\subseteq R$ and $r'\leq r$,
	the adversary can make sure the subgraph on $\hat V$ is isomorphic to the initial lower bound graph $H_0$,
	and the only edge connecting $\hat V$
	and the rest of the graph is $\hat e$,
	in
	$\ell=40 \frac{r}{p}\log n$
	steps w.h.p. 
	Start the simulation with these $\ell$ steps.
	
	The $n$-node graph (of $\cal G$) now has a subgraph on $\hat V$  isomorphic to $H_0$ and connected to the rest of the graph by a single edge $\hat e$.
	We now simulate the first part of the phase, which ends with $H'$.
	Let $R'$ be the set of edges in $\hat V\times \hat V$ that are isomorphic to $H_0\symdif H'$, and set $r'=|R'|$;
	the structure of the lower bound graph $\cal H$ guarantees 
	$r'=2\sqrt{\hat n}=\sqrt{pn/5}$.
	The adversary uses \cref{lem:edge control adaptive} 
	to perform
	$\ell=10 \frac{r'}{p}=O(\sqrt{n/p})$ steps (note that $\ell$ is large enough due to the range of $p$)
	such that after them the subgraph on $\hat V$ is isomorphic to $H'$ and is connected to the rest of the graph only by $\hat e$ w.h.p.
	
	We now turn to simulate the sub-phases.
	Recall that there are $k=\sqrt{\hat n}$ sub-phases,
	each with $O(1)$ edge changes an a single query.
	By \cref{lem:edge control adaptive with few changes},
	this can by be simulated in $\ell=O(\sqrt{\hat n}/p)=O(\sqrt{n/p})$ steps
	with failure probability $\hat n^{-c/2}=O((pn)^{-c/2})$ for any constant $c$.
	
	In each query, we return the same output on the simulated $H_i$ as on the $n$-node graph.
	This is because 
	the subgraph on $\hat V$ is connected iff $H_i$ is connected, by the isomorphism:
	the rest of the graph is a random graph with each edge existing w.p.~$1/2$, so it is connected w.h.p.\ if $n$ is large enough (see \cref{claim:negligible_probabilities});
	and the edge $\hat e$ connects 
	$\hat V$ to the rest of the graph.
	Hence, the whole graph is connected iff $H_i$ is connected.

	For complexity, 
	note that 
	each query is answered in a constant time, 
	and each phase is simulated using 
	$O(\sqrt{\hat n})$ updates, which are
	$O(\sqrt{\hat n})\tu(n)$ steps.
	As a step must take $\Omega(\sqrt{\hat n}\log(\hat n))$
	steps on average,
	we cannot have $\tu(n)=o(\log(\hat n))=o(\log(pn))$ amortized time.
\end{proof}

\subsection{Lower bounds with oblivious add/remove adversary}
In this section we focus on oblivious adversaries, but with the power to add/remove edges, which allows it to ``insist'' on the embedding. 
A weaker version of~\cref{lem:edge control adaptive} can be shown for this adversary, one that only holds with $p = 1 - o(1)$, but it is enough to leverage to a non-trivial lower bound, showing separation between the oblivious add/remove and the oblivious flip adversaries (along with the upper bounds of \cref{lem:ub_oblivious_flip}).

For convenience, write $q = 1 - p$, the probability that a random edge is flipped at a given step. We show:  

\begin{lemma}[Embedding a graph, oblivious add/remove adversary]
	\label{lem:edge control oblivious}
	Consider an $n$-node graph $G$, a set $R\subseteq \binom{V}{2}$ of $r=|R|$ potential edges,
	and a subset $R'\subseteq R$ of size $|R'|= r'$.
	An adaptive adversary can reach a new graph $G'$ satisfying 
	$(G\symdif G')\cap R = R'$
	within at most 
    $\ell$ steps with failure probability at most
    \[
        r' \cdot q^{\frac{\ell}{r'}} + \frac{q\ell r}{n^2}.
    \]
\end{lemma}

The oblivious adversary uses all its changes for $R'$, trying to enforce the desired arrangement (using add and remove operations). This can be done even with no knowledge about the current graph, assuming the previous embedding worked well.

Interestingly, even without knowledge of $G_0$, the oblivious adversary that can insist on certain changes, which allows it to control its own piece of the graph, albeit quite small. 
The flip adversary cannot do even this.

\begin{proof}
    Consider a sequence of $\ell $ changes.
    The adversary equally divides its attention among the $r'$ edges, repeatedly adding or removing them accordingly.
    Each adversarial change is swapped by a random one with probability $q$, and so for each edge, with probability $q^{\ell/r'}$ none of the $q$ add/remove operations on it follow through and the change does not happen.
    Using a union bound over all edges of $R'$, we have probability of $r'\cdot q^{\ell/r'}$ that some edge was not changed by the adversary.

    On the other hand, each step involves a probability of $q \cdot (r/n^2)$ for a random edge to be chosen, and then hit the set $R$. 
    Using a union bound over $\ell$ changes, we have an upper bound of $q\ell r/n^2$ that any random change in the entire sequence touched $R$.
\end{proof}

We use this lemma to prove lower bounds for the $p$-smoothed setting with an oblivious add/remove adversary, in a way similar to the proof of \cref{thm:lb for many problems-smoothed advers}.

\begin{theorem}
	\label{thm:lb for many problems-smoothed oblivious ar}
	Fix a constant $0<\delta<1$ and $1-n^{-\delta} \leq p \leq 1$. 
	There is no dynamic algorithm for
	$p$-smoothed dynamic graphs with an              oblivious add/remove adversary
	for 
	\textbf{%
		counting $s$ $k$-cycles for $k\geq 3$ odd} (and hence $s$-triangle detection),
	\textbf{%
		counting \paths{s}{t}{3},
		counting \paths{s}{t}{4},
		bipartite perfect matching,
		bipartite maximum matching} 
	and 
	\textbf{bipartite minimum vertex cover} 
	with error probability at most $1/3$,
	polynomial preprocessing time,
	$\tu(n)=O(n^{\frac{\delta}{3}-\epsilon})$
	and
	$\tq(n)
	=O(n^{\frac{2\delta}{3}-\epsilon})$
	(both amortized)
	for any $\epsilon>0$ 
	unless the \omv conjecture is false.	
\end{theorem}

This result is clearly much weaker than 
\cref{thm:lb for many problems-smoothed advers},
which is not surprising as the oblivious add/remove adversary is weaker than the oblivious one.
Yet, it is interesting as it shows a separation between the oblivious add/remove adversary and the oblivious flip adversary for the problem of bipartite perfect matching --- compare to \cref{lem:ub_oblivious_flip}.

\begin{proof}
	Assume for contradiction that in the $p$-smoothed setting with oblivious add/remove adversary, one of the above problems has an algorithm with better running times than claimed on an $n$-node graph.
	We devise a fast dynamic algorithm for the same problem in a worst-case dynamic graph $\cal H$ on $\hat n=n^{\delta/3}$ nodes
	by creating a $p$-smoothed $n$-node dynamic graph $\cal G$ and running the claimed fast algorithm on it, as in the proof of \cref{thm:lb for many problems-smoothed advers}.
	
	Consider an instance of the problem on a fully-dynamic non-smoothed graph $\mathcal{H}$ on $\hat n $ nodes, with an initial graph $H_0$.
	Choose a random $n$-node graph $G_0$ where each edge exists w.p.~$1/2$ as the initial graph (\cref{observation:initial graph}). 
	Chooses an arbitrary set $\hat V$ of $\hat n$ nodes, 
	which includes the designated nodes ($s$ and sometimes $t$) if those exist in the problem,
	and is a balanced bipartite graph if the problem is on bipartite graphs.
	Let $R=\hat V\times V$ be the set of all potential edges inside $\hat V$ and connecting $\hat V$ to the rest of the graph. 
	Set $r=|R|= n^{1+\delta/3}$.
	
	By \cref{lem:edge control oblivious} with $R'\subseteq R$ (and $r'\leq r)$,
	the adversary can make sure the subgraph on $\hat V$ is isomorphic to the initial lower bound graph $H_0$,
	and is disconnected from the rest of the graph,
	in
	$\ell=3r\log n$
	steps with error probability 
	\[
	r'\cdot q^{\ell/r'} + \frac{q\ell r}{n^2}
	\leq r (1/2)^{3\log n} + 3\log n \cdot n^{-\delta/3}
	= o(1) .
	\]
	
	The $n$-node graph (of $\cal G$) now has a subgraph on $\hat V$  isomorphic to $H_0$ and disconnected from the rest of the graph.
	We now simulate the first phase.
	Let $R'$ be the set of edges in $\hat V\times \hat V$ that are isomorphic to $H_0\symdif H_1$, and set $r'=|R'|$;
	the structure of the lower bound graph $\cal H$ guarantees 
	$r'=\hat n$.
	The adversary uses \cref{lem:edge control oblivious}  
	to perform 
	$\ell=3r'\log n$
	updates
	such that after them the subgraph on $\hat V$ is isomorphic to $H_1$ and disconnected from the rest of the graph with error probability at most
	\[
	r'\cdot q^{\ell/r'} + \frac{q\ell r}{n^2}
	\leq n^{-2\delta\log n} + 3n^{-1-\delta/3}\log n
	= o(1) .
	\]
	Note that even repeating this for $\hat n$ phases leaves the error probability in $o(1)$.
	It then performs the first query, and returns the same value or a function of it, as in the proof of
	\cref{thm:lb for many problems-smoothed advers}.
	
	For complexity, 
	note that the preprocessing takes $O(n^2)$ edge updates (in the graph $\cal G$),
	i.e. $O(n^2\tu(n))$ operations.
	By assumption, this is
	$O(n^{2+\delta/3-\epsilon})$
	which is polynomial in $\hat n$.
	Each phase of $\hat n$ updates of the original graph involves $3r'\log n=O(\hat n \log n)$ updates and a single query.
	Hence, each update of the original graph takes amortized 
	\[
		\begin{split}
			\hat\tu(\hat n)
			&=\frac{1}{\hat n}O(\hat n \log n)\cdot \tu(n)\\
			&=O(\tu(n)\log n)\\
			&=O(n^{\frac{\delta}{3}-\epsilon}\log n)\\
			&=O(\hat n^{1-\epsilon})
		\end{split}
	\]
	and each query takes 
	\[
		\begin{split}
			\hat\tq(\hat n)
			&=\tq(n)\\
			&=O(n^{\frac{2\delta}{3}-\epsilon})\\
			&=O(\hat n^{2-\epsilon})
		\end{split}
	\]   
	contradicting \cref{lem:lb for many problems-original}.
\end{proof}

\section{Lower bounds for counting small subgraphs}
\label{sec:lb small subgraphs}
The main focus of this section is proving a conditional lower bound for counting \paths{s}{t}{3} in a $p$-smoothed model, even with the weakest adversary---an oblivious flip adversary. 
We then extend this lower bound to other subgraph counting problems.

Given the \omv conjecture, the average-case parity \oumv problem is also hard (\cref{lem:conjecture_to_ac_oumv}).
The proof then goes through the problem of counting \paths{s}{t}{3} in a specific family of well-structured graphs, called $P_3$-partite graph, similarly to the outline of~\cite{HLS22}.
We show:
\begin{itemize}
    \item Average-case parity \oumv can be solved using an algorithm that counts \paths{s}{t}{3} in $P_3$-partite $p$-smoothed dynamic graphs (\cref{sec:lb:ac_oumv_to_restricted_graphs}).

    \item An algorithm that counts \paths{s}{t}{3} in a (general) $p$-smoothed dynamic graph, can be used to count the specific type of such paths in a $P_3$-partite $p$-smoothed dynamic graph (\cref{sec:lb:restricted_graphs_to_general_graphs}).
    
    \item From these reductions, we conclude the conditional hardness of counting 
    \paths{s}{t}{3} in a general $p$-smoothed dynamic graph (\cref{sec:lb:wrapping_things_up}).    
\end{itemize}

Overall, we get an (almost) tight lower bound for counting \paths{s}{t}{3}, for (almost) any value of $p$.
This is extended to other subgraph counting problems in \cref{sec:lb_other_small_graphs}. 
Throughout this section, we use adversaries that pick the initial random graphs at random (with each edge existing with an independent probability of~$1/2$), which assures that the smoothed initial graphs are also random, by \cref{observation:initial graph}.

\subsection{Solving OuMv by counting \paths{s}{t}{3} on $P_3$-partite graphs}
\label{sec:lb:ac_oumv_to_restricted_graphs}

A graph is called $P_3$-partite if it is composed of 4 node sets, $V = \set{s} \sqcup A \sqcup B \sqcup \set{t}$,
and all its edges are of type $sA, AB$, or $Bt$.
Formally, it is a $(H,\Pi)$-partite graph with $H$ being a path $(v_1,v_2,v_3,v_4)$ and $\Pi$ mapping $s$ to $v_1$, all the nodes of $A$ to $v_2$, all the nodes of $B$ to $v_3$, and $t$ to~$v_4$.
In this section we prove the following lemma.

\begin{lemma}
    \label{lem:ac_oumv_to_restricted_graphs}
    If there exists a data structure for counting \paths{s}{t}{3} on a $p$-smoothed $P_3$-partite dynamic $n'$-node graph initialized at random, with an oblivious flip adversary and  preprocessing, update and query times $\tp(n'), \tu(n'), \tq(n')$ and error probability at most $1/100$, then there exists an algorithm solving the average-case parity \oumv problem of dimension $n$ in
    \[
        3\tp(n') + O\left(\frac{n \log n}{p}\right) \cdot \tu(n') + O(n) \cdot \tq(n') + O\left(\frac{n^2 \log n \log(n/p)}{p}\right) 
    \]
    steps of computation, where $n' = 2n+2$.
\end{lemma}

\subsubsection{The setting}
 
\paragraph{\oumv and \paths{s}{t}{3}.} 
Following the connection established in~\cite{HenzingerKNS15}, 
we associate the \oumv problem of dimension $n$ with $P_3$-partite graphs on $n' = 2n+2$ nodes $V = \set{s} \sqcup A \sqcup B \sqcup \set{t}$, where $\card{A} = \card{B} = n$.
A $P_3$-partite graph has edges of types $sA, AB, Bt$, and they can be naturally represented algebraically: 
$u$ (resp., $v$) represent all neighbors of $s$ in $A$ (resp., of $t$ in $B$), and $M$ is the adjacency matrix between $A$ and $B$.
It is easy to verify that the multiplication $u^T M v$ (over $\R$) is the number of $3$-paths of type $sABt$.

\paragraph{Indexing.} With the aforementioned connection in mind, it is more convenient to index the vectors $u, v$ and the matrix $M$ with graph edges. 
That is, we use $u(s,a_i)$ for $u(i)$, use $v(b_j,t)$ for $v(j)$ and use $M(a_i, b_j)$ for $M(i,j)$. 
The indices of the three objects are disjoint and together they correspond to all $n(n+2)$ allowed edges in the $P_3$-partite graph above. 
We often use $u(e), v(e), M(e)$, with an appropriate edge $e$.

\paragraph{The adversarial strategy.}
We fix a (randomized) oblivious flip  adversary throughout the section, an adversary that flips an edge of $sA$ or $Bt$ chosen u.a.r, and never flips an edge of $AB$:

\[
\begin{aligned}
        \advdist (e):=
            \begin{cases}
        		\frac{1}{2n} , & \text{if $e$ is of type $sA$ or $Bt$}\\
                0, & \text{if $e$ is of type $AB$}
		  \end{cases}
\end{aligned}
\]
In the $p$-smoothed model, each change samples an edge independently and according to the distribution $\advdist^p = p\cdot \advdist + (1-p)\cdot U_{P_3}$, where $U_{P_3}$ is the uniform distribution over all $n(n+2)$ allowed edges in the $P_3$-partite graph. 
More specifically:
\[
\begin{aligned}
        \advdist^p (e):=
            \begin{cases}
        		\frac{p}{2n} + \frac{1-p}{n(n+2)}, & \text{if $e$ is of type $sA$ or $Bt$}\\
                \frac{1-p}{n(n+2)}, & \text{if $e$ is of type $AB$}
		  \end{cases}
\end{aligned}
\]

\paragraph{Random sequence of updates.}
Our reduction needs to create an authentic sequence of changes, as if it came from the $p$-smoothed dynamic graph with the adversarial strategy above. 
Such a sequence is simply a (randomly ordered) sample set from $\advdist^p$, which can be represented by its histogram $H$, a vector of occurrences where $H(e)$ is the number of times an edge $e$ appears.

In actuality, we require \emph{conditional} sampling: for new \oumv inputs $u_i, v_i$, we wish to create a sequence of changes that ends with certain edge sets for types $sA, Bt$ (dictated by $u_i, v_i$). This means that (the parity of) $H(e)$ for all these edges is fixed ahead, which can potentially skew other things in our sample set (e.g., since the sum of all $H(e)$ is the total number of changes performed).

\paragraph{Poissonization.}
In order to overcome such complications, we use Poissonization~(\cref{fact:poisson_properties}): we sample $\poisson{t}$ changes from $\advdist^p$ by independently sampling a histogram: the number of occurrences of edge $e$, denoted $H(e)$, is sampled from $\poisson{\advdist^p (e)\cdot t}$. 
This way, conditioning on the parity of $H(e)$ does not affect the distribution of $H(e')$ for any $e' \neq e$.

\paragraph{Changing edges of type $AB$.}
Generating a sequence of changes that complies with the target vectors~$(u_i, v_i)$ might change edges that correspond to $M$ along the way (unless $p=1$).
This means a query for the number of \paths{s}{t}{3} is answered with $u_i^T M' v_i$ for some $M'$ instead of the original $M$. 
Similarly to~\cite{HLS22}, we circumvent this problem by creating $3$ correlated sequences of changes that contain the same changes on edges of $u_i, v_i$ and otherwise their changes correspond to matrices $M_1, M_2, M_3$ that sum (modulo 2) exactly to $M$. 
Thus, one can retrieve the desired (parity) value by distributivity over $\F_2$:
\[
    u_i^T M v_i = u_i^T M_1 v_i + u_i^T M_2 v_i + u_i^T M_3 v_i .
\]

\subsubsection{The reduction}
Assume a data structure that runs on graphs with $n' = 2n+2$ nodes with preprocessing time $\tp(n')$, update time $\tu(n')$ and query time $\tq(n')$. We show how to use these operations in order to solve the parity \oumv problem (of dimension $n$).  

The initial inputs for \oumv are $M$ and $u_0, v_0$, where all entries have random Boolean value. These correspond perfectly to an initial random $P_3$-partite graph on $n'$ nodes and random initialization. 
One can construct the graph, apply the preprocessing and use a single query to retrieve the parity of $u_0 ^T M v_0$.

Next we describe how, given new inputs $u_i, v_i$, one can utilize the data structure to efficiently compute the parity of $u_i^T M v_i$.
This is done using three data structures, for each we generate a sequence of changes.
We can then query once each data structure, and combine the three answers to finish up.

\noindent
\begin{algorithm}[H]
        \label{alg:sol}
	\caption{Algorithm \sol for solving \oumv by counting paths in $P_3$-partite graphs}
		 \nl $\roundnumber \gets 5 n \log(n) / p$
   
		\nl Initialize $\mathcal{S}^1,\mathcal{S}^2,\mathcal{S}^3$ to be empty sequences of changes
		 
		\nl Compute the differences $\udif \gets u_i - u_{i-1}$ and $\vdif \gets v_i - v_{i-1}$ (modulo $2$)
		
		\nl For each $sA$ edge $e$, draw $z_e \sim \poisson{\advdist^p (e)\cdot \roundnumber} \vert_{z_e \sameparity \udif(e)}$. 		 
		Add $z_e$ copies of $e$ to $\mathcal{S}^j$ for $j\in[3]$
		
		\nl For each $Bt$ edge $e$, draw $z_e \sim \poisson{\advdist^p (e)\cdot \roundnumber} \vert_{z_e \sameparity \vdif(e)}$. 
		Add $z_e$ copies of $e$ to $\mathcal{S}^j$ for $j\in[3]$
		
		\nl \For{For edges of type $AB$, for each $j\in[3]$}{
			{Draw $z^j \sim \poisson{\frac{(1-p)n}{n+2} \cdot (\roundnumber/2)}$}
			
			 Draw $z^j$ uniformly random matrix entries $e$ (with repetitions), and add each such edge to the two sequences $\mathcal{S}^{j'}$ for which $j' \neq j$.
		}
    
    	 \nl Randomize the order of each of the three sequences $\mathcal{S}^1,\mathcal{S}^2,\mathcal{S}^3$. Feed each sequence to a different data structure instance, applying the update procedure with each change.
    	
    	 \nl Query all $3$ instances for the number of \paths{s}{t}{3}, sum the answers and output the parity of the sum. 
\end{algorithm}

Call the procedure above \sol
and run it for each new pair $u_i, v_i$ to produce the one-bit answer.
In particular, we refer to steps $(4)$--$(7)$ as the reduction steps, taking two difference vectors $\udif, \vdif$ and returning three sequences of changes $\mathcal{S}^1, \mathcal{S}^2, \mathcal{S}^3$.

\subsubsection{Analysis}

For the sake of clarity, in this section we use $\redseq^j$ to refer to the sequences $\mathcal{S}^j$  output by the reduction (for any $j\in[3]$), and denote the histogram of such a sequence by $\redhist^j$.
For $AB$-type edges, drawn in step (6), denote the amount of times the edge $e$ was chosen for some $j\in[3]$ by $z_e^j$. 
Thus, for every $AB$-type edge $e$ we have 
\[
    \redhist^j(e) = \sum_{j' \in [3]\setminus \set{j}} z_e^{j'} ,
\]
since edges were added to each sequence from the two iteration with different index.

Note that the first inputs $M, (u_0, v_0)$ are simply mirrored perfectly in the initial graphs of all three copies.
We next prove the correctness of \sol for any fixed sequence of vector pairs. That is, we show that the reduction maintains the connection between \oumv and \paths st3.

\begin{claim}[Correctness for a fixed input]
    \label{claim:correctness_fixed}
    Fix inputs $M$ and $(u_0, v_0), \dots, (u_n, v_n)$.
    For any $i = 0, \dots, n$, 
    if all the queries to the $3$ data structures are answered correctly,
    then the online algorithm correctly outputs~$u_i^T M v_i$.
\end{claim}
 
\begin{proof}
    We prove that the following invariant holds for $i=0, \dots, n$: at the end of iteration $i$ the edges of types $sA$ and $Bt$ in all three copies correspond to $(u_i, v_i)$, and edges of type $AB$ in the three copies correspond to three matrices $M^1, M^2, M^3$ such that $M^1 + M^2 + M^3 \sameparity M$.
    
    This suffices for proving the claim, since for each iteration $i = 0, \dots, n$ we have
    \[
        u_i^T M^1 v_i + u_i^T M^2 v_i + u_i^T M^3 v_i = u_i^T M v_i .
    \]
        
    The invariant is proved by induction.
    For the base case $i=0$, all three initial graphs mirror $M$ and $u_0, v_0$ perfectly, and indeed $M + M + M = 3M \sameparity M$. 
    
    For the inductive step, focus on a single iteration of \sol. 
    Edges of types $sA$ are shared by all $3$ copies and undergo the same changes.
    By induction hypothesis, at the beginning of the $i^{th}$ iteration, the existence of any $sA$ edge $e$ corresponds to $u_{i-1}(e)$. By the conditional sampling process, at the end of the iteration its existence $\ind{}(e)$ satisfies:
    \[
\begin{aligned}
        \ind{}(e)
       \sameparity u_{i-1}(e) + z_e
       \sameparity u_{i-1}(e) + \udif(e) \sameparity u_i(e) . 
    \end{aligned}
\]
    The congruence implies equality, since values are Boolean on both ends.
    The same holds for $Bt$ edges and the vectors $v_{i-1}, v_i$.
    
    Edges of type $AB$ are added differently to each copy. 
    For each $j\in[3]$, let $M^j$ and $\ind{}^j(e)$ denote the edges in the $j^{th}$ copy at the beginning and the end of the iteration, respectively.
    By the induction hypothesis we have $M^1 + M^2 + M^3 \sameparity M$.
    For each $AB$ edge $e$ and copy $j$ we have $\ind{}^j(e) \sameparity M^j(e) + \redhist^j(e)$. Altogether:
    \[
        \begin{aligned}
            \ind{}^1(e) + \ind{}^2(e) + \ind{}^3(e)
            &\sameparity \left((M^1(e) + M^2(e) + M^3(e)\right)) + \left(\redhist^1(e) + \redhist^2(e) + \redhist^3(e)\right)\\
            &\sameparity M(e) + 2\left(z_e^1 + z_e^2 + z_e^3\right)\\
            &\sameparity M(e) ,
        \end{aligned}
    \]
    where the second congruence uses the induction hypothesis for the first summand and the reduction construction for the second. 
    Equality is implied as both ends are Boolean values. 
\end{proof}

Next, we analyze a distributional $\oumv$ input, but not yet a uniformly random one. Instead, we focus on a distribution $\muadv$ over pairs $(\udif,\vdif)$ that naturally arises from a sequences of changes. 

Consider a $p$-smoothed sequence of $\poisson{t}$ updates, with the adversarial strategy $\advdist$. Denote this sequence by $\advseq$, and its histogram by $\advhist$. Their distribution is exactly:
\[
\begin{aligned}
    \advseq \sim \left(\advdist^p\right)^{\poisson{\roundnumber}};
    && \advhist(e) \sim \poisson{\advdist^p(e) \cdot \roundnumber} .
\end{aligned}
\]

Define the projection of $\advhist$ to parities for edges of types $sA, Bt$:
\[
\begin{aligned}
    \uadv(e) := \advhist(e) \ \text{  (mod 2)}
    && &&
    \vadv(e) := \advhist(e) \ \text{  (mod 2)}
\end{aligned}
\]
Our desired distribution $\muadv$ is the projection of the random histogram $\advhist$ to these vectors:
\[
    \muadv(u,v) := \Prob{\advhist}{\uadv = u \wedge \vadv = v}
\]

Our next claim is that on a pair sampled from $\muadv$, the resulting histograms from the reduction $\redhist^j$ (for $j\in[3]$) have the same distribution as $\advhist$:

\begin{claim}[reduction output on distributional input]
    \label{claim:distributional_input}
    If the reduction (steps (4)-(7)) is fed with distributional inputs $(\udif, \vdif)\sim \muadv$, then the three output sequences satisfy, for all $j\in[3]$:
    \[
            \redseq^j \samedist \advseq
    \]
    
\end{claim}

\begin{proof}
    As both adversarial and the reduction sequences are randomly order, it suffices to prove the statement for the histograms.
    Fix two Boolean vectors $(u,v)$, and denote by $\redhist^j[u,v]$ the histograms created when $(\udif,\vdif) = (u,v)$.

    By Poissonization, the distribution of $\advhist$ is a product of independent distributions $\advhist(e)$. Thus, for any pair of Boolean vectors $u,v$, the conditioning on $\set{\uadv = u \wedge \vadv = v}$ affects each entry $\advhist(e)$ separately, for $e$ of types $sA, Bt$.
    This is mirrored by steps (4), (5) which determine $\redhist^j[u,v]$ for $j\in[3]$.
    
    For edges of type $AB$, step (6) takes a Poisson number of samples, each is uniform over the $n^2$ edges of type $AB$. Thus, again by Poissonization:  
    \[
        z_e^j \sim \poisson{\advdist(e) \cdot (\roundnumber/2)} ,
    \]
    where we used $\advdist(e) = \frac{1-p}{n(n+2)}$ for $AB$-type edges.
    By additivity of Poisson distributions, for $j\in[3]$ and any $AB$-type edge $e$, the following distributions are equal:
    \[
        \redhist^j(e) \samedist \sum_{j'\in[3]\setminus\set{j}} z_e^{j'} \ \sim \  \poisson{\advdist(e) \cdot \roundnumber} .
    \]
    
    Overall, conditioned on $(\udif,\vdif) = (u,v)$, we get:
\begin{equation}
        \label{eq:conditional_equality}
        \redhist^j [u,v]
        \samedist \advhist \vert \left(\uadv = u \wedge \vadv = v \right) .
\end{equation}

    By the law of total probability, with $(\udif,\vdif)\sim \muadv$ and for $j\in[3]$, we get the following:
    \[
        \redhist^j 
        \samedist \sum_{(u,v)} \muadv(u,v) \cdot \redhist^j[u,v]
        \samedist \sum_{(u,v)} \muadv(u,v) \cdot \advhist \vert \left(\uadv = u \wedge \vadv = v \right)
        \samedist \advhist .\qedhere
    \]
\end{proof}

To relate the last claim to the average-case \oumv, note that instead of drawing a uniformly random pair $(u_i, v_i)$, one can uniformly draw the difference vectors $(\udif, \vdif)$ at each iteration.
We thus relate $n$ independent draws from $\muadv$ (denoted $\muadv^{n}$) to a uniform choice of $n$ pairs of difference vectors.

By a slight abuse of notation, we think of both distributions as over $2n^2$ bits ($2n$ vectors of $n$ bits), and use $U_m$ for the uniform distribution on $m$ bits.
Our closeness measure is statistical distance, which is a normalized $\ell_1$ distance (with factor $1/2$).
\begin{claim}
    \label{claim:distributional_vs_uniform}
    It holds that
    \[
        \sd{\muadv^n}{U_{2n^2}}  \leq 2n^{-3}
    \]
\end{claim}
\begin{proof}
    Focus on a single iteration.
    All edges of type $sA, Bt$ have the same probability under $\advdist^p$, call it $p_1$.
    By Poissonization, each of the $2n$ bits in $\muadv$ is i.i.d according to the parity of $\poisson{p_1 \cdot t}$, which by~\cref{fact:poisson_properties} is simply $\bernoulli{(1+e^{-p_1\cdot t})/2}$.
    Thus, for the $j^{th}$ bit of $\muadv$:
    \[
        \sd{\left(\muadv\right)_j}{U_1}
        \leq e^{-2 p_1 \cdot t}
        \leq e^{-5 \log n}
        = n^{-5},
    \]
    where the inequality is due to $p_1 \geq p/(2n)$ and $t = 5n\log n / p$.

    Finally, due to independence of the $2n$ edges within the same iteration, and independence among the $n$ iterations, we can use subadditivity of statistical distance (\cref{fact:item:SD_subadditivity} of \cref{fact:SD_properties}):
    \[
        \sd{\muadv^n}{U_{2n^2}}
        \leq n \sum_{j=1}^{2n} \sd{\left(\muadv\right)_j}{U_1}
        \leq 2n^2\cdot n^{-5} 
        = 2n^{-3} .\qedhere
    \]
\end{proof}

Combining~\cref{claim:correctness_fixed} and~\cref{claim:distributional_input}, and~\cref{claim:distributional_vs_uniform} proves the correctness of \sol:

\begin{claim}[Correctness]
    \label{claim:correctness_final}
    Assume there exists a data structure for counting \paths{s}{t}{3} which answers all queries correctly with probability at least $0.99$ on a $p$-smoothed sequence of changes,
    then algorithm \sol solves average-case $\oumv$ w.p.\ at least~$0.96$.
\end{claim}

\begin{proof}
    An average-case input for $\oumv$ can be constructed as follows.
    Draw uniformly random $M, u_0, v_0$ and then for $n$ iterations, draw uniformly random difference vectors $\udif, \vdif$.
    By setting $u_i$ as the xor of $u_{i-1}$ and the $i$-th difference vector $\udif$, we get a sequence of uniformly random vectors $u_0,\ldots,u_n$;
    the vectors $v_0,\ldots,v_n$ are defined similarly.
    
    First, we consider a similar input, where the difference vectors $\udif,\vdif$ are drawn from $\muadv$ instead of uniformly (each iteration independently).    
    By~\cref{claim:distributional_input}, the sequences $\redseq^j$ at each iteration are distributed according to the adversarial strategy $\advdist$, 
    and so each copy of the data structure has error probability at most $0.01$.
    By union bound, the probability of any error within all executions is at most $0.03$. 
    
    By~\cref{claim:correctness_fixed}, whenever there are no errors, \sol  correctly outputs $u_i^T M v_i$. Thus, over the distributional inputs it succeeds with probability at least $0.97$.
    
    Lastly, by~\cref{claim:distributional_vs_uniform}, the entire input has statistical distance at most $2n^{-3}$ from the one of the average-case. 
    By~\cref{fact:item:SD_error_difference} of \cref{fact:SD_properties}, the error of an algorithm can only differ by the statistical distance between the input distributions.
    Thus, on an average-case input, \sol succeeds at least with probability $0.97 - 2n^{-3} > 0.96$.
\end{proof}

\begin{claim}[Running time]
    \label{claim:running_time}
    If there exists a data structure for counting \paths{s}{t}{3} with pre-processing time $\tp$, update time $\tu$ and query time $\tq$, then with probability $1-o(1)$ the running time of our algorithm is
    \[
        3\tp(n') + O\left(\frac{n^2 \log n}{p}\right) \cdot \tu(n') + O(n) \cdot \tq(n') + O\left(\frac{n^2 \log n \log(n/p)}{p}\right) .
    \]
\end{claim}

\begin{proof}
    We pre-process three separate executions, and query $3(n+1)$ times overall. 
    The rest of the proof analyzes the number of update operations, and the time for all other computations.

    Focus on one iteration.
    As we successfully simulate a sequence of $\poisson{t}$ changes, by~\cref{fact:poisson_properties}:
    \[
        \Prob{}{\poisson{t} > e t}
        \leq \frac{(e t)^{e t} e^{-t}}{(e t)^{e t}}
        = e^{-t} 
        \leq e^{-n} .
    \]
    That is, there is a negligible probability we exceed $O(t)$ changes (where we can abort). Assume henceforth this is not the case. 
    In particular, the reduction uses a total of $O(t)$ update operations.

    The nitty-gritty is in generating Poisson values. By~\cref{fact:poisson_properties}, one can generate a sample of $\poisson{\lambda}$ with running time $O(m)$ where $m$ is the output value and $\Expc{}{m} = \lambda$. When generating multiple values, we can rely on their sum for running time, but it is important to take into account the number of values generated.\footnote{Indeed, consider generating $k$ samples from $\poisson{1/k}$. Their sum is w.h.p. $O(1)$, but the overall running time is $\Theta(k)$, as we pay additional $\Theta(1)$ for each value.}
    Differently put, the running time is actually $O(m_i+1)$ for each generated value $m_i$, and for $k$ such values the running time is $O(m+k)$ where $m = \sum_{i\in[k]} m_i$.

    In steps (4) and (5) we need conditional samples, which can be obtained by \emph{rejection sampling}.
    This is true since the parity of each such value is almost a fair random coin (as seen in the proof of~\cref{claim:distributional_vs_uniform}). Thus, any desired parity value has probability at least $1/3$, and using $O(\log n)$ repetitions suffices to get a polynomially small error probability (even after union bounding the $2n$ conditional samples).
    Generating $2n$ such values, with sum $O(t)$ (with probability $1-o(1)$) takes~$O\left((t + 2n)\cdot \log n\right)$ time.

    In step (6) we generate all $AB$-type changes. Generating each value separately would result in $n^2$ values (most of which are $0$), which is too costly.
    Instead, we sample the total amount of changes (which is $O(t)$ w.h.p), and assign random edge to each one, using $\log n$ random bits. The overall running time of this step is $O\left(t \log n\right)$. 
    In step (7) we randomly order the sequences, each of length $O(t)$, which takes $O(t \log t)$ steps.

    Overall, for a single iteration we require $O(t \log t)$ steps of computation, on top of $O(t)$ update operations and one query. The numbers for $n$ iterations follow once we plug $t = O\left(\frac{n \log n}{p}\right)$.
\end{proof}

Finally we can deduce the lemma.
\begin{proof}[Proof of~\cref{lem:ac_oumv_to_restricted_graphs}]
    union bounding the error probabilities from \cref{claim:correctness_final} and \cref{claim:running_time}, the existence of an algorithm for counting \paths{s}{t}{3} with running times $\tp, \tu, \tq$ implies an algorithm that correctly solves average-case parity \oumv using \sol, with error probability at most $0.04 + o(1) < 0.05$ and the running time from~\cref{claim:running_time}. 
\end{proof}

\subsection{Reducing $P_3$-partite graphs to general graphs}
\label{sec:lb:restricted_graphs_to_general_graphs}

In this section, we reduce counting \paths{s}{t}{3} in $p$-smoothed $P_3$-partite dynamic graphs to counting all \paths{s}{t}{3} in a general $p'$-smoothed dynamic graph, for some $p' = \Theta(p)$. We prove the following lemma.

\begin{lemma}
\label{lem:restricted_graphs_to_general_graphs}
    If there exists a data structure for counting \paths{s}{t}{3} on (general) $p'$-smoothed dynamic graphs initialized at random 
    that fails with probability at most $1/1600$,
    then there exists a data structure for counting \paths{s}{t}{3} on $p$-smoothed $P_3$-partite dynamic graphs initialized at random, where $p \in [p',2p']$, with the same asymptotic running times and error probability at most~$1/100$.
\end{lemma}

To prove this lemma, we take a $p$-smoothed $P_3$-partite dynamic graph, 
and create 16
$p'$-smoothed unrestricted dynamic graphs,
each with the same set of nodes as the original graph.
Moreover, the 16 dynamic graphs differ only by their initial graphs, while they all consist of the same sequence of changes;
hence, we only introduce one sequence of changes.
Later on, we show that from the number of 
\st $3$-paths in the 16 unrestricted graphs,
one can deduce the number of 
\st $3$-paths in the $P_3$-partite graph.

The reduction works for any value of $p\in[0,1]$, and coincides with the relevant models in the cases of $p=0$ and $p=1$:
when $p=0$, the original sequence is uniform on the edges allowed in  
a $P_3$-partite graph
and the new sequence is uniform on all the  graph edges, as in~\cite{HLS22};
when $p=1$, both sequences are the same, and the reduction shows an adversarial strategy in the general model that focuses only on edges that can appear in a $P_3$-partite graph.
 
\subsubsection{Simulating a single step}
We transform a sequence of changes in the $p$-smoothed restricted model (hereafter: restricted sequence) into a longer sequence of changes in a $p'$-smoothed general model (hereafter: general sequence). Recall that $R_H$ is the set of allowed edges in the restricted model.
In order to construct the general sequence we use a parameter $0\leq\alpha\leq 1$. 
For each change of an edge $e$ in the restricted sequence, w.p.~$\alpha$ we add to the general sequence the same edge $e$, and w.p.~$1-\alpha$ we add to the general sequence a uniformly random edge from the set $\overline{R_H} = \binom{V}{2} \setminus R_H$ of edges not allowed in the restricted sequence. 

Let
\[
\begin{aligned}
	&\alpha = \frac{\card{R_H}}{\card{R_H} + (1-p)\card{\overline{R_H}}}
\end{aligned}
\]
and note that $\alpha \in [\card{R_H} / \binom{\card{V}}{2} , 1]$ since $p\in[0 , 1]$.
Given a flip of an edge $e$ in the restricted sequence, we choose the change in the general sequence as follows.
\begin{enumerate}
    \item Flip a coin $C$ with probability $\alpha$.

    \item If $C = 1$, make an \emph{interior} update, i.e. flip the edge $e$ in the general sequence.
    Note that the choice of $e$ in the restricted model implies that, conditioned on $C=1$, w.p.~$p$ the edge $e$ is chosen by the adversary (from $R_H$), and w.p.~$1-p$ it is chosen uniformly at random from $R_H$.

    \item If $C = 0$, make an \emph{exterior} update: add to the general sequence an edge chosen uniformly at random from~$\overline{R_H}$.
\end{enumerate} 

To see that the update sequences created by this process is indeed a $p'$-smoothed sequence on general graphs, note that the updates of the following types happen with the prescribed probabilities.
\begin{itemize}
    \item An adversarial interior update w.p.~$p' = \alpha \cdot p$;

    \item A random interior update w.p.~$(1-p')\card{R_H}/\binom{\card{V}}{2} = \alpha \cdot (1-p)$;

    \item A random exterior update w.p.~$(1-p')\card{\overline{R_H}}/\binom{\card{V}}{2} = 1 - \alpha$.
\end{itemize}
That is, an adversarial update happens w.p.~$p'$, and otherwise the update is chosen uniformly at random from all of $\binom{V}{2}$. 
This corresponds to a $p'$-smoothed sequence with an adversary that only chooses edges from~$R_H$.

Recall that a $P_3$-partite graph consists of $4$ parts, denoted $V = \set{s} \cup A \cup B \cup \set{t}$, where $s$ and $t$ are single nodes and $\card{A} = \card{B} = n$.
Its set of interior edges is thus
\[
\begin{aligned}
    R_{P_3} = \left(\set{s}\times A\right) \cup \left(A\times B\
\right) \cup \left(B\times \set{t}\right) .
\end{aligned}
\]
Since 
$\card{V}=2n+2$ and 
$\card{R_{P_3}} = n(n+2)$,
we have $\alpha\in[\frac{1}{2}, 1]$ (for any value of $p$), and therefore $p' \in [p/2,p]$.

The complement set $\overline{R_{P_3}}$ consists of edges of exterior types: $sB, At, AA, BB, st$. 
As the edge $(s,t)$ cannot participate in any simple $3$-path from $s$ to $t$, we are interested in the four other edge types.

\subsubsection{The actual reduction}
In order to reduce the problem of counting \st 3-paths in general (smoothed dynamic) graphs into the same problem in $P_3$-partite (smoothed dynamic) graphs
we use the inclusion-edgclusion idea (which appeared in~\cite{HLS22}, and earlier in~\cite{DalirrooyfardLW20,Boix-AdseraBB19}).

We start with a $p$-smoothed sequence of edge flips in the $P_3$-partite graph $\calG$,
with node set
$V = \set{s} \cup A \cup B \cup \set{t}$.
The dynamic graph $\calG$ consists of an initial graph $G_0$ chosen uniformly at random with edges only from $R_{P_3}$, and a sequence $e_1, \dots, e_T$ of $T$ updates in the form of edges from $R_{P_3}$ being flipped.

From this, we define $16$ (highly correlated) dynamic graphs $\calG^{(0)}, \dots, \calG^{(15)}$, each of which being a $p'$-smoothed dynamic graph with the same node set $V$.
The 16 initial graphs are constructed by randomly partitioning each of the $4$ sets of exterior edges ($sB, At, AA, BB$) into two subsets, and allocate these subsets to the $16=2^4$ graphs such that each combination of subsets occurs once.
The 16 dynamic graph are composed of these 16 different initial graphs, 
while their sequences of changes are identical.

Formally, the reduction does the following:

\begin{algorithm}[H]\DontPrintSemicolon
	\caption{Algorithm \ptg for counting paths in a general graph by counting them in $P_3$-partite graphs}
	
	\nl
	Partition $sB$ into $sB= E_{sB}^1 \sqcup E_{sB}^0$, where each edge goes to one of the sets u.a.r.
	
	\nl
	Similarly, partition $At, AA, BB$ into $E_{At}^1, E_{AA}^1, E_{BB}^1$ and their complements $E_{At}^0, E_{AA}^0, E_{BB}^0$
	
	\tcp*{partition each set of external edges into two}
	
	\nl
	Initialize $E'=\emptyset$, and w.p.~$1/2$ assign $E'\gets \{s,t\}$
	
	\nl
	\For{$i\gets 0,\ldots,15$}{
		 Let $b_{i1}, b_{i2}, b_{i3}, b_{i4}$ be the binary representation of $i$
		
		 $E^{(i)}_0\gets E_0 \sqcup E_{sB}^{b_{i1}} \sqcup E_{At}^{b_{i2}} \sqcup E_{AA}^{b_{i3}} \sqcup E_{BB}^{b_{i4}} \sqcup E'$		 
		 \tcp*{define the $16$ initial graphs}
	}
	
	\nl Initialize $c\gets1$, $c'\gets 1$. 
	
	\nl
	\While{$c \leq T$}{	
	\lIf{next operation is query}{query all graphs, compute the output and continue}
	\tcp*{compute output as described in \cref{eq:number of paths in restricted graph} and end the current loop iteration}
		
	Flip a coin $C\sim Bin(\alpha)$ \tcp*{$\alpha$ as above}
		
	\lIf{$C=1$}{$f \gets e_{c}$}
	\lElse{draw $f$ uniformly from $\overline{R_H}$}
		
	\For{$i\gets 0,\ldots,15$}{update $E^{(i)}_{c'} \gets E^{(i)}_{c'-1} \symdif \set{f}$}
		
	Update counters: 
		$c'\gets c'+1$, and \lIf{$C=1$}{$c\gets c+1$}
	}
\end{algorithm}

Let $T'$ be the number of edge flips in the dynamic graph $\calG^{(i)}$ (the last value of $c'$), which is identical for all $i$.
Since $\alpha \geq 1/2$, by a simple Chernoff bound the probability that $T' \geq 3T$ is at most $e^{-\Omega(T)}$.
This completes the construction of the dynamic graphs $\calG^{(0)}, \dots, \calG^{(15)}$.

\subsubsection{Analysis}
We now show that the reduction indeed creates 
$p'$-smoothed dynamic (general) graphs, 
and that by counting \paths{s}{t}{3} on them we can deduce the number of \paths{s}{t}{3}
in the $p$-smoothed $P_3$-partite dynamic graph $\calG$,
thus proving \cref{lem:restricted_graphs_to_general_graphs}.

\begin{claim}
	For each $i\in \set{0,\dots, 15}$, the dynamic graph $\calG^{(i)} = (G^{(i)}_0, \dots, G^{(i)}_{T'})$ is a $p'$-smoothed dynamic graph.
\end{claim}

\begin{proof}
    Fix $i\in \set{0,\dots, 15}$.
    The probability of any node pair $(x,y)$ to be an edge in $G^{(i)}_0$ is exactly $1/2$.
    Indeed, if $(x,y)$ is interior (of types $sA, AB, Bt$) then it is taken from $G_0$, where it is chosen w.p. $1/2$.
    If $(x,y)$ is an exterior edges, then it is either $(s,t)$ which is taken w.p.~$1/2$ (see the choice of $E'$).
    Otherwise, it is of one of four types $sB, At, AA, BB$.
    Say $(x,y)$ is of type $sB$ (similar argument works for $At, AA, BB$), then it has probability exactly $1/2$ to be chosen in $E^{b_{i,1}}_{sB}$ (either if $b_{i,1} = 0$ or $b_{i,1} = 1$). 
    Overall, any $(x,y)\in \binom{V}{2}$ is chosen to $G^i_0$ w.p.~$1/2$.
    By \cref{observation:initial graph},
    this is a $p'$-smoothed initial graph, for an adversary that picks a random initial graph.

    Finally, by our choice of $\alpha$ and $p'$, the edge $f$ is drawn in a $p'$-smoothed manner, as seen before.
\end{proof}

The following definition will be useful.
\begin{definition}[properly partitioning multiple edge set]
	Let $F_1, \dots, F_m$ be 
	$m$ pairwise disjoint edge sets on the same node set $V$, 
	and $F = \bigsqcup_{i=1}^{m} F_i$. 
	
	A set of $2^m$ graphs $G^{(0)}, \dots, G^{(2^m - 1)}$, with $G^{(i)}=(V,E^{(i)})$, 
	is said to \emph{properly partition} the edge sets $F_1, \dots, F_m$ if there exist~$m$ partitions $F_{\ell} = F'_{\ell} \sqcup F''_{\ell}$ such that:
	\begin{itemize}
		\item Each graph chooses a part of each edge set:
		\[
		\forall i\in\set{0,\dots,2^m - 1}, \ell\in[m]: E^{(i)} \cap F_{\ell} \in \set{F'_{\ell}, F''_{\ell}}
		\]
		
		\item Each graph chooses a different combination of parts:
		\[
		\forall i\neq  j\in \set{0,\dots,2^m - 1}, F\in \{F_0,\ldots,F_{15}\}: E^{(i)} \cap F \neq E^{(j)} \cap F
		\]
	\end{itemize}
	
	In the context of $P_3$-partite graphs,  $m=4$, and each $F_i$ corresponds to an exterior edge type (e.g., $F_1 = \set{s} \times B$). 
	If the conditions above hold for a set of $16$ graphs, we say they \emph{properly partition} the exterior edges.  
\end{definition}

The next claim shows that the $16$ graphs are correlated at all times.
\begin{claim}
    \label{lem:always_partitioned}
    For any $c'\in [T']$, the $16$ graphs $G^{(0)}_{c'}, \dots, G^{(15)}_{c'}$ properly partition the exterior edge types $F_1 = \set{s} \times B, F_2 = A \times \set{t}, F_3 = \binom{A}{2}, F_4 = \binom{B}{2}$.
\end{claim}

\begin{proof}
    We prove the claim by induction.
    At time $c'=0$, this holds by definition, with $F'_0$ and $F''_0$ being the randomly chosen sets $E_{sB}^1, E_{sB}^0$, and similarly for all other exterior edge types.
    
    For the step, assume that at time $c'$ the graphs $G^{(0)}_{c'}, \dots, G^{(15)}_{c'}$ properly partition the exterior edge types using $\set{F'_j, F''_j}_{j=1}^{4}$.
    If at time $c'+1$ an interior edge $f$ was updated (that is, $f=e_c$ for some $c\in[T]$), then the same partition trivially works.
    
    If at time $c'+1$ an exterior edge $f$ was updated, then $f\in F_{\ell}$ for some $\ell\in[4]$.
    Recall that $f$ is flipped in \emph{all} $16$ graphs. 
    Take the partition where
    $\left(F'_{\ell}, F''_{\ell}\right) \gets
    \left(F'_{\ell}\text{(old)} \symdif \set{f},F''_{\ell}\text{(old)} \symdif \set{f}\right)$ and all other $(F'_i,F''_i)$ for $i\neq {\ell}$ remain unchanged.
    
    Formally, the first condition holds for the edge set $F_{\ell}$ since
    \[
\begin{aligned}
        \forall i: &&
        E^{(i)}_{c'+1} \cap F_{\ell} 
        = \left( E^{(i)}_{c'} \cap F_{\ell} \right) \symdif \set{f} 
        \in \set{F'_{\ell}\text{(old)} \symdif \set{f},
        F''_{\ell}\text{(old)} \symdif \set{f}} 
        = \set{F'_{\ell}, F''_{\ell}} ,
    \end{aligned}
\]
    where the containment is due to induction hypothesis. 
    Similarly, the second condition holds trivially for each partition $(F'_i,F''_i)$ that was left unchanged, and 
    \[
\begin{aligned}
        \forall i\neq j: &&
        E^{(i)}_{c'+1} \cap F 
        = \left( E^{(i)}_{c'} \cap F_{\ell} \right) \symdif \set{f} 
        \neq \left( E^{(j)}_{c'} \cap F_{\ell} \right) \symdif \set{f} 
        = E^{(j)}_{c'+1} \cap F
    \end{aligned}
\]
    for $(F'_\ell,F''_\ell)$ that was changed.
    The inequality uses the induction hypothesis and the fact that $S \mapsto S\symdif\set{f}$ is a bijection.

    Thus, $G^{(0)}_{c'+1}, \dots, G^{(15)}_{c'+1}$ properly partition $F_1, F_2, F_3, F_4$, with partitions $F_{\ell} = F'_{\ell} \sqcup F''_{\ell}$.
\end{proof}

\begin{claim}
    Fix $c\in [T]$ and $c'\in[T']$ 
    such that in $G^{(i)}_{c'}$ (for all $i$)
    the interior edge change $e_c$ already occurred while the change $e_{c+1}$ did not.
    For each choice of $X, Y \in \set{A,B}$ (where $X$ and $Y$ may be identical),
    denote by $C^i_{sXYt}$ the count of all $s$-$t$ $3$-paths of type $sXYt$ in graph $G^{(i)}_{c'}$, 
    and by $C$ the number of paths of type $sABt$ in the original graph $G_{c}$ (recall it has no \st 3-paths of other types).
    Then
    \[
\begin{aligned}
        & \sum_{i=0}^{15} C^i_{sABt} = 16 \cdot C; &
        & \sum_{i=0}^{15} C^i_{sBAt} = 4\cdot \card{E_c \cap \left(A \times B\right)}\\
        & \sum_{i=0}^{15} C^i_{sAAt} = 4(n-1)\cdot \card{E_c \cap \left(\set{s} \times A\right)}; &
        & \sum_{i=0}^{15} C^i_{sBBt} = 4(n-1)\cdot \card{E_c \cap \left(B \times \set{t}\right)}
    \end{aligned}
\]
\end{claim}

\begin{proof}
    The proof takes into account the different path types.
    \subparagraph{Paths of type $sABt$.}
    Note that the interior edges are the same for all $16$ graphs as well as the original graph. Paths of type $sABt$ consist only of interior edges, and therefore appear in all $16$ graphs as well as the original graph, which proves the first equality.
    
    The other equalities use paths with exterior edges as well. 
    Due to Lemma~\ref{lem:always_partitioned}, the $16$ graphs properly partition the exterior edge with $\set{(F'_{\ell}, F''_{\ell})}_{\ell\in[4]}$. 
    
    \subparagraph{Paths of type $sBAt$.}
    For each choice of $i\in \set{0,\dots,15}, a\in A, b\in B$, denote by $\ind{a,b}^i$ the event that $\set{(s,b), (b,a), (a,t)} \subseteq E^{(i)}_c$. 
    Switching the order of summation we have:
    \[
\begin{aligned}
        \sum_{i=0}^{15} C^i_{sBAt}
        = \sum_{i=0}^{15} \sum_{\substack{a\in A\\b\in B}} \ind{a,b}^i
        = \sum_{\substack{a\in A\\b\in B}} \sum_{i=0}^{15} 
          \ind{a,b}^i .
    \end{aligned}
\]
    Any interior edge $(x,y)\in E_{c}$ in the original graph at time $c$ also exists in all other graphs at time $c'$.
    For each choice where $(x,y)\in E_c$, we add $1$ to the count if and only if $(s,b)$ and $(a,t)$ are both in graph $G^{(i)}_{c'}$.
    Since $(s,b)$ and $(a,t)$ each belong to a specific part (either $F'_{\ell}$ or $F''_{\ell}$ for some $\ell$), there is a unique combination for both, and exactly $4$ graphs of the $16$ are aligned with it. Thus
    \[
\begin{aligned}
        \sum_{\substack{a\in A\\b\in B}} \sum_{i=0}^{15} 
          \ind{a,b}^i
        = \sum_{\substack{a\in A\\b\in B}}
          4\cdot \ind{(a,b)\in E_c}
        = 4\cdot \card{E_c \cap \left(A \times B\right)} ,
    \end{aligned}
\]
    which proves the second equality.

    \paragraph{Paths of type $sAAt$.} For any $i\in\set{0,\dots,15}$ and $a,a'\in A$ $a\neq a'$, denote by $\ind{a,a'}^i$ the event that $\set{(s,a), (a,a'), (a',t)} \subseteq E^{(i)}$. 
    Switching the order of summation we have:
    \[
\begin{aligned}
        \sum_{i=0}^{15} C^i_{sAAt}
        = \sum_{i=0}^{15} \sum_{\substack{a\in A\\a\neq a'\in A}} \ind{a,a'}^i
        = \sum_{\substack{a\in A\\a\neq a'\in A}} \sum_{i=0}^{15} \ind{a,a'}^i .
    \end{aligned}
\]
    Here, the interior edges are $(s,a)$ of type $sA$, which appear in all graphs. 
    Each choice $a\in A$ such that $(s,a)\in E_t$, and any choice $a'\in A\setminus\set{a}$ induce two exterior edges $(a,a')$ and $(a',t)$ that belong to certain parts, causing a unique combination that aligns with exactly $4$ graphs. Thus
    \[
\begin{aligned}
		\sum_{\substack{a\in A\\a\neq a'\in A}} \sum_{i=0}^{15} \ind{a,a'}^i
        = 
        \sum_{\substack{a\in A\\a\neq a'\in A}} 4\cdot \ind{(s,a)\in E_c}        
        = \sum_{a\in A} 4(n-1)\cdot \ind{(s,a)\in E_c} 
        = 4(n-1)\cdot \card{E_c \cap \left(\set{s} \times A\right)} ,
    \end{aligned}
\]
    proving the third equality. The fourth equality is analogous.
\end{proof}

The above claim allows us to easily deduce \cref{lem:restricted_graphs_to_general_graphs}.
By maintaining the 16 $p'$-smoothed general graphs and making queries to the number $C^i$ of \paths st3 in each graph $G_{c'}^{(i)}$ at time $c'$, we can find the total number 
\[
C_\textnormal{all}=\sum_{i=0}^{15} C^i
\]
of \paths st3 in these graphs (with repetitions).
Note that this number is also the sum of the sums above, that is,
\[
C_\textnormal{all} = \sum_{i=0}^{15} 
\left( 
  C^i_{sABt} 
+ C^i_{sBAt} 
+ C^i_{sAAt}
+ C^i_{sBBt} \right).
\]
The addends from the above claim are all easy to compute, except for $C$.
Define 3 counters
\[
\begin{aligned}
	& C_{AB}= 
	\card{E_c \cap \left(A \times B\right)};
	& C_{sA}= 
	\card{E_c \cap \left(\set{s} \times A\right)};\;\;
	& C_{Bt}=
	\card{E_c \cap \left(B \times \set{t}\right)}
\end{aligned}
\]
and update them whenever an edge of $AB$, $sA$ or $Bt$ is added or removed.
The sum of the above equalities thus gives
\[
    \sum_{i=0}^{15} C^i
    =
    16 \cdot C 
    + 
    4\cdot \card{E_c \cap \left(A \times B\right)}
    +
    4(n-1)\cdot \card{E_c \cap \left(\set{s} \times A\right)}
    +
    4(n-1)\cdot \card{E_c \cap \left(B \times \set{t}\right)}
\]
which simplifies to

\begin{align}
\label{eq:number of paths in restricted graph}
\sum_{i=0}^{15} C^i
= 
16 \cdot C 
+ 
4\cdot C_{AB}
+
4(n-1)\cdot C_{sA}
+
4(n-1)\cdot C_{Bt}
\end{align}
and it is trivial to compute the number $C$ of \paths st3 in $G_c$, as desired.

\subsection{Putting it all together}
\label{sec:lb:wrapping_things_up}
Our main theorem is derived by~\Cref{lem:conjecture_to_ac_oumv},~\Cref{lem:ac_oumv_to_restricted_graphs} and~\Cref{lem:restricted_graphs_to_general_graphs}.

We remark that for ease of presentation throughout the proof we used $n$ for dimension of the algebraic objects (matrix $M$ and vectors $u_i, v_i$), and $n' = 2n+2 = \Theta(n)$ for the number of nodes in the graphs. 
The polynomial factors can use either.
For the sake of clarity, in the statement below we revert back to~$n$ being the number of nodes in the input dynamic graph.

\begin{theorem}
\label{thm:lb for counting st3 paths}
    Fix $0 < p \leq 1$. 
    Conditioned on the \omv conjecture, any dynamic graph algorithm for counting \paths{s}{t}{3} with error probability at most $1/1600$ on $p$-smoothed dynamic graphs, with preprocessing, update and query times $\tp(n), \tu(n), \tq(n)$, must satisfy
    \[
\begin{aligned}
        \tp(n) + \frac{n \log n}{p} \cdot \tu(n) + n \cdot \tq(n) + \frac{n^2 \log n \log(n/p)}{p} = \Omega \left(n^{3-\eps}\right),
    \end{aligned}
    \]
    for any $\eps > 0$.
    Specifically, no algorithm for counting \st 3-paths 
    can simultaneously have
    $\tp=O(n^{3-\epsilon})$,
    $\tu=O(pn^{1-\epsilon})$,
    and
    $\tq=O(n^{2-\epsilon})$,  
    for any $\epsilon>0$.
\end{theorem}

\begin{proof}
    For any small $\eps > 0$, we note the inequality holds trivially for $p = O\left(n^{\eps-1}\right)$ (indeed, the rightmost term on the l.h.s.\ exceeds the r.h.s.\ on its own right).
    We may therefore assume that $p = \omega(n^{\eps-1})$ for some arbitrarily small $\eps$.
    
    Assume a dynamic graph algorithm exists with the mentioned running times, then by \cref{lem:restricted_graphs_to_general_graphs} there exists an algorithm that counts $sABt$ paths on a $P_3$-partite $p'$-smoothed dynamic graph. We have $p' = \Theta(p)$ which leaves the same polynomial dependency.
     
    By \cref{lem:ac_oumv_to_restricted_graphs}, the average-case parity \oumv problems can be solved with running time as stated.
    
    Finally, by \cref{lem:conjecture_to_ac_oumv}, this running time must exceed $\Omega\left(n^{3-\eps}\right)$, the r.h.s in the statement of the theorem.

    For the conclusion at the end, note that all logarithmic factors cancel out by the conjecture, as $\log n = o(n^{\eps})$ for any fixed $\eps > 0$.
\end{proof}

\subsection{Lower bounds for other counting problems}
\label{sec:lb_other_small_graphs}
Our lower bound method can be altered in order to attain lower bounds for several other subgraph-counting problems. Our focus naturally lies with problems that present a gap between average-case and worst-case complexities (see~\cite{HLS22}). 
Along with the results in \cref{sec:upper}, we obtain near tight bounds for all these problems (up to an $\eps$ in the exponent that stems from the \omv conjecture). 

\begin{theorem}
\label{thm:lb for counting st4 paths s3 and s4 cycles}
    Fix $0 < p \leq 1$. 
    Conditioned on the \omv conjecture, any dynamic graph algorithm for counting either
\textbf{%
	(a) \paths{s}{t}{4}; (b) $3$-cycles through $s$}; or \textbf{(c) $4$-cycles through $s$},
with error probability at most $1/1600$ on $p$-smoothed dynamic graphs, with preprocessing, update and query times $\tp(n), \tu(n), \tq(n)$,
must satisfy the equation from \cref{thm:lb for counting st3 paths}.

    Specifically, no algorithm for any of the above problems 
    can simultaneously have
    $\tp=O(n^{3-\epsilon})$,
    $\tu=O(pn^{1-\epsilon})$,
    and
    $\tq=O(n^{2-\epsilon})$,  
    for any $\epsilon>0$.
\end{theorem}

\begin{proof}
The proof for $s$-$4$-cycles is by reduction to counting \paths{s}{t}{3}, which is detailed next.
For counting $s$-$3$-cycles and or \paths{s}{t}{4},
the proof goes by alternating the proof of \Cref{thm:lb for counting st3 paths}.
We refrain from presenting the full proofs again (which will require repeating
sections 
\ref{sec:lb:ac_oumv_to_restricted_graphs}, \ref{sec:lb:restricted_graphs_to_general_graphs} and
\ref {sec:lb:wrapping_things_up} with minimal adjustments)
and instead only present the necessary  changes.

\noindent\textbf{Counting $4$-cycles through $s$.}
Here, we get the lower bound directly from \Cref{thm:lb for counting st3 paths},
be showing how to use an algorithm for counting 
$4$-cycles through $s$
in order to count \st 3-paths.

Given a dynamic graph with two designate nodes $s,t$,
create two copies of it, one with the edge $(s,t)$ and one without this edge;
apply all the changes on both graphs throughout the execution, including flipping the edge $(s,t)$ in both if it is flipped in the original dynamic graph.
Assuming an algorithm for counting \st 3-paths, 
apply it to both graphs, 
and whenever there is a query (for the number of \st 3-paths) regarding the original graph,
make queries 
(for the number of $4$-cycles through $s$) 
to the algorithm to both copies, and return the difference.

To see the answer is correct, note that any $4$-cycles through $s$ that does not go through the edge $(s,t)$ appears in both graphs and hence is subtracted, 
while any such $4$-cycle that uses $(s,t)$ appears only on the graph where the edge $(s,t)$ exists, and the number of such paths is the difference returned by the reduction.
Finally, note that there is a trivial one-to-one correspondence between $4$-cycles through $s$  in which $(s,t)$ appears in the new graph that contains $(s,t)$, and 
\st 3-paths in the original graph.
Hence, the reduction returns the correct value.

To conclude, note that the model (and $p$) are exactly the same, 
i.e.\ the simulation of the original $p$-smoothed dynamic graph 
indeed creates two 
$p$-smoothed dynamic graphs with the same number of nodes and the same parameter $p$.

\noindent\textbf{Counting \paths{s}{t}{4}.}
A straightforward attempt to reduce this problem to counting 
\st 3-paths is by adding a single node $t'$ that is connected only to $t$,
and counting \paths{s}{t'}{4}. 
While this reduction would easily go through in the worst-case setting, in the smoothed case it will not: random edges might connect $t'$ to other nodes in the graph, breaking the tight connection between the numbers of \paths{s}{t}{3} and \paths{s}{t'}{4}.

Instead, one can perform the reduction at an earlier stage, and obtain a variant of ~\Cref{lem:ac_oumv_to_restricted_graphs}, where no edges from $t'$ are allowed, except for $(t',t)$.
Formally, this will imply a lower bound for the restricted case of $P_4$-partite graphs. Then one must apply the inclusion-edgeclusion technique, similar to \cref{lem:restricted_graphs_to_general_graphs}, invoking a larger multiplicative error factor due to the added parts in the graph partition (and potential edge types to exclude).

\noindent\textbf{Counting $3$-cycles through $s$.}
Here again one should focus on the two parts and fix each one of them separately. For the parallel of \cref{lem:ac_oumv_to_restricted_graphs}, note the graph only has $3$ parts: $s, A, B$. We wish to count triangles $sABs$, once each. Such a count would represent the \oumv multiplication just as in the case of $3$-paths (as we only use one orientation).

Lastly, the parallel to \cref{lem:restricted_graphs_to_general_graphs} would need to exclude counts of $sAAs$ and $sBBs$. Such counts would simply add up to $\binom{\deg_A(s)}{2} + \binom{\deg_{_B}(s)}{2}$, where $\deg_A(s)$ and $\deg_B(s)$ represent the number of neighbors $s$ has in $A$ and in $B$, respectively. 
The degrees can be easily maintained and counted out by the algorithm.
The error factor is even smaller in this case, as there are less parts in the partition of $V$.
\end{proof}

\newpage
\part*{Appendix}
\appendix
\section{Tools from probability theory}
\label{sec:tools_from_probability}

We use the following versions of Chernoff/Hoeffding bounds:
\begin{fact}
    \label{fact:Chernoff/Hoefdding}
    For independent random variables $X_1, \dots, X_n$ with sum $X = \sum_{i\in[n]} X_i$ and expectation $\Expc{}{X} = \mu$, the following hold:
    \begin{itemize}
        \item If all $X_i$ are supported on $[a,b]$, then for any $t > 0$:
        \[
            \Prob{}{X \geq \mu + t} \leq \exp\left(\frac{-2t^{2}}{n (b-a)^2)}\right).
        \]

        \item If all $X_i$ are Bernoulli variables, then for any $0 < \delta < 1$:
        \[
            \Prob{}{X \leq (1-\delta)\mu} \leq \exp\left(\frac{-\delta^2 \mu }{2}\right).
        \]
    \end{itemize}
    
\end{fact}

We extensively use the Poisson distribution, and the Poissonization technique.
The Poisson distribution is supported on $\N_{\geq 0}$ and parameterized by a single parameter $\lambda$, where the probability of the outcome to be $k\in \N_{\geq 0}$ is defined by:
\[
    \Prob{x \sim \poisson{\lambda}}{x = k} := \frac{\lambda^k \cdot e^{-\lambda}}{k!} .
\]
In our proofs we extensively use the following well-known facts about the Poisson distribution (see, e.g., \cite[Section~8.4]{books:probability}).
\begin{fact}
    \label{fact:poisson_properties}~
    \begin{enumerate}
        \item
        Additivity of Poisson: for two parameters $\lambda_1, \lambda_2$, it holds that:
        \[
            \poisson{\lambda_1 + \lambda_2} \sim \poisson{\lambda_1} + \poisson{\lambda_2} .
        \]
        
        \item \label{fact:item:poisson_parity}
        Parity of Poisson: for any parameter $\lambda > 0$, and random variable $R \sim \poisson{\lambda}$, we have
        \[
            \Prob{}{R \equiv_2 0} = \frac{1 + e^{-2\lambda}}{2} .
        \]

        \item \label{fact:item:poissonization}``Poissonization'': consider a set of $\poisson{t}$ samples from a distribution $\mathcal{D}$ with support $[n]$, and denote by $H(i)$ the number of times element $i$ was sampled.
        Then $H(i)$ distributes according to $\poisson{\mathcal{D}(i)\cdot t}$, and $H(1),\dots, H(n)$ are independent from one another.
        
        \item Concentration: let $z\sim \poisson{\lambda}$. For any $a < \lambda$ and $b > \lambda$, we have:
        \[
\begin{aligned}
            \Prob{}{z < a} \leq \frac{(e\lambda)^a e^{-\lambda}}{a^a}
            && \Prob{}{z > b} \leq \frac{(e\lambda)^b e^{-\lambda}}{b^b} .
        \end{aligned}
\]

        \item Efficient sampling (e.g.,~\cite{books:Knuth}): a value from $\poisson{t}$ can be generated with running time $O(k)$ where $k$ is the eventual output value (recall $\Expc{}{k} = t$).
    \end{enumerate}
\end{fact}

Finally, to compare random variables (or distributions), we use the well-known notion of \emph{statistical distance}, which equals exactly half the $\ell_1$ distance (also known as total variation distance).
We write $P(x)$ for the probability of element $x$ in $P$. For two distributions $P, Q$ over the same domain $\mathcal{X}$ the statistical distance between them is defined by
\[
    \sd{P}{Q} :=
    \frac{1}{2}\sum_{X \in \mathcal{X}} \card{P(x) - Q(x)} .
\]
This is simply the $\ell_1$ distance of the two distributions laid out as probability vectors. We use the following known facts about statistical distance (see, e.g.,~\cite{CT2001}):
\begin{fact}
    \label{fact:SD_properties}~
    \begin{enumerate}
        \item
        \label{fact:item:SD_subadditivity}
        Sub-additivity of statistical distance for product distributions: given distribution $P_i, Q_i$ supported on $\mathcal{X}_i$ for $i=1, \dots, k$, the product distributions satisfy:
        \[
            \sd{P_1 \times \dots \times P_k\ }{\ Q_1 \times \dots \times Q_k}
            \leq \sum_{i=1}^{k} \sd{P_i}{Q_i} .
        \]

        \item
        \label{fact:item:SD_error_difference}
        An equivalent definition of the statistical distance over finite spaces uses the best distinguisher function     $D:\mathcal{X} \to \set{0,1}$. That is
        \[
            \sd{P}{Q} = \max_{D} \card{\Prob{x\sim P}{D(x) = 1} - \Prob{x \sim Q}{D(x) = 1}} .
        \]
        This is in particular useful when $P, Q$ are distributions over inputs to some algorithm $A$: one can define the distinguisher $E$ that outputs $1$ if $A$ erred on the input $X$ (note that $A$ can be a randomized algorithm, making $E$ a randomized distinguisher).
        Thus, when $P$ and $Q$ are statistically close, so is the probabilty that $A$ errs on them. Indeed, by the equality above
        \[
            \card{\Prob{x\sim P}{E(x) = 1} - \Prob{x \sim Q}{E(x) = 1}}
            \leq \sd{P}{Q} .
        \]
    \end{enumerate}
\end{fact}

\section{Massaging the \omv problem}
\label{app: from omv to parity oumv}

Most of our lower bounds are conditioned no the \omv conjecture which can be seen as a variant of the informal claim ``there is no $O(n^{3-\epsilon})$ combinatorial matrix multiplication algorithm''.  See~\cite{HenzingerKNS15} for further discussion of this connection.

\begin{definition}[the \omv problem]
    The \emph{\omv} problem is the following online problem. Fix an integer $n$. 
    The initial inputs are a Boolean $n$-vector $v_0$ and a Boolean $n\times n$ matrix $M$.
	Then, Boolean vectors $v_i$ of dimension $n$ arrive online, for $i=1,\dots, n$.
	An algorithm solves the \omv problem correctly if it outputs the Boolean vector $Mv_i$ before the arrival of $v_{i+1}$, for $i = 0,1,\dots, n-1$, and outputs $M v_n$.
\end{definition}

\begin{definition}[the \oumv problem]
    The \emph{\oumv} problem is the following online problem. Fix an integer $n$. 
    The initial inputs are two Boolean $n$-vectors $u_0, v_0$, and a Boolean $n\times n$ matrix $M$.
	Then, pairs of Boolean vectors $(u_i, v_i)$ of dimension $n$ arrive online, for $i=1,\dots, n$.
	An algorithm solves the \oumv problem correctly if it outputs the Boolean value $u_i^T M v_i$ before the arrival of $(u_{i+1},v_{i+1})$, for $i = 0,1,\dots, n-1$, and outputs $u_n^T M v_n$.
\end{definition}

\omvconj*

This conjecture is a standard starting point of many lower bounds for dynamic graph algorithms since its introduction.
It is also common to reduce this conjecture to other conjectures that are easier to work with, most notably the \oumv conjecture.
We follow the footsteps of~\cite{HLS22} and make further reductions, reaching the parity \oumv problem:
\DefParityOUMV*

The following result is established in~\cite{HLS22}.

\lemavgparoumv*

We overview the proof here for completeness. It consists of two reductions.

\paragraph{Reduction \#1: existence to parity.}
The first reduction works on worst-case inputs, and replaces the matrix-vector Boolean multiplication (using operators $(\vee, \wedge)$) by a multiplication over the field $\F_2$ (using $(+, \cdot)$),
i.e., it asks to output $uMv \bmod 2$.
In graph notations, we replace the question ``is there a path'' with the question ``is the number of paths even or odd''?

The reduction is done by replacing some $1$ entries in $M$ with $0$. Such replacement occurs at each entry independently with probability $1/2$, and the outcome of this process is a random matrix $M'$.
The core of the proof is in the following insights.
\begin{itemize}
	\item If $u^T M v = 0$ over $\Z$, then $u^T M' v = 0$ as well
	
	\item If $u^T M v = \ell > 0$, then $u^T M' v$ is even or odd with equal probability.
\end{itemize}

To see the second item, consider an entry $M_{ij} = 1$ that affects the value $u^T M v$ (that is, $u_i = v_j = 1$).
First draw all other entries of $M'$, to obtain temporarily $M''$ with $M''_{ij} = 1$ and integer value $a = u^T M'' v$.
Still, $M'_{ij} \sim Ber(1/2)$, and thus the value of $u^T M' v$ will be $a$ or $a-1$ with equal probability.
Overall, we get that the parity of $u^T M' v$ is even or odd with equal probability.

Formally, the reduction gets $u, M, v$ and runs the process above independently for $\Theta(\left(\log n\right)$ times, using a solver for parity \oumv for each resulting matrix $M'$.
If at least once the answer was $1$ (`odd'), it returns $1$ for the existence question. Otherwise it returns $0$.

Note that in fact, matrix multiplication is usually studied over a field and not using Boolean operations. 
It is thus only natural to conjecture a variant of the $\omv$ conjecture over $\F_2$.
A reduction of this conjecture to \oumv would give exactly the parity \oumv conjecture discussed here, without using randomization.
We nevertheless stick to the standard \omv conjecture.

\paragraph{Reduction \#2: worst-case parity to average-case parity.}
For the question of existence (using operators $(\vee, \wedge)$), the average-case is much easier than the worst case: the answer is $1$ w.h.p.
Interestingly, this is not at all the case for the parity question (that is, using $(+, \cdot)$ over the field $\F_2$).
We show that the parity version of \oumv discussed above is hard also on average, establishing \cref{lem:conjecture_to_ac_oumv}.

The proof follows an elegant idea (that dates back to~\cite{BLR93}): given a Boolean matrix $M$, draw a unifomly random Boolean matrix $M^1$, and define another matrix $M^2$ by $M^2 = M + M^1$ (over $\F_2$); note that $M^2$ is also a random matrix. 
Thus we can write $M = M^1 + M^2$ where both $M^1$ and $M^2$ are uniformly random.
Similarly, split each vector $u_i$ (or $v_i$) as $u_i^1 + u_i^2$ where both are uniformly random.
By distributivity of matrix multiplication over $\F_2$, we can write:
    \[
\begin{aligned}
	u_i ^T M v_i = \sum_{b_1, b_2, b_3\in \set{0,1}} (u_i^{b_1})^T M^{b_2} v_i^{b_3} .
\end{aligned}
\]
Each of the eight addends on the right is a multiplication of uniformly random Boolean matrix and vectors.

Formally, assume we have an algorithm solving parity \oumv in the average-case, with error $\eps$.
An algorithm for the worst-case gets an arbitrary input, and does the following: split the matrix (and vectors) randomly as in the above equation, and apply the average-case algorithm on all $8$ parts. 
It can only fail if (at least) one of the $8$ average-case executions failed, which by union bound occurs with probability at most $8\eps$.

 \bibliographystyle{plain}
 \bibliography{refs}

\end{document}